\documentclass[a4paper]{article}

\usepackage[utf8]{inputenc}
\usepackage[T1]{fontenc}
\usepackage{mathscinet}
\usepackage{comment}
\usepackage{graphicx}
\usepackage{booktabs}
\usepackage{hyperref}
\usepackage{pgfkeys}
\usepackage{enumerate}
\usepackage{verbatim}
\usepackage{nicefrac}
\usepackage{framed}
\usepackage{paralist}
\usepackage{subcaption}
\usepackage[inline]{enumitem}
\usepackage{tikz}
\usetikzlibrary{fit}
\usetikzlibrary{intersections}
\usetikzlibrary{decorations.pathreplacing}
\usetikzlibrary{calc}

\usepackage{amsmath}

\usepackage{amssymb,amsthm}
\usepackage{thmtools}
\usepackage{thm-restate}

\usepackage{algpseudocode,algorithm,algorithmicx}
\newcommand*\Let[2]{\State #1 $\gets$ #2}
\algrenewcommand\algorithmicrequire{\textbf{Precondition:}}
\algrenewcommand\algorithmicensure{\textbf{Postcondition:}}

\usepackage{xspace}
\usepackage[color=green!20]{todonotes}
\usepackage{footmisc}

\newtheorem{theorem}{Theorem}
\newtheorem{lemma}{Lemma}
\newtheorem{corollary}{Corollary}

\newtheorem{claim}{Claim}

\newtheorem{problem}{Problem}

\def\wcli{\textsc{Maximum Weighted Clique}\xspace}
\def\cli{\textsc{Maximum Clique}\xspace}
\def\mis{\textsc{Maximum Independent Set}\xspace}

\tikzstyle{vp}=[circle,fill,inner sep=0pt, minimum size=0.1cm]
\tikzstyle{vps}=[circle,fill,inner sep=0pt, minimum size=0.065cm]

\def\centerarc[#1](#2)(#3:#4:#5)
    { \draw[#1] ($(#2)+({#5*cos(#3)},{#5*sin(#3)})$) arc (#3:#4:#5); }
    

\def\seg{\text{seg}}

\newcommand{\needle}[1]{\nneed(#1)}
\newcommand{\fork}[3]{[#1 \leftarrow #2 \rightarrow #3]}
\newcommand{\forkb}[3]{[#1 \rightarrow #2 \leftarrow #3]}
\newcommand{\dir}[1]{\ddir(#1)}
\DeclareMathOperator{\vcdim}{VCdim}
\DeclareMathOperator{\iocp}{iocp}
\DeclareMathOperator{\nneed}{Needle}
\DeclareMathOperator{\ddir}{dir}

\DeclareMathOperator{\ocp}{ocp}

\newcommand{\cl}{\mathcal X}

\title{EPTAS and Subexponential Algorithm for Maximum Clique on Disk and Unit Ball Graphs%
\thanks{Research partially supported by EPSRC grant FptGeom (EP/N029143/1), ANR grant ESIGMA (ANR-17-CE23-0010), the LABEX MILYON (ANR-10- LABX-0070) of Université de Lyon, within the program "Investissements d'Avenir" (ANR-11-IDEX-0007) operated by the French National Research Agency (ANR), the ANR Project DISTANCIA (ANR-17-CE40-0015) operated by the French National Research Agency (ANR), and by the European Research Council (ERC) under the European Union's Horizon 2020 research and innovation programme Grant Agreement 714704.}
\thanks{Extended abstracts of this paper appeared in~\cite{bonnetetal18} and~\cite{Bonamy18}.}
}

\date{}

\author{
Marthe Bonamy
\thanks{CNRS, LaBRI, Université de Bordeaux, Bordeaux, France.
E-mail: marthe.bonamy@u-bordeaux.fr}
\and
Édouard Bonnet
\thanks{Université de Lyon (COMUE), CNRS, ENS de Lyon, Université Claude-Bernard Lyon 1, LIP, Lyon, France.
E-mail: edouard.bonnet@dauphine.fr}
\and
\and
Nicolas Bousquet
\thanks{CNRS, G-SCOP laboratory, Grenoble-INP, France.
E-mail: nicolas.bousquet@grenoble-inp.fr}
\and
Pierre Charbit
\thanks{Université Paris Diderot - IRIF, Paris, France
and 
Université de Lyon (COMUE), CNRS, ENS de Lyon, Université Claude-Bernard Lyon 1, LIP.
E-mail: charbit@irif.fr}
\and
Panos Giannopoulos
\thanks{giCentre, Department of Computer Science, City University of London, London, United Kingdom.
E-mail: panos.giannopoulos@city.ac.uk}
\and
Eun Jung Kim
\thanks{Universit\'{e} Paris-Dauphine, PSL University, CNRS UMR, LAMSADE, Paris, France.
E-mail: eun-jung.kim@dauphine.fr}
\and
Pawe\l{} Rz\k{a}\.zewski
\thanks{Faculty of Mathematics and Information Science, Warsaw University of Technology, Warsaw, Poland
and
Institute of Informatics, University of Warsaw, Warsaw, Poland.
E-mail: p.rzazewski@mini.pw.edu.pl}
\and
Florian Sikora
\thanks{Universit\'{e} Paris-Dauphine, PSL University, CNRS UMR, LAMSADE, Paris, France.
E-mail: florian.sikora@dauphine.fr}
\and
Stéphan Thomassé
\thanks{Université de Lyon (COMUE), CNRS, ENS de Lyon, Université Claude-Bernard Lyon 1, LIP, Lyon, France
and
Institut Universitaire de France.
E-mail: stephan.thomasse@ens-lyon.fr}
}

\begin{document}

\maketitle
\newpage
\begin{abstract}
  A (unit) disk graph is the intersection graph of closed (unit) disks in the plane.
  Almost three decades ago, an elegant polynomial-time algorithm was found for \textsc{Maximum Clique} on unit disk graphs [Clark, Colbourn, Johnson; Discrete Mathematics '90].
  Since then, it has been an intriguing open question whether or not tractability can be extended to general disk graphs.
  We show that the disjoint union of two odd cycles is never the complement of a disk graph nor of a unit (3-dimensional) ball graph.
  From that fact and existing results, we derive a simple QPTAS and a subexponential algorithm running in time $2^{\tilde{O}(n^{2/3})}$ for \textsc{Maximum Clique} on disk and unit ball graphs.
  We then obtain a randomized EPTAS for computing the independence number on graphs having no disjoint union of two odd cycles as an induced subgraph, bounded VC-dimension, and linear independence number.
  This, in combination with our structural results, yields a randomized EPTAS for \textsc{Max Clique} on disk and unit ball graphs.
  \textsc{Max Clique} on unit ball graphs is equivalent to finding, given a collection of points in $\mathbb R^3$, a maximum subset of points with diameter at most some fixed value.
  
  In stark contrast, \textsc{Maximum Clique} on ball graphs and unit $4$-dimensional ball graphs, as well as intersection graphs of filled ellipses (even close to unit disks) or filled triangles is unlikely to have such algorithms.
  Indeed, we show that, for all those problems, there is a constant ratio of approximation which cannot be attained even in time $2^{n^{1-\varepsilon}}$, unless the Exponential Time Hypothesis fails.
\end{abstract}

\section{Introduction}\label{sec:intro}
In an \emph{intersection graph}, the vertices are geometric objects with an edge between any pair of intersecting objects. 
Intersection graphs have been studied for many different families of objects due to their practical applications and their rich structural properties~\cite{McKee1999, Brandstadt1999}.
Among the most studied ones are \emph{disk graphs}, which are intersection graphs of closed disks in the plane, and their special case, \emph{unit disk graphs}, where all the radii are equal.
Their applications range from sensor networks to map labeling~\cite{DBLP:conf/waoa/Fishkin03}, and many standard optimization problems have been studied on disk graphs, see for example~\cite{EJvL2009} and references therein.
Most of the hard optimization and decision problems remain NP-hard on disk graphs and even unit disk graphs.
For instance, disk graphs contain planar graphs \cite{koebe} on which several of those problems are intractable.
However, shifting techniques and separator theorems may often lead to subexponential classical or parameterized algorithms \cite{Alber04,Marx15,SmithWormald,Biro17}.
Many approximation algorithms have been designed specifically on (unit) disk graphs, or more generally on geometric intersection graphs, see for instance \cite{Chan03,Nieberg04,Nieberg05,Erlebach05,AgarwalM06,Leeuwen06,Gibson10,ChanH12} to cite only a few.
Besides ad hoc techniques, local search and VC-dimension play an important role in the approximability of problems on (unit) disk graphs.
For the main packing and covering problems (\mis, \textsc{Min Vertex Cover}, \textsc{Minimum Dominating Set}, \textsc{Minimum Hitting Set}, and their weighted variants) at least a PTAS is known.

However, all the mentioned techniques are only amenable to packing and covering problems.  
The \cli problem is arguably the most prominent problem which does not fall into those categories.
For example, anything along the lines of exploiting a small separator cannot work for \cli, where the densest instances are the hardest.
Therefore, it seems that new ideas are necessary to get improved approximate or exact algorithms for this problem.
In this paper, we focus on solving \cli on (unit) disk graphs in dimension 2 or higher.

\paragraph*{Previous results.}

In 1990, Clark \emph{et al.}~\cite{Clark90} gave an elegant polynomial-time algorithm for \cli on unit disk graphs when the input is a geometric representation of the graph.
It goes as follows: guess in quadratic time the two more distant centers of disks in a maximum clique (at distance at most $2$), remove all the centers that would contradict this maximality, observe that the resulting graph is co-bipartite.
Hence, one can find an optimum solution in polynomial time by looking for a maximum independent set in the complement graph, which is bipartite. 
However, recognizing unit disk graphs is NP-hard \cite{Breu98}, and even $\exists \mathbb{R}$-complete~\cite{Kang12}.
In particular, if the input is the mere unit disk graph, one cannot expect to efficiently compute a geometric representation in order to run the previous algorithm.
Raghavan and Spinrad showed how to overcome this issue and suggested a polynomial-time algorithm which does not require the geometric representation~\cite{Raghavan03}. 
Their algorithm is a subtle \emph{blind} reinterpretation of the algorithm by Clark \emph{et al.}
It solves \cli on a superclass of the unit disk graphs or correctly claims that the input is not a unit disk graph.
Hence, it cannot be used to efficiently recognize unit disk graphs.

The complexity of \cli on general disk graphs is a notorious open question in computational geometry.
On the one hand, no polynomial-time algorithm is known, even when the geometric representation is given.
On the other hand, the NP-hardness of the problem has not been established, even when only the graph is given as input.

The piercing number of a collection of geometric objects is the minimum number of points that hit all the objects.
It is known since the fifties (although the first published records of that result came later in the eighties) that the piercing number of pairwise intersecting disks is 4 \cite{stacho,danzer}, meaning that 4 points are always sufficient and sometimes necessary.
An account of this story can be found in a recent paper by Har-Peled \emph{et al.} \cite{Harpeled18}.
The same paper gives a linear-time algorithm to find 5 points piercing any collection of pairwise intersecting disks, as well as a configuration with 13 pairwise intersecting disks for which checking that 4 points are necessary to hit all the disks is quite simple.
Carmi et al. \cite{Carmi18} improved the former result by presenting a linear-time algorithm that finds a hitting set of at most 4 points.

Amb\"uhl and Wagner observed that this yields a $2$-approximation for \cli~\cite{Ambuhl05}.
Indeed, after guessing in polynomial time four points hitting a maximum clique and removing every disk not hit by those points, the instance is partitioned into four cliques; or equivalently, two co-bipartite graphs.
One can then solve optimally each instance formed by one co-bipartite graph and return the larger solution of the two. 
This cannot give a solution more than twice smaller than the optimum.
Since then, the problem has proved to be elusive with no new positive or negative results.
The question on the complexity and further approximability of \cli on general disk graphs is considered as folklore~\cite{bang2006}, but was also explicitly mentioned as an open problem by Fishkin~\cite{DBLP:conf/waoa/Fishkin03}, Amb\"uhl and Wagner~\cite{Ambuhl05}.
Cabello even asked if there is a 1.99-approximation for disk graphs with two sizes of radii~\cite{CabelloOpen,Cabello2015}.

\paragraph*{Our results.}

We fully characterize the disk graphs whose complements have maximum degree~2, i.e., are unions of cycles and paths:
A complement of a disjoint union of paths and cycles is a disk graph if and only if the number of odd cycles in this union is at most one.
For algorithmic purposes, the interesting part of that statement is:
\begin{restatable}{theorem}{forbiddenDisk}\label{thm:main-structural-non-disk}
  A complement of a disk graph cannot have a disjoint union of two odd cycles as an induced subgraph.
  In other words, if $G$ is a disk graph, then $\iocp(\overline G) \leqslant 1$.
\end{restatable}
In the previous statement $\iocp$ denotes the \emph{induced odd cycle packing number} of a graph, i.e., the maximum number of odd cycles as a disjoint union in an induced subgraph.

We show the same forbidden induced subgraphs for unit ball graphs as for disk graphs.
The proofs for disk graphs and unit ball graphs are quite different and the classes are incomparable.
\begin{restatable}{theorem}{forbidden}\label{thm:no-two-odd-cycles-ubg}
  A complement of a unit ball graph cannot have a disjoint union of two odd cycles as an induced subgraph.
  In other words, if $G$ is a unit ball graph, then $\iocp(\overline G) \leqslant 1$.
\end{restatable}

We then present a randomized EPTAS (Efficient Polynomial-Time Approximation Scheme, that is, a PTAS in time $f(\varepsilon)n^{O(1)}$) for \mis on graphs of $\cl(d,\beta,1)$.
The class $\cl(d,\beta,1)$ denotes the class of graphs whose neighborhood hypergraph has VC-dimension at most $d$, independence number at least $\beta n$, and no disjoint union of two odd cycles as an induced subgraph (for formal definitions see Section~\ref{sec:prelim}).

\begin{restatable}{theorem}{eptas}\label{thm:eptas}
  For any constants $d \in \mathbb N$, $0 < \beta \leqslant 1$, for every $0 < \varepsilon < 1$, there is a randomized $(1-\varepsilon)$-approximation algorithm running in time $2^{\tilde{O}({1/\varepsilon}^3)}n^{O(1)}$, and a deterministic PTAS running in time $n^{\tilde{O}({1/\varepsilon}^3)}$ for \mis on graphs of $\cl(d,\beta,1)$ with $n$ vertices.
\end{restatable}

It is then easy to reduce \cli on disk graphs to \mis on $\cl(4,\beta,1)$ for some constant $\beta$.
We therefore obtain, due to Theorem~\ref{thm:eptas} and Theorem~\ref{thm:main-structural-non-disk}, a randomized EPTAS (and a deterministic PTAS) for \cli on disk graphs, settling almost\footnote{The NP-hardness, ruling out a 1-approximation, is still to be shown.} completely the approximability of this problem.
\begin{restatable}{theorem}{disk}\label{thm:disks-eptas}
  There is a randomized EPTAS for \cli on disk graphs, even if a geometric representation is not given.
  Its running time is $2^{\tilde{O}(1/\varepsilon^3)}n^{O(1)}$ for a $(1-\varepsilon)$-approximation on a graph with $n$ vertices.
\end{restatable}

Similarly, Theorem~\ref{thm:eptas} and Theorem~\ref{thm:no-two-odd-cycles-ubg} naturally lead to:
\begin{restatable}{theorem}{ball}\label{thm:ubg}
  There is a randomized EPTAS in time $2^{\tilde{O}({1/\varepsilon^3})}n^{O(1)}$ for \cli on unit ball graphs, even if a geometric representation is not given.
\end{restatable}
Before that result, the best approximation factor was 2.553, due to Afshani and Chan~\cite{Afshani05}.
In particular, even getting a 2-approximation algorithm (as for disk graphs) was open.

Finally we show that such an approximation scheme, even in subexponential time, is unlikely for ball graphs (i.e., intersection graphs of 3-dimensional balls with arbitrary radii), unit 4-dimensional ball graphs, as well as intersection graphs of (filled) triangles and (filled) ellipses.
Our lower bounds also imply NP-hardness.
To the best of our knowledge, the NP-hardness of \cli on unit $d$-dimensional ball graphs was only known when $d$ is superconstant ($d=\Omega(\log n)$)~\cite{Afshani08}.

In the following paragraphs, we sketch the principal lines of our three main contributions.

\noindent\paragraph*{The complement of the union of two odd cycles is not a disk graph}
Let $H$ be the complement of two cycles $C_1 \uplus C_2$.
We focus on the possible configurations for the four centers associated to two non-edges of $H$: one edge of $C_1$ and one edge of $C_2$.
It can be shown that, if those four centers are in convex position, then the non-edges of $H$ should \emph{cross}.
In other words, there is a forbidden configuration: four centers in convex position with the two non-edges being sides of the quadrangle.
For each edge $e$ of $C_1$ (resp. $C_2$), we define three quantities counting how many edges of $C_2$ (resp. $C_1$) are \emph{crossed} by $e$ (for three different meanings of \emph{crossing}).
The forbidden configuration and parity arguments on sums of those quantities bring the desired result that the two cycles cannot be both of odd length.

\noindent\paragraph*{The complement of the union of two odd cycles is not a unit ball graph}
Given a needle in $\mathbb R^3$ whose middle-point is attached to the origin, one can apply a continuous motion in order to turn it around (a motion à la Kakeya, henceforth \emph{Kakeya motion}).
A Kakeya motion can be seen as a closed antipodal curve on the 2-sphere.
If we now consider two needles, each with a Kakeya motion, then the two needles have to go through a same position. 
This simply follows from the fact that two antipodal curves on the 2-sphere intersect.
The second main result of this paper is a translation of this Jordan-type theorem in terms of intersection graphs:
The complement of a unit ball graph does not contain the disjoint union of two odd cycles.
The proof can really be seen as two Kakeya motions, each one along the two odd cycles, leading to a contradiction when the needles achieve parallel directions.

Together with the next result, it implies a randomized EPTAS for the following problem: Given a set $S$ of points in $\mathbb R^3$, find a largest subset of $S$ of diameter at most $1$.

\noindent\paragraph*{EPTAS for \mis on $\cl(d,\beta,1)$}
We show that if a graph $G$ satisfies that every two odd cycles are joined by an edge, the Vapnik-Chervonenkis dimension (VC-dimension for short) of the hypergraph of the neighborhoods of $G$ is bounded, and $\alpha(G)$ is at least a constant fraction of $|V(G)|$, then $\alpha(G)$ can be approximated in polynomial time at any given precision.
More precisely, we present in that case a randomized EPTAS running in time $2^{\tilde{O}({1/\varepsilon^3})}n^{O(1)}$ and a deterministic PTAS.

Our algorithm works as follows.
We start by sampling a small subset of vertices.
Hoping that this small subset is entirely contained in a fixed optimum solution $I$, we include the selected vertices to our solution and remove their neighborhood from the graph.
Due to the classic result of Haussler and Welzl \cite{HausslerW86} on $\varepsilon$-nets of size $O(d/\varepsilon \log{1/\varepsilon})$ (where $d$ is the VC-dimension), this sampling lowers the degree in $I$ of the remaining vertices.
We compute a shortest odd cycle.
If this cycle is short, we can remove its neighborhood from the graph and solve optimally the problem in the resulting graph, which is bipartite by assumption.
If this cycle is long, we can efficiently find a small odd-cycle transversal.
This is shown by a careful analysis on the successive neighborhoods of the cycle, and the recurrent fact that this cycle is a shortest among the ones of odd length.

\noindent\paragraph*{Organization}
The rest of the paper is organized as follows.
In Section~\ref{sec:prelim}, we recall some relevant notations for graphs and elementary geometry, the definitions of VC-dimension, disk graphs, and approximation schemes.
We finish this section by introducing a class $\cl(d,\beta,i)$ of graphs parameterized by three constants: $d$, upper-bounding the VC-dimension, $\beta$, lower-bounding the ratio $\alpha(G)/|V(G)|$ (\emph{independence number divided by number of vertices}), and $i$, upper-bounding the maximum number of odd cycles that can be found as a disjoint union in an induced subgraph.
In Section~\ref{sec:structural}, we characterize disk graphs whose complements have maximum degree 2.
The most useful result of this section is that the complement of the disjoint union of two odd cycles is not a disk graph.
In Section~\ref{sec:unit-ball-obstruction}, we show that unit ball graphs have the same forbidden induced subgraphs.
In Section~\ref{sec:algorithms}, we show that those forbidden induced subgraphs, together with existing results and algorithms, readily give a subexponential algorithm and a QPTAS for \cli on disk and unit ball graphs.
In Section~\ref{sec:eptas}, we design a randomized EPTAS for \mis on the class $\cl(d,\beta,i)$.
We then show how this yields a randomized EPTAS for \cli on disk and unit ball graphs.
This is tight in two directions: having different values of radii, and the dimension of the ambient space.
Indeed, in Section~\ref{sec:gen&lim}, we complement those positive results by showing that \cli is unlikely to have a QPTAS (even a SUBEXPAS) on ball graphs where all the radii are arbitrarily close to 1, and on 4-dimensional unit ball graphs.
We also show the same lower bounds for intersection graphs of simple objects in the plane, namely filled triangles and filled ellipses.
In Section~\ref{sec:perspectives}, we make some observations about the EPTAS and propose some lines of thoughts on how to tackle the computational complexity of \cli in disk and unit ball graphs.

\section{Preliminaries}\label{sec:prelim}

For two integers $i \leqslant j$, we denote by $[i,j]$ the set of integers $\{i,i+1,\ldots, j-1, j\}$.
For a positive integer $i$, we denote by $[i]$ the set of integers $[1,i]$.

The \emph{Exponential Time Hypothesis} (ETH) is a conjecture by Impagliazzo et al. asserting that there is no $2^{o(n)}$-time algorithm for \textsc{3-SAT} on instances with $n$ variables \cite{ImpagliazzoETH}. 
The ETH, together with the sparsification lemma \cite{ImpagliazzoETH}, even implies that there is no $2^{o(n+m)}$-time algorithm solving \textsc{3-SAT}.

\paragraph*{Graph notations}
Let $G$ be a simple graph.
We denote by $\overline{G}$ its complement, i.e., the graph obtained by making every non-edge an edge and vice versa.
$V(G)$ and $E(G)$ represent its set of vertices and its set of edges, respectively.
We denote by $\alpha(G)$ the \emph{independence number} of $G$, i.e., the size of a maximum independent set (or stable set), and by $\omega(G)$ the \emph{clique number} of $G$, i.e., the size of a maximum clique.
For $S \subseteq V(G)$, its open neighborhood, denoted by $N_G(S)$, is the set of vertices that are not in $S$ and have a neighbor in $S$, and its closed neighborhood is defined by $N_{G}[S]=S\cup N_{G}(S)$. We omit the subscript $G$ if the graph is obvious from the context and we write $N_{G}(x)$ instead of $N_{G}(\{x\})$.

The \emph{odd cycle packing number} of $G$, denoted by $\ocp(G)$, is defined as the maximum number of vertex-disjoint odd cycles and the \emph{induced odd cycle packing number} of $G$, denoted by $\iocp(G)$, is the maximum number of vertex-disjoint odd cycles with no edge between any two of them.

The \emph{2-subdivision} of a graph $G$ is the graph $H$ obtained by subdividing each edge of $G$ exactly twice.
If $G$ has $n$ vertices and $m$ edges, then $H$ has $n+2m$ vertices and $3m$ edges.
The \emph{co-2-subdivision} of $G$ is the complement of $H$.
Hence it has $n+2m$ vertices and ${n+2m \choose 2} - 3m$ edges.
The \emph{co-degree} of a graph is the maximum degree of its complement.
A \emph{co-disk} is a graph that is the complement of a disk graph.

\paragraph*{VC-dimension}
VC-dimension has been introduced by Vapnik and Chervonenkis in the seminal paper~\cite{vapnik}.
Let $H=(V,E)$ be a hypergraph.
A set $X \subseteq V$ of vertices of $H$ is \emph{shattered} if for every subset $Y$ of $X$ there exists a hyperedge $e \in E$ such that $e \cap X = Y$.
An intersection between $X$ and a hyperedge $e$ of $E$ is called a \emph{trace} (on $X$).
Equivalently, a set $X$ is shattered if all its $2^{|X|}$ traces exist.
The \emph{VC-dimension} of a hypergraph is the maximum size of a shattered set.
As an abuse of language, we call VC-dimension of a graph $G$, denoted by $\vcdim(G)$, the VC-dimension of the neighborhood hypergraph $(V(G),\{N_G(v) $ $|$ $v \in V(G)\})$.
The \emph{geometric} VC-dimension of a collection of objects is the VC-dimension of the (uncountable infinite) hypergraph where vertices are all the points of the ambient space and the hyperedges are the objects.
If, in general, the geometric VC-dimension does not necessarily coincide with the VC-dimension of the intersection graph, it does coincide for unit balls in any dimension.
Indeed, shattering a set $S$ of points with balls of radius 1 is equivalent to shattering the balls of radius $1/2$ centered at $S$ with balls of radius $1/2$.

\paragraph*{Geometric notations}
For a positive integer $d$, we denote by $\mathbb R^d$ the $d$-dimensional euclidean space.
We denote by $d(x,y)$ the euclidean distance between $x$ and $y$.
If $x$ and $y$ are two points of $\mathbb R^d$, $\seg(x,y)$, or simply $xy$, is the straight-line segment whose endpoints are $x$ and $y$.
If $x$ and $y$ are distinct, $\ell(x,y)$ is the unique line going through $x$ and $y$.
If $s$ is a segment with positive length, then we denote by $\ell(s)$ the unique line containing $s$.

We will often define disks and elliptical disks by their boundary, i.e., circles and ellipses, and also use the following basic facts.
There are exactly two circles that pass through a given point with a given tangent at this point and a given radius; one if we further specify on which side of the tangent the circle is.
There is exactly one circle which passes through two points with a given tangent at one of the two points, provided the other point is \emph{not} on this tangent.
Finally, there exists one (not necessarily unique) ellipse which passes through two given points with two given tangents at those points. 

A \emph{$d$-dimensional closed ball} is defined from a center $x \in \mathbb R^d$ and a radius $r \in \mathbb R^+$, as the set of points $\{y \in \mathbb R^d $ $|$ $d(x,y) \leqslant r\}$, i.e., at distance at most $r$ from $x$.
The \emph{diameter} of a subset $S \subseteq \mathbb R^d$ is defined as $\underset{x,y \in S}{\sup}d(x,y)$.
The \emph{piercing number} (also called \emph{hitting set} or \emph{transversal}) of a collection $\mathcal O$ of geometric objects in $\mathbb R^d$ is the minimum number of points of $\mathbb R^d$ that \emph{pierce} (or \emph{hit}) all the objects of $\mathcal O$, i.e., each object contains at least one of these points.

\paragraph*{Disk graphs and their forbidden induced subgraphs}

A \emph{$d$-dimensional ball graph} is the intersection graph of $d$-dimensional closed balls of $\mathbb R^d$.
We shorten $2$-dimensional ball graph in \emph{disk graph}, and $3$-dimensional ball graph in \emph{ball graph}.
A {$d$-dimensional unit ball graph} is the intersection graph of unit $d$-dimensional closed balls of $\mathbb R^d$, that is, balls with radius $1$.
Unit $d$-dimensional ball graphs can be thought of only with points: vertices are points (at the center of the ball) and two points are adjacent if they are at distance at most $2$.
In particular, solving \cli on those graphs is equivalent to finding a maximum sub-collection of points whose diameter is at most a fixed value.

\paragraph*{Approximation schemes}
A PTAS (Polynomial-Time Approximation Scheme) for a minimization (resp. maximization) problem is an approximation algorithm which takes an additional parameter $\varepsilon>0$ and outputs in time $n^{f(\varepsilon)}$ a solution of value at most $(1+\varepsilon)\text{OPT}$ (resp. at least $(1-\varepsilon)\text{OPT}$) where $\text{OPT}$ is the optimum value.
Observe that from now on, we consider that approximation ratios of maximization problems are smaller than $1$, unlike the convention we used in the introduction. 
An EPTAS (Efficient PTAS) is the same with running time $f(\varepsilon)n^{O(1)}$, an FPTAS (Fully PTAS) has running time $(1/\varepsilon)^{O(1)}n^{O(1)}$, a QPTAS (Quasi PTAS) has running time $n^{\text{polylog}~n}$ for every $\varepsilon > 0$.
Finally, and this is less standard, we call SUBEXPAS (subexponential AS) an approximation scheme with running time $2^{n^\gamma}$ for some $\gamma < 1$ not depending on $\varepsilon$, for every $\varepsilon > 0$.
All those approximation schemes can come deterministic or randomized.

\paragraph*{The class $\cl(d,\beta,i)$}
In Section~\ref{sec:eptas}, we present a randomized EPTAS and a deterministic PTAS for approximating the independence number $\alpha$ on graphs with constant VC-dimension, linear independence number, and the induced odd cycle packing number equal to 1. 

Actually, we extend these algorithms to the case $\iocp(G)=i$, for any constant $i$. Let $\cl(d,\beta,i)$ be the class of simple graphs $G$ satisfying:
\begin{itemize}
\item $\vcdim(G) \leqslant d$,
\item $\alpha(G) \geqslant \beta |V(G)|$, and
\item $\iocp(G) \leqslant i$.
\end{itemize}

For any positive constants $d, \beta<1, i$, we get a deterministic PTAS and a randomized EPTAS for \mis on $\cl(d,\beta,i)$.

\section{Disk graphs with co-degree at most 2}\label{sec:structural}

In this section, we fully characterize the degree-2 complements of disk graphs.
We show the following:

\begin{theorem}\label{thm:main-structural}
  A disjoint union of paths and cycles is the complement of a disk graph if and only if the number of odd cycles is at most one. 
\end{theorem}

We split this theorem into two parts.
In the first one, Section~\ref{subsec:notco-disk}, we show that the union of two disjoint odd cycles is not the complement of a disk graph.
This is the part that will be algorithmically useful.
As disk graphs are closed under taking induced subgraphs, it implies that in the complement of a disk graph two vertex-disjoint odd cycles have to be linked by at least one edge.
This will turn out useful when solving \mis on the complement of the graph (to solve \cli on the original graph).
In the second part, Section~\ref{subsec:co-disk}, we show how to represent the complement of the disjoint union of even cycles and exactly one odd cycle.
Although this result is not needed for the forthcoming algorithmic section, it nicely highlights the singular role that parity plays and exposes the complete set of disk graphs of co-degree at most 2. 

\subsection{The disjoint union of two odd cycles is not co-disk}
\label{subsec:notco-disk}

We call \emph{positive distance} between two non-intersecting disks the minimum of $d(x,y)$ where $x$ is in one disk and $y$ is in the other.
If the disks are centered at $c_1$ and $c_2$ with radius $r_1$ and $r_2$, respectively, then this value is $d(c_1,c_2)-r_1-r_2$.
We call \emph{negative distance} between two intersecting disks the length of the straight-line segment defined as the intersection of three objects: the two disks and the line joining their centers.
This value is $r_1+r_2-d(c_1,c_2)$, which is positive.

We call \emph{proper representation} a disk representation where every edge is witnessed by a proper intersection of the two corresponding disks, i.e., the interiors of the two disks intersect.
It is easy to transform a disk representation into a proper representation (of the same graph). 

\begin{lemma}\label{lem:proper}
  If a graph has a disk representation, then it has a proper representation.
\end{lemma}

\begin{proof}
  If two disks intersect non-properly, we increase the radius of one of them by $\varepsilon/2$ where $\varepsilon$ is the smallest positive distance between any two disks in the representation.
\end{proof}

In order not to have to discuss about the special case of three aligned centers in a disk representation, we show that such a configuration is never needed to represent a disk graph. 

\begin{lemma}\label{lem:generalPosition}
  If a graph has a disk representation, it has a proper representation where no three centers are aligned.
\end{lemma}

\begin{proof}
  By Lemma~\ref{lem:proper}, we have or obtain a proper representation.
  Let $\varepsilon$ be the minimum between the smallest positive distance and the smallest negative distance.
  As the representation is proper, $\varepsilon > 0$.
  If three centers are aligned, we move one of them to any point which is not lying in a line defined by two centers in a ball of radius $\varepsilon/2$ centered at it.
  This decreases the number of triples of aligned centers by at least one, and can be repeated until no three centers are aligned.
\end{proof}

From now on, we assume that every disk representation is proper and without three aligned centers.
We show the folklore result that in a representation of $K_{2,2}$ that sets the four centers in convex position, both non-edges have to be \emph{diagonal}.

\begin{lemma}\label{lem:cotwoKtwo}
  In a disk representation of $K_{2,2}$ with the four centers in convex position, the non-edges are between vertices corresponding to opposite centers in the quadrangle. 
\end{lemma}
\begin{proof}
  Let $c_1$ and $c_2$ be the centers of one non-edge, and $c_3$ and $c_4$ the centers of the other non-edge.
  Let $r_i$ be the radius associated to center $c_i$ for $i \in [4]$.
  It should be that $d(c_1,c_2)>r_1+r_2$ and $d(c_3,c_4)>r_3+r_4$ (see Figure~\ref{fig:k22}).
  Assume $c_1$ and $c_2$ are consecutive on the convex hull formed by $\{c_1, c_2, c_3, c_4\}$, and say, without loss of generality, that the order is $c_1, c_2, c_3, c_4$.
  Let $c$ be the intersection of $\seg(c_1,c_3)$ and $\seg(c_2,c_4)$.
  It holds that $d(c_1,c_3) + d(c_2,c_4) = d(c_1,c) + d(c,c_3) + d(c_2,c) + d(c,c_4) = (d(c_1,c)+d(c,c_2)) + (d(c_3,c)+d(c,c_4)) > d(c_1,c_2) + d(c_3,c_4) > r_1 + r_2 + r_3 + r_4 = (r_1+r_3)+(r_2+r_4)$.
  Which implies that $d(c_1,c_3) > r_1+r_3$ or $d(c_2,c_4) > r_2+r_4$; a contradiction.
\begin{figure}
\centering
\begin{tikzpicture}[scale=0.25,
    dot/.style={fill,circle,inner sep=-0.02cm}
  ]
\draw[very thin,fill=red, fill opacity=0.2] (11.8405399558, 12.6637388353) circle (3.82299906735);
\node[dot] at (11.8405399558, 12.6637388353) {};
\node at (11.8405399558, 12.6637388353 - 1.2) {$c_1$};
\draw[very thin,fill=red, fill opacity=0.2] (18.7026380772, 18.0821427855) circle (4.22302853751);
\node[dot] at (18.7026380772, 18.0821427855) {};
\node at (18.7026380772, 18.0821427855 + 1.2) {$c_2$};
\draw[very thin,fill=red, fill opacity=0.2] (11.5839455246, 16.9641945052) circle (3.96909254758);
\node[dot] at (11.5839455246, 16.9641945052) {};
\node at (11.5839455246, 16.9641945052 + 1.2) {$c_3$};
\draw[very thin,fill=red, fill opacity=0.2] (18.4743856348, 12.97903103) circle (3.46657454294);
\node[dot] at (18.4743856348, 12.97903103) {};
\node at (18.4743856348, 12.97903103 - 1.2) {$c_4$};
\end{tikzpicture}
\caption{Disk realization of a $K_{2,2}$. As the centers are positioned, it is impossible that the two non-edges are between the disks 2 and 3, and between the disks 1 and 4 (or between the disks 1 and 3, and between the disks 2 and 4).}\label{fig:k22}
\end{figure}
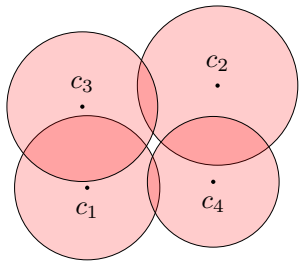
\end{proof}

We derive a useful consequence of the previous lemma, phrased in terms of intersections of lines and segments.

\begin{corollary}\label{cor:intersect}
In any disk representation of $K_{2,2}$ with centers $c_1, c_2, c_3, c_4$ with the two non-edges between the vertices corresponding to $c_1$ and $c_2$, and between $c_3$ and $c_4$, it should be that $\ell(c_1,c_2)$ intersects $\seg(c_3,c_4)$ or $\ell(c_3,c_4)$ intersects $\seg(c_1,c_2)$.    
\end{corollary}
\begin{proof}
  First, assume that the disk representation has the four centers in convex position.
  Then, by Lemma~\ref{lem:cotwoKtwo}, $\seg(c_1,c_2)$ and $\seg(c_3,c_4)$ are the diagonals of a convex quadrangle.
  Hence they intersect, and \emph{a fortiori}, $\ell(c_1,c_2)$ intersects $\seg(c_3,c_4)$ ($\ell(c_3,c_4)$ intersects $\seg(c_1,c_2)$, too).

  In the other case, the disk representation has one center, say without loss of generality, $c_1$, in the interior of the triangle formed by the other  three centers.
  In this case, $\ell(c_1,c_2)$ intersects $\seg(c_3,c_4)$.
  If instead a center in $\{c_3,c_4\}$ is in the interior of the triangle formed by the other centers, then $\ell(c_3,c_4)$ intersects $\seg(c_1,c_2)$.
\end{proof}

We can now prove the main result of this section thanks to the previous corollary, parity arguments, and \emph{some elementary properties of closed plane curves}, namely Property I and Property III of the eponymous paper \cite{Tait1877}.


\forbiddenDisk*

\begin{proof}
  Let $s$ and $t$ be two positive integers and $G=\overline{C_s + C_t}$ the complement of the disjoint union of a cycle of length $s$ and a cycle of length $t$.
  Assume that $G$ is a disk graph.
  Let $\mathcal C_1$ (resp. $\mathcal C_2$) be the cycle embedded in the plane formed by $s$ (resp. $t$) straight-line segments joining the consecutive centers of disks along the first (resp. second) cycle.
 Observe that the segments of those two cycles correspond to the non-edges of $G$.
 We number the segments of $\mathcal C_1$ from $S_1$ to $S_s$, and the segments of $\mathcal C_2$, from $S'_1$ to $S'_t$.
  
  For the $i$-th segment $S_i$ of $\mathcal C_1$, let $a_i$ be the number of segments of $\mathcal C_2$ intersected by the line $\ell(S_i)$ prolonging $S_i$, let $b_i$ be the number of segments $S'_j$ of $\mathcal C_2$ such that the prolonging line $\ell(S'_j)$ intersects $S_i$, and let $c_i$ be the number of segments of $\mathcal C_2$ intersecting $S_i$.
  For the second cycle, we define similarly $a'_j$, $b'_j$, $c'_j$.
  The quantity $a_i+b_i-c_i$ counts the number of segments of $\mathcal C_2$ which can possibly represent a $K_{2,2}$ with $S_i$ according to Corollary~\ref{cor:intersect}.
  As we assumed that $G$ is a disk graph, $a_i+b_i-c_i = t$ for every $i \in [s]$.
  Otherwise there would be at least one segment $S'_j$ of $\mathcal C_2$ such that $\ell(S_i)$ does not intersect $S'_j$ \emph{and} $\ell(S'_j)$ does not intersect $S_i$.
  
  By summing, it holds that $\Sigma_{i=1}^s (a_i+b_i-c_i)= s t$.
  Observe that $a_i$ is an even integer since $\mathcal C_2$ is a closed curve.
  This implies that $s t$ and $\Sigma_{i=1}^s (b_i-c_i)$ have the same parity.
  Note that $\Sigma_{i=1}^s c_i$ counts the number of intersections of the two closed curves $\mathcal C_1$ and $\mathcal C_2$, and is therefore even.
  Finally, observe that $\Sigma_{i=1}^s b_i=\Sigma_{j=1}^t a'_j$ by reordering and reinterpreting the sum from the point of view of the segments of $\mathcal C_2$.
  Since the $a'_j$'s are all even, $\Sigma_{i=1}^s b_i$ is also even.

  This means than $s t$ is even.
  Therefore, $s$ and $t$ cannot be both odd integers.
\end{proof}

\subsection{The disjoint union of cycles with at most one odd is co-disk}
\label{subsec:co-disk}

We only show the following part of Theorem~\ref{thm:main-structural} to emphasize that, rather unexpectedly, parity plays a crucial role in disk graphs of co-degree at most 2.
It is also amusing that the complement of any odd cycle is a \emph{unit} disk graph while the complement of any even cycle of length at least 8 is not \cite{Atminas16}.
Here, the situation is somewhat reversed: complements of even cycles are \emph{easier} to represent than complements of odd cycles.

\begin{theorem}\label{thm:coEvenCycles}
The complement of the disjoint union of even cycles and one odd cycle is a disk graph.
\end{theorem}

\begin{proof}
We start with a disk representation of the complement of one even cycle $C_{2s}$.
Recall that this construction is not possible with \emph{unit} disks for even cycles of length at least 8.
  We assume that the vertices of the cycle $C_{2s}$ are $1, 2, \ldots, 2s$ in this order.
  For each $i \in [2s]$, the disk $\mathcal D_i$ encodes the vertex $i$. 
  We start by fixing the disks $\mathcal D_1$, $\mathcal D_2$, and $\mathcal D_{2s}$.
  Those three disks have the same radius.
  We place $\mathcal D_2$ and $\mathcal D_{2s}$ side by side: their centers have the same $y$-coordinate.
  They intersect and the distance between their centers is $\varepsilon > 0$.
  We define $\mathcal D_1$ as the disk above $\mathcal D_2$ and $\mathcal D_{2s}$ tangent to those two disks and sharing the same radius.
  We denote by $p_1$ its intersection with $\mathcal D_2$ and by $p_s$ its intersection with $\mathcal D_{2s}$.
  We then slightly shift $\mathcal D_1$ upward so that it does not touch (nor does it intersect) $\mathcal D_2$ and $\mathcal D_{2s}$ anymore.
  While we do this translation, we imagine that the points $p_1$ and $p_s$ remain fixed at the boundary of $\mathcal D_2$ and $\mathcal D_{2s}$ respectively (see Figure~\ref{fig:complement-one-even-cycle1}).
  Let $p_2, p_3, \ldots, p_{s-1}$ points in the interior of $\mathcal D_1$ and below the line $\ell(p_1,p_s)$ such that $p_1, p_2, \ldots, p_{s-1}, p_s$ form an $x$-monotone convex chain (see Figure~\ref{fig:complement-one-even-cycle2}). 
  
  \begin{figure}[h!]
    \centering
    \begin{minipage}{0.3\textwidth}
      \centering
    \begin{tikzpicture}[
        dot/.style={fill,circle,inner sep=-0.01cm},
        vert/.style={draw, fill=red, opacity=0.2},
        verta/.style={draw, fill=blue, opacity=0.2},
      ]
      \def\s{1}
      \coordinate (c2) at (0,0) ;
      \draw[vert] (c2) circle (1) ;
      \node at (c2) {$\mathcal D_2$} ;

      \coordinate (c2s) at (\s,0) ;
      \draw[vert] (c2s) circle (1) ;
      \node at (c2s) {$\mathcal D_{2s}$} ;

      \coordinate (c1) at (\s / 2,2 - \s / 17) ;
      \draw[verta] (c1) circle (1) ;
      \node at (c1) {$\mathcal D_1$} ;

      \node[dot] at (0.25,0.96) {} ;
      \node[dot] at (\s - 0.25,0.96) {} ;

      \node at (0.05,0.96) {$p_1$} ;
      \node at (\s - 0.05,0.96) {$p_s$} ;   
    \end{tikzpicture}
    \subcaption{Three important disks with the same size $\mathcal D_1$, $\mathcal D_2$, $\mathcal D_{2s}$.}
    \label{fig:complement-one-even-cycle1}
    \end{minipage}
  \qquad
  \begin{minipage}{0.6\textwidth}
      \centering
      \begin{tikzpicture}
        [
          dot/.style={fill,circle,inner sep=-0.01cm}
        ]
        \centerarc[draw, fill=blue, fill opacity=0.2](0,0)(-50:-130:5) ;

        \def\xp{-3.2}
        \def\yp{-4}
        \def\xpb{3.2}
        \node[dot] at (\xp,\yp) {} ;
        \node at (\xp,\yp - 0.2) {$p_1$} ;
        \node[dot] at (\xpb,\yp) {} ;
        \node at (\xpb,\yp - 0.2) {$p_s$} ;

        \foreach \k/\i/\j in {2/0.5/0.2,3/1/0.36,4/1.5/0.48}{
          \node[dot] at (\xp+\i,\yp-\j) {} ;
          \node at (\xp+\i,\yp-\j- 0.2) {$p_\k$} ;
        }
        \foreach \k/\i/\j in {1/0.5/0.2,2/1/0.36,3/1.5/0.48}{
          \node[dot] at (\xpb-\i,\yp-\j) {} ;
          \node at (\xpb-\i,\yp-\j- 0.2) {$p_{s\text{-}\k}$} ;
        }
        \node at (0,-4.7) {$\ldots$} ;
        \node at (0,-4.2) {$\mathcal D_1$} ;
      \end{tikzpicture}
  \subcaption{Zoom where $\mathcal D_1$ almost touches $\mathcal D_2$ and $\mathcal D_{2s}$.}
  \label{fig:complement-one-even-cycle2}
  \end{minipage}
   \caption{The disks $\mathcal D_1$, $\mathcal D_2$, $\mathcal D_{2s}$ and the convex chain $p_1, p_2, \ldots, p_s$. The curvature of the boundary of $\mathcal D_1$ is exaggerated in the zoom for the sake of clarity.}
  \label{fig:complement-one-even-cycle}
\end{figure}
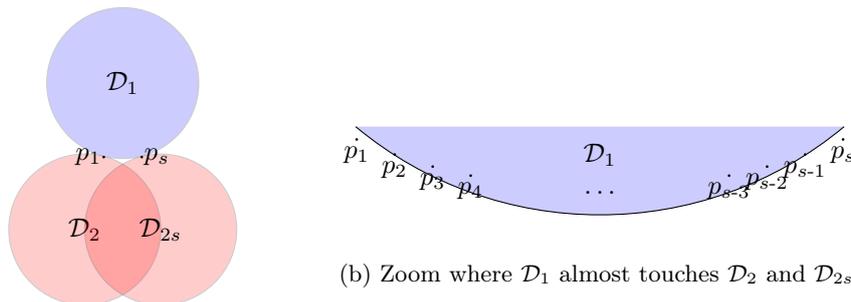

  Now, we define the disks $\mathcal D_4, \mathcal D_6, \ldots, \mathcal D_{2s-2}$.
  For each $i \in \{4,6,\ldots,2s-2\}$, let $\mathcal D_i$ be the unique disk with the same radius as $\mathcal D_2$ and such that the boundary of $\mathcal D_i$ crosses $p_{i/2}$ and is below its tangent $\tau_{i/2}$ at this point which has the direction of $\ell(p_{i/2-1},p_{i/2+1})$.
  
  It should be observed that the only disk with even index $i$ which contains $p_{i/2}$ is $\mathcal D_i$.
  We can further choose the convex chain $\{p_i\}_{i \in [s]}$ such that one co-tangent $\tau_{i,i+1}$ to $\mathcal D_{2i}$ and $\mathcal D_{2i+2}$ has a slope between the slopes of $\tau_i$ and $\tau_{i+1}$.
  Finally we define the disks $\mathcal D_3, \mathcal D_5, \ldots, \mathcal D_{2s-1}$.
  For each $i \in \{3,5,\ldots,2s-1\}$, let $\mathcal D_i$ be tangent to $\tau_{i,i+1}$ at the point of $x$-coordinate the mean between the $x$-coordinates of $p_{\frac{i-1}{2}}$ and $p_{\frac{i+1}{2}}$.
  Moreover, $\mathcal D_i$ is above $\tau_{i,i+1}$ and has a radius sufficiently large to intersect every disk with even index which are not $\mathcal D_{i-1}$ and $\mathcal D_{i+1}$.
  It is easy to see that the disks $\mathcal D_i$ with even index (resp. odd index) form a clique.
  By construction, the disk $\mathcal D_i$ with odd index greater than 3 intersects every disk with even index except $\mathcal D_{i-1}$ and $\mathcal D_{i+1}$ since $\mathcal D_i$ is on the other side of $\tau_{i,i+1}$ than those two disks.
  As the line $\tau_{i,i+1}$ intersects every other disk with even index, there is a sufficiently large radius so that $\mathcal D_i$ does so, too.
  The particular case of $\mathcal D_1$ has been settled at the beginning of the construction.
  This disk avoids $\mathcal D_2$ and $\mathcal D_{2s}$ and contains $p_2, p_3, \ldots, p_{s-1}$, so intersects all the other disks with even index.

  We now explain how to \emph{stack} even cycles. 
  We make the distance $\varepsilon$ between the center of $\mathcal D_2$ and $\mathcal D_{2s}$ a thousandth of their common radius.
  Note that this distance does not depend on the value of $s$.
  We identify the small region (point) where the disk $\mathcal D_1$ intersects with the disks of even index, between two different complements of cycles.
  We then rotate from this point one representation by a small angle (see Figure~\ref{fig:even-cycles-complement} for multiple complements of even cycles stacked).

\begin{figure}[h!]
      \centering
      \begin{tikzpicture}[
          dot/.style={fill,circle,inner sep=-0.01cm},
          vert/.style={draw, very thin, fill=red, fill opacity=0.2},
          verta/.style={draw, very thin, fill=blue, fill opacity=0.2},
          extended line/.style={shorten >=-#1,shorten <=-#1},
          extended line/.default=1cm,
          one end extended/.style={shorten >=-#1},
          one end extended/.default=1cm,
        ]
        \def\s{1}
        \def\hs{-4}
        \def\he{4}
        \def\ve{3}
        \coordinate (c1) at (0,1) ;
        \draw[verta] (c1) circle (1) ;
        \coordinate (c2) at (0,-1) ;
        \draw[vert] (c2) circle (1) ;

        \draw[very thin] (\hs,0) -- (\he,0) ;
        \fill[blue,fill opacity=0.2] (\hs,0) -- (\he,0) -- (\he,\ve) -- (\hs,\ve) -- cycle ;

        \foreach \i in {-20,-10,10,20}{
        \begin{scope}[rotate=\i]
        \coordinate (c1) at (0,1) ;
        \draw[verta] (c1) circle (1) ;
        \coordinate (c2) at (0,-1) ;
        \draw[vert] (c2) circle (1) ;
        \draw[very thin] (\hs,0) -- (\he,0) ;

        \fill[blue,opacity=0.2] (\hs,0) -- (\he,0) -- (\he,\ve) -- (\hs,\ve) -- cycle ;
        \end{scope}
        }

        \node at (0,1) {$\mathcal D_1$} ;
        \node at (0,2.4) {$\mathcal D_{2i+1}$} ;
        \node at (0,-1) {$\mathcal D_{2i}$} ;
      \end{tikzpicture}
      \caption{A disk realization of the complement of the disjoint union of an arbitrary number of even cycles.}
      \label{fig:even-cycles-complement}
\end{figure}
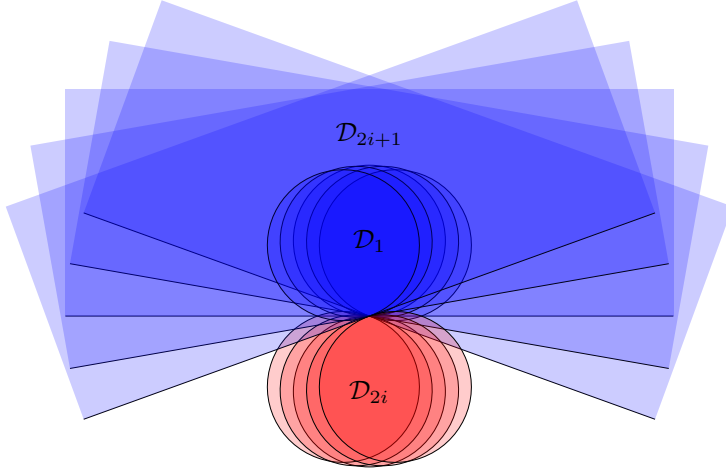

The reason why there are indeed all the edges between two complements of cycles is intuitive and depicted in Figure~\ref{fig:clique-minus-matching} and more specifically Figure~\ref{fig:cmm-lines}.
We superimpose all the complements of even cycles in a way that the maximum rotation angle between two complements of cycles is small (see for instance Figure~\ref{fig:even-cycles-odd-cycle-complement}).

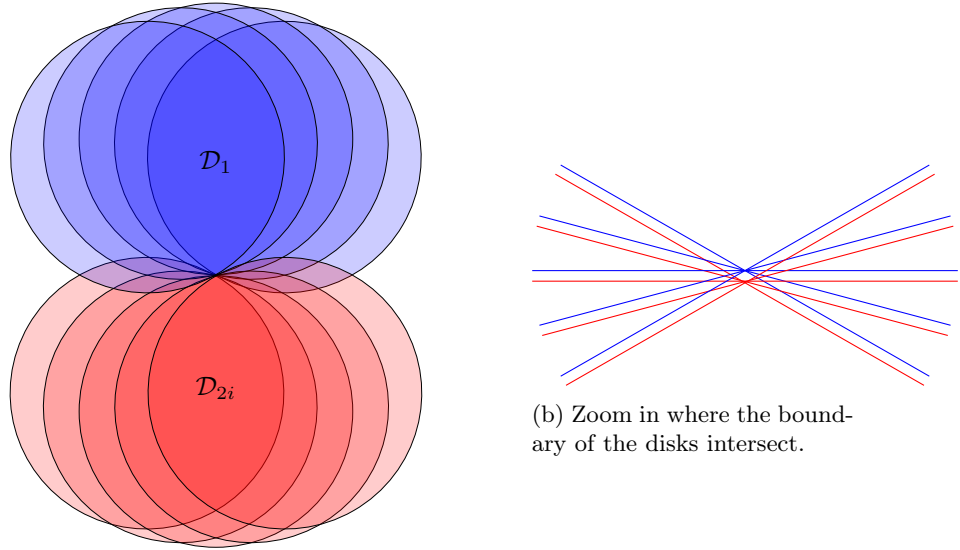
\begin{figure}[h!]
  \centering
  \begin{minipage}{0.55\textwidth}
    \centering
    \begin{tikzpicture}
      [
        scale = 1.5,
          dot/.style={fill,circle,inner sep=-0.01cm},
          vert/.style={draw, very thin, fill=red, fill opacity=0.2},
          verta/.style={draw, very thin, fill=blue, fill opacity=0.2},
          vertb/.style={draw, very thin, fill=green, fill opacity=0.2},
      ]
    \def\r{1.2}
    \def\n{5}
    \foreach \h in {1,...,\n}{
      \pgfmathsetmacro{\i}{15 * (\h - \n / 2 - 0.5)}
      \pgfmathsetmacro{\j}{25 * (\h-1)}
      \draw[rotate=\i,verta] (0,0) arc (-90:270:\r) ;
      \draw[rotate=\i,vert] (0,-0.01) arc (90:-270:\r) ;
    }
     \node at (0,1) {$\mathcal D_1$} ;
    \node at (0,-1) {$\mathcal D_{2i}$} ;
  \end{tikzpicture}
  \subcaption{The only potential non-edges are between two disks represented almost tangent.}
  \label{fig:cmm-circles}
  \end{minipage}
  \qquad
  \begin{minipage}{0.35\textwidth}
      \centering
      \begin{tikzpicture}
        [
          scale = 1.4,
          dot/.style={fill,circle,inner sep=-0.01cm},
          vert/.style={draw, very thin, fill=red, fill opacity=0.2},
          verta/.style={draw, very thin, fill=blue, fill opacity=0.2},
          vertb/.style={draw, very thin, fill=green, fill opacity=0.2},
        ]
        \def\n{5}
        \foreach \h in {1,...,\n}{
          \pgfmathsetmacro{\i}{15 * (\h - \n / 2 - 0.5)}
          \pgfmathsetmacro{\j}{25 * (\h-1)}
          \draw[rotate=\i,blue] (0,0) -- (-2,0) ;
          \draw[rotate=\i,blue] (0,0) -- (2,0) ;
          \draw[rotate=\i,red] (0,-0.1) -- (-2,-0.1) ;
          \draw[rotate=\i,red] (0,-0.1) -- (2,-0.1) ;
        }
  \end{tikzpicture}
  \subcaption{Zoom in where the boundary of the disks intersect.}
  \label{fig:cmm-lines}
  \end{minipage}
   \caption{Zoom in where the disk $\mathcal D_1$ of the several complements of even cycles intersects all the $\mathcal D_{2i}$ of the other cycles.}
  \label{fig:clique-minus-matching}
\end{figure}

Finally, we need to add one disjoint odd cycle in the complement.
There is a nice representation of a complement of an odd cycle by unit disks in the paper of Atminas and Zamaraev \cite{Atminas16} (see Figure~\ref{fig:atminas-zamaraev}).
\begin{figure}[h!]
      \centering
      \begin{tikzpicture}[
          dot/.style={fill,circle,inner sep=-0.01cm},
          vert/.style={draw, very thin, fill=red, fill opacity=0.2},
          verta/.style={draw, very thin, fill=blue, fill opacity=0.2},
          vertb/.style={draw, very thin, fill=green, fill opacity=0.2},
        ]
        \def\s{13}
        \foreach \i in {1,...,\s}{
        \begin{scope}[rotate=360 * \i / 13]
        \coordinate (d\i) at (0,1) ;
        \draw[vertb] (d\i) circle (0.95) ;
        \end{scope}
        }
        
      \end{tikzpicture}
      \caption{A disk realization of the complement of an odd cycle with unit disks as described by Atminas and Zamaraev \cite{Atminas16}. Unfortunately, we cannot use this representation.}
      \label{fig:atminas-zamaraev}
\end{figure}
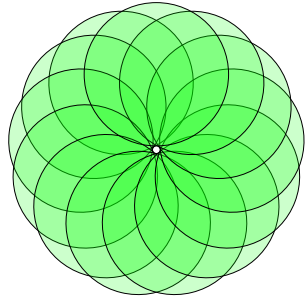

However, we will use a different and non-unit representation for the next step to work.
Let $2s+1$ be the length of the cycle.
We use a similar construction as for the complement of an even cycle.
We denote the disks $\mathcal D'_1, \mathcal D'_2, \ldots, \mathcal D'_{2s+1}$.
The difference is that we separate $\mathcal D'_1$ away from $\mathcal D'_2$ but not from $\mathcal D'_{2s}$.
Then, we represent all the disks with odd index but $\mathcal D'_{2s+1}$ as before.
The disk $\mathcal D'_{2s+1}$ is chosen as being cotangent to $\mathcal D'_1$ and $\mathcal D'_{2s}$ and to the left of them.
Then we very slightly move $\mathcal D'_{2s+1}$ to the left so that it does not intersect those two disks anymore.
The disk $\mathcal D'_{2s}$ has the rightmost center among the disks with even index.
Therefore $\mathcal D'_{2s+1}$ still intersects all the other disks of even index.

Moreover, the disks with even index form a clique and the disks with odd index form a clique minus an edge between the vertex $1$ and the vertex $2s+1$.
Hence, the intersection graph of those disks is indeed the complement of $C_{2s+1}$ (see Figure~\ref{fig:odd-cycle-complement}).

\begin{figure}[h!]
      \centering
      \begin{tikzpicture}[
          dot/.style={fill,circle,inner sep=-0.01cm},
          vert/.style={draw, very thin, fill=red, fill opacity=0.2},
          verta/.style={draw, very thin, fill=blue, fill opacity=0.2},
          vertb/.style={draw, very thin, fill=green, fill opacity=0.2},
        ]
        \def\s{1}
        \def\hs{-4}
        \def\he{4}
        \def\ve{3}
        \def\vs{-2}
        \coordinate (c1) at (0,1) ;
        \draw[verta] (c1) circle (1) ;
        \coordinate (c2) at (0,-1) ;
        \draw[vert] (c2) circle (1) ;

        \draw[very thin] (\hs,0) -- (\he,0) ;
        \fill[green,fill opacity=0.2] (\hs,0) -- (\he,0) -- (\he,\ve) -- (\hs,\ve) -- cycle ;

        \draw[very thin] (-1,\vs) -- (-1,\ve) ;
        \fill[green,fill opacity=0.2] (-1,\vs) -- (-1,\ve) -- (\hs,\ve) -- (\hs,\vs) -- cycle ;
        
        \node at (0,2.4) {$\mathcal D'_{2i+1}$} ;
        \node at (0,-1) {$\mathcal D'_{2i}$} ;
        \node at (-2.5,-0.3) {$\mathcal D'_{2s+1}$} ;
        \node at (0,1) {$\mathcal D'_1$} ;
      \end{tikzpicture}
      \caption{A disk realization of the complement of an odd cycle of length $2s+1$.}
      \label{fig:odd-cycle-complement}
\end{figure}
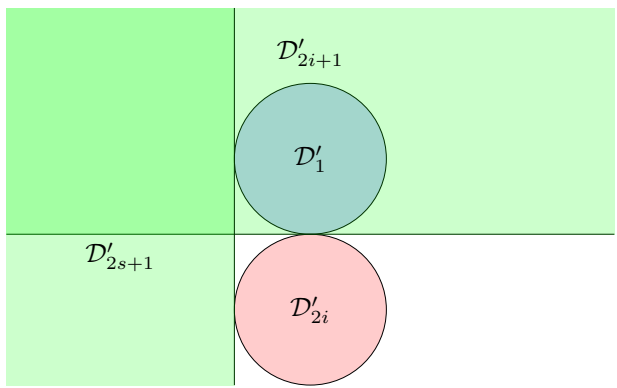

This representation of $\overline{C_{2s+1}}$ can now be put on top of complements of even cycles.
We identify the small region (point) where the disk $\mathcal D_1$ intersects the disks of even index (in complements of even cycles) with the small region (point) where the disk $\mathcal D'_1$ intersects the disks of even index (in the one complement of odd cycle).
We make the disk $\mathcal D'_1$ significantly smaller than $\mathcal D_1$ and rotate the representation of $\overline{C_{2s+1}}$ by a sizable angle, say 60 degrees (see Figure~\ref{fig:even-cycles-odd-cycle-complement}). 

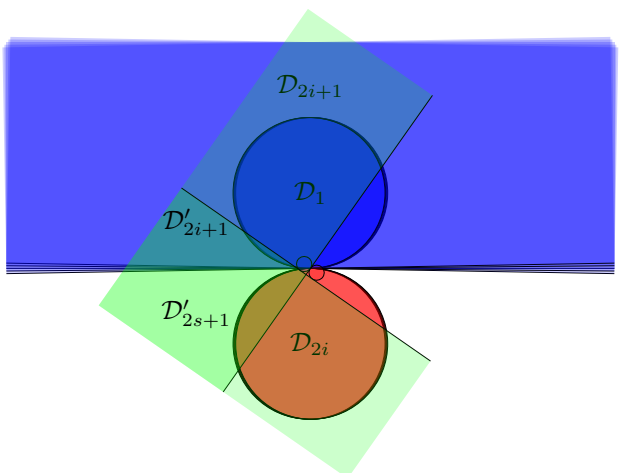
\begin{figure}[h!]
      \centering
      \begin{tikzpicture}[
          dot/.style={fill,circle,inner sep=-0.01cm},
          vert/.style={draw, very thin, fill=red, fill opacity=0.2},
          verta/.style={draw, very thin, fill=blue, fill opacity=0.2},
          vertb/.style={draw, very thin, fill=green, fill opacity=0.2},
        ]        
         \def\s{1}
        \def\hs{-4}
        \def\he{4}
        \def\ve{3}
        \coordinate (c1) at (0,1) ;
        \draw[verta] (c1) circle (1) ;
        \coordinate (c2) at (0,-1) ;
        \draw[vert] (c2) circle (1) ;

        \draw[very thin] (\hs,0) -- (\he,0) ;
        \fill[blue,fill opacity=0.2] (\hs,0) -- (\he,0) -- (\he,\ve) -- (\hs,\ve) -- cycle ;

        \foreach \i in {-1,-0.5,0.5,1}{
        \begin{scope}[rotate=\i]
        \coordinate (c1) at (0,1) ;
        \draw[verta] (c1) circle (1) ;
        \coordinate (c2) at (0,-1) ;
        \draw[vert] (c2) circle (1) ;
        \draw[very thin] (\hs,0) -- (\he,0) ;

        \fill[blue,opacity=0.2] (\hs,0) -- (\he,0) -- (\he,\ve) -- (\hs,\ve) -- cycle ;
        \end{scope}
        }

        \node at (0,1) {$\mathcal D_1$} ;
        \node at (0,2.4) {$\mathcal D_{2i+1}$} ;
        \node at (0,-1) {$\mathcal D_{2i}$} ;

        \def\t{0.1}
        \def\nve{2}
        \def\nhs{-2}
        \def\nhe{2.8}
        
        \begin{scope}[rotate=55]
        \coordinate (d1) at (0,\t) ;
        \draw[verta] (d1) circle (\t) ;
        \coordinate (d2) at (0,-\t) ;
        \draw[vert] (d2) circle (\t) ;

        \draw[very thin] (\nhs,0) -- (\nhe,0) ;
        \fill[green,fill opacity=0.2] (\nhs,0) -- (\nhe,0) -- (\nhe,\nve) -- (\nhs,\nve) -- cycle ;

        \draw[very thin] (-\t,-\nve) -- (-\t,\nve) ;
        \fill[green,fill opacity=0.2] (-\t,-\nve) -- (-\t,\nve) -- (\nhs,\nve) -- (\nhs,-\nve) -- cycle ;
        \end{scope}

        \node at (-1.5,0.6) {$\mathcal D'_{2i+1}$} ;
        \node at (-1.5,-0.6) {$\mathcal D'_{2s+1}$} ;

      \end{tikzpicture}
      \caption{Placing the complement of odd cycle on top of the complements of even cycles.}
      \label{fig:even-cycles-odd-cycle-complement}
\end{figure}

It is easy to see that the disks of the complement of the odd cycle intersect all the disks of the complements of even cycles.
A good sanity check is to observe why we cannot stack representations of complements of odd cycles, with the same rotation scheme.
In Figure~\ref{fig:sanity-check}, the rotation of two representations of the complement of an odd cycle leaves disks $\mathcal D'_1$ and $\mathcal D''_{2s'+1}$ far apart when they should intersect.
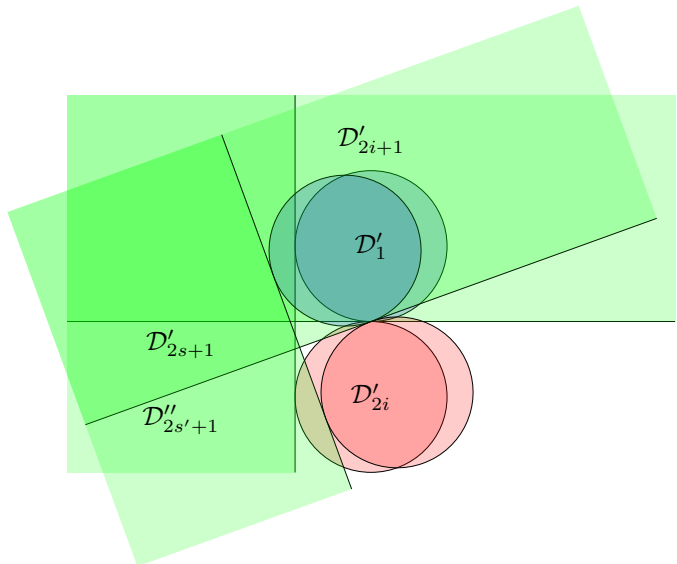
\begin{figure}[h!]
      \centering
      \begin{tikzpicture}[
          dot/.style={fill,circle,inner sep=-0.01cm},
          vert/.style={draw, very thin, fill=red, fill opacity=0.2},
          verta/.style={draw, very thin, fill=blue, fill opacity=0.2},
          vertb/.style={draw, very thin, fill=green, fill opacity=0.2},
        ]
        \def\s{1}
        \def\hs{-4}
        \def\he{4}
        \def\ve{3}
        \def\vs{-2}
        \foreach \i in {0,20}{
        \begin{scope}[rotate=\i]
        \coordinate (c1) at (0,1) ;
        \draw[verta] (c1) circle (1) ;
        \coordinate (c2) at (0,-1) ;
        \draw[vert] (c2) circle (1) ;

        \draw[very thin] (\hs,0) -- (\he,0) ;
        \fill[green,fill opacity=0.2] (\hs,0) -- (\he,0) -- (\he,\ve) -- (\hs,\ve) -- cycle ;

        \draw[very thin] (-1,\vs) -- (-1,\ve) ;
        \fill[green,fill opacity=0.2] (-1,\vs) -- (-1,\ve) -- (\hs,\ve) -- (\hs,\vs) -- cycle ;
        \end{scope}
        }
        
        \node at (0,2.4) {$\mathcal D'_{2i+1}$} ;
        \node at (0,-1) {$\mathcal D'_{2i}$} ;
        \node at (-2.5,-0.3) {$\mathcal D'_{2s+1}$} ;
        \node at (0,1) {$\mathcal D'_1$} ;

        \node at (-2.5,-1.3) {$\mathcal D''_{2s'+1}$} ;
      \end{tikzpicture}
      \caption{Sanity check: trying to stack the complements of two odd cycles fails.
      The disks $\mathcal D'_1$ and $\mathcal D''_{2s'+1}$ do not intersect.}
      \label{fig:sanity-check}
\end{figure}
\end{proof}

Theorem~\ref{thm:main-structural-non-disk} and Theorem~\ref{thm:coEvenCycles}, together with the fact that disk graphs are closed by taking induced subgraphs prove Theorem~\ref{thm:main-structural}.

\section{Same obstruction for unit ball graphs}\label{sec:unit-ball-obstruction}

In this section, we show that unit ball graphs, like disk graphs, do not contain the complement of two disjoint odd cycles.
In other words, for any unit ball graph $G$, $\iocp(\overline G) \leqslant 1$.
It is interesting to note that the proofs we could find for disk graphs and unit ball graphs turn out to be quite different.

A \emph{closed polygonal chain} $C$ in $\mathbb R^d$ is defined by a sequence of points (or vertices) $x_1, x_2, \ldots, x_p \in \mathbb R^d$ as the straight-edge segments $x_1x_2$, $x_2x_3$, \dots, $x_{p-1}x_p$, $x_px_1$.
We call \emph{direction} of a non-zero vector its equivalence class by the relation $\vec{u} \sim \vec{v} \Leftrightarrow \exists \lambda \in \mathbb R^+, \vec{u} = \lambda \vec{v}$.
We denote the direction of $\vec{u}$ by $\dir{\vec{u}}$.
We define the set of directions $$\needle{C}:=\underset{1 \leqslant i \leqslant p}{\bigcup} \fork{x_{i-1}}{x_i}{x_{i+1}} \cup \forkb{x_{i-1}}{x_i}{x_{i+1}},$$ 
where the indices are taken modulo $p$, $\fork{x_{i-1}}{x_i}{x_{i+1}} := \{\dir{\overrightarrow{x_ix}} ~|~ x \in x_{i-1}x_{i+1} \},$ and $\forkb{x_{i-1}}{x_i}{x_{i+1}} := \{\dir{\overrightarrow{xx_i}} ~|~ x \in x_{i-1}x_{i+1} \}$. 

\begin{lemma}\label{lem:same-direction}
  Let $C_1$ and $C_2$ be two closed polygonal chains of~$\mathbb R^3$ on an odd number of vertices each.
  Then, $\needle{C_1} \cap \needle{C_2}$ is non-empty.
\end{lemma}

\begin{proof}
  We want to establish the existence of a direction which is common to $\needle{C_1}$ and $\needle{C_2}$.
  We identify $\needle{C_i}$ ($i \in \{1,2\}$) to its trace on the 2-sphere.
  Indeed the set of directions in $\mathbb R^3$ is isomorphic to the set of points on the 2-sphere.
  
  Let $x_1, x_2, \ldots, x_p$ be the vertices of $C_1$.
  We show that $\needle{C_1}$ is path-connected (due to $p$ being odd).
  To do so, we continuously modify (within $\needle{C_1}$) the initial vector $\dir{\overrightarrow{x_1x_2}}$ into any other vector of $\needle{C_1}$.
  We start with $a$ in $x_1$ and $b$ in $x_2$.
  We continuously move $a$ from $x_1$ to $x_3$ (on the straight-edge segment $x_1x_3$) while $b$ stays fixed at $x_2$ 
  (we are moving in $\fork{x_1}{x_2}{x_3}$).
  For the next step, $a$ is fixed at $x_3$ and $b$ continuously moves from $x_2$ to $x_4$.
  (we are moving in $\forkb{x_{2}}{x_3}{x_{4}}$).
  And in general, we move the point with index $i-1$ from $x_{i-1}$ to $x_{i+1}$ while the other point stays fixed at $x_i$ (where the indices are modulo $p$).
  Since $p$ is odd, we reach the situation where $b$ is set to $x_1$ and $a$ is set to $x_2$, when we have completed once the walk on the closed polygonal chain.
  We repeat a walk on the chain once again, so that $a$ is back to $x_1$ and $b$ is back to $x_2$, and we stop.
  This process spans $\needle{C_1}$.
  
  Therefore, $\mathcal C_1$ is a closed curve on the 2-sphere since we are finally back to $\dir{\overrightarrow{x_1x_2}}$, that is, from where we started.
  Furthermore, $\mathcal C_1$ is \emph{antipodal}, i.e., closed by taking antipodal points.
  Indeed, for each direction attained in a $\fork{x_{i-1}}{x_i}{x_{i+1}}$, we reach the opposite direction in $\forkb{x_{i-1}}{x_i}{x_{i+1}}$. 
  Similarly $\needle{C_2}$ draws a closed antipodal curve $\mathcal C_2$ on the 2-sphere.
  The curves $\mathcal C_1$ and $\mathcal C_2$ intersect since they are closed and antipodal.
  An intersection point corresponds to a direction shared by $\needle{C_1}$ and $\needle{C_2}$.
\end{proof}

We will apply this lemma on the closed polygonal chains $C_1$ and $C_2$ formed by the centers of unit balls realizing two odd cycle complements.
The contradiction will come from the fact that not all the pairs of centers $x \in C_1$ and $y \in C_2$ can be at distance at most 2.

\forbidden*

\begin{proof}
  Let $x_1, x_2, \ldots, x_p \in \mathbb R^3$ (resp.~$y_1, y_2, \ldots, y_q \in \mathbb R^3$) be the centers of unit balls representing the complement of an odd cycle of length $p$ (resp. $q$), such that $x_i$ and $x_{i+1}$ (resp. $y_i$ and $y_{i+1}$) encode the non-adjacent pairs.
  Let $C_1$ (resp. $C_2$) be the closed polygonal chain obtained from the centers $x_1, x_2, \ldots, x_p$ (resp. $y_1, y_2, \ldots, y_q$).
  By Lemma~\ref{lem:same-direction}, there are two collinear vectors $\overrightarrow{x_ix}$ and $\overrightarrow{y_jy}$ with $x$ on the straight-line segment $x_{i-1}x_{i+1}$ and $y$ on the straight-line segment $y_{j-1}y_{j+1}$.

  \begin{figure}[!ht]
    \centering
    \begin{tikzpicture}[scale=1.5,dot/.style={fill,circle,inner sep=-0.03cm}]
      \coordinate (a) at (-0.095,0,0) ;
      \coordinate (b) at (-0.2,1.6,0.5) ;
      \coordinate (c) at (0,1.6,-0.5) ;
      \coordinate (x) at (-0.05,1.6,-0.25) ;

      \coordinate (d) at (-0.28,0,0) ;
      \coordinate (e) at (0.8,2,0.8) ;
      \coordinate (f) at (-0.4,2,-0.4) ;
      \coordinate (y) at (-0.2,2,-0.2) ;

      \foreach \i in {a,b,c,d,e,f,x,y}{
        \node[dot] at (\i) {} ;
      }

      \draw[dashed,thin,red] (b) -- (a) -- (c) ;
      \draw[dashed,thin,blue] (e) -- (d) -- (f) ;
      \draw[dotted] (b) -- (c) ;
      \draw[dotted] (e) -- (f) ;
      \draw[thick,red] (a) -- (x) ;
      \draw[thick,blue] (d) -- (y) ;

      \node at (-0.3,-0.2,0) {$x_i$} ;
      \node at (-0.17,2.13,-0.2) {$x$};
      \node at (0,1.75,-0.2) {$y$} ;
      \node at (-0.08,-0.2,0) {$y_j$} ;

      \begin{scope}[xshift=3cm]
      \coordinate (x) at (-0.095,0,0) ;
      \coordinate (b) at (-0.2,0.3,0.5) ;
      \coordinate (c) at (0,-0.3,-0.5) ;
      \coordinate (a) at (-0.05,1.6,-0.25) ;

      \coordinate (d) at (-0.28,0,0) ;
      \coordinate (e) at (0.8,2,0.8) ;
      \coordinate (f) at (-0.4,2,-0.4) ;
      \coordinate (y) at (-0.2,2,-0.2) ;

      \foreach \i in {a,b,c,d,e,f,x,y}{
        \node[dot] at (\i) {} ;
      }

      \draw[dashed,thin,red] (b) -- (a) -- (c) ;
      \draw[dashed,thin,blue] (e) -- (d) -- (f) ;
      \draw[dotted] (b) -- (c) ;
      \draw[dotted] (e) -- (f) ;
      \draw[thick,red] (a) -- (x) ;
      \draw[thick,blue] (d) -- (y) ;

      \node at (-0.3,-0.2,0) {$x_i$} ;
      \node at (-0.17,2.13,-0.2) {$x$};
      \node at (0,1.75,-0.2) {$y_j$} ;
      \node at (-0.08,-0.2,0) {$y$} ;
      \end{scope}
    \end{tikzpicture}
    \caption{The two cases for the collinear vectors $\protect\overrightarrow{x_ix}$ and $\protect\overrightarrow{y_jy}$.}
    \label{fig:two-cases-collinear}
  \end{figure}
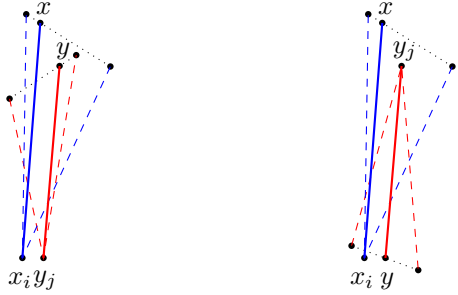
  Let us suppose that $\overrightarrow{x_ix}$ and $\overrightarrow{y_jy}$ have the same direction (Figure~\ref{fig:two-cases-collinear}, left).
  In the plane\footnote{or \emph{a} plane if it is not unique} containing $x_i$, $x$, $y_j$, and $y$, those four points are in convex position and the convex hull is cyclically ordered $x_i,x,y,y_j$.
  We obtain a contradiction by showing that the sum of the diagonals $d(x_i,y)+d(y_j,x)$ is strictly smaller than the sum of two opposite sides $d(x_i,x)+d(y_j,y)$.
  
  Considering edges and non-edges, for every $i$, the only vertices at distance greater than $2$ from $y_i$ in $\cup \{ x_j , y_j\}$ are $y_{i-1}$ and $y_{i+1}$. Then, we have $d(x_i,y_{j-1}) \leqslant 2 < d(y_j,y_{j-1})$ and $d(x_i,y_{j+1}) \leqslant 2 < d(y_j,y_{j+1})$.
  The points strictly closer to $x_i$ than to $y_j$ form an open half-space.
  In particular, they form a convex set and all the points in the segment $y_{j-1}y_{j+1}$ are therefore closer to $x_i$ than to $y_j$.
  Hence, $d(x_i,y) < d(y_j,y)$.
  Symmetrically, $d(y_j,x) < d(x_i,x)$.
  So $d(x_i,y)+d(y_j,x) < d(y_j,y)+d(x_i,x)$, a contradiction.
  
  Let us now assume that $\overrightarrow{x_ix}$ and $\overrightarrow{y_jy}$ have opposite direction (Figure~\ref{fig:two-cases-collinear}, right).
  In that case, the four coplanar points $x_i,x,y_j,y$ are in convex position in their plane and the convex hull is cyclically ordered $x_i,x,y_j,y$.
  We will attain the similar contradiction that $d(x,y)+d(x_i,y_j) < d(x_i,x)+d(y_j,y)$.
  As previously, $d(y_{j-1},x_{i-1}) \leqslant 2 < d(x_i,x_{i-1})$ and $d(y_{j-1},x_{i+1}) \leqslant 2 < d(x_i,x_{i+1})$.
  Hence, by convexity of the set of points closer to $y_{j-1}$ than to $x_i$, we have $d(y_{j-1},x) < d(x_i,x)$.
  Similarly, $d(y_{j+1},x) < d(x_i,x)$.
  We obtained that point $x$ is closer to $y_{j-1}$ and $y_{j+1}$ than to $x_i$.
  Therefore, applying again the convexity argument, we get that $d(x,y) < d(x,x_i)$.
  Besides, $d(x_i,y_j) \leqslant 2 < d(y,y_j)$.
  So we arrive at the contradiction $d(x,y)+d(x_i,y_j) < d(x,x_i)+d(y,y_j)$.
\end{proof}

\section{Direct algorithmic consequences}\label{sec:algorithms}

Now we show how to use the structural results from Section \ref{sec:structural} and Section~\ref{sec:unit-ball-obstruction} to obtain algorithms for \cli in disk and unit ball graphs.
A clique in a graph $G$ is an independent set in $\overline{G}$. 
So, leveraging the result from Theorem \ref{thm:main-structural-non-disk} and Theorem~\ref{thm:no-two-odd-cycles-ubg}, we will focus on solving \mis in graphs without two vertex-disjoint odd cycles as an induced subgraph.
 
\subsection{QPTAS}

Here, we present a simple argument to get a QPTAS for \cli in disk and unit ball graphs due to a known approximation algorithm for \mis on graphs with relatively small odd cycle packing number.
In Section~\ref{sec:eptas}, we will improve the QPTAS to a randomized EPTAS using a different approach.
We recall that the odd cycle packing number $\ocp(H)$ of a graph $H$ is the maximum number of vertex-disjoint odd cycles in $H$.
Unfortunately, the condition that $\overline{G}$ does not contain two vertex-disjoint odd cycles as an induced subgraph is not quite the same as saying that the odd cycle packing number of $\overline{G}$ is 1.
Otherwise, we would immediately get a PTAS by the following result of Bock et al.~\cite{Bock14}.
\begin{theorem}[Bock et al.~\cite{Bock14}]\label{thm:bock-ptas}
For every fixed $\varepsilon >0$ there is a polynomial $(1-\varepsilon)$-approximation algorithm for \mis for graphs $H$ with $n$ vertices and  $\ocp(H) = o(n / \log n)$.
\end{theorem}
The algorithm by Bock et al. \cite{Bock14} works in polynomial time if $\ocp(H) = o(n / \log n)$, but it does not need the odd cycle packing explicitly given as an input. This is important, since finding a maximum odd cycle packing is NP-hard \cite{DBLP:conf/stoc/KawarabayashiR10}.
We start by proving a structural lemma, which spares us having to determine the odd cycle packing number.

\begin{lemma}\label{lem:bigdeg}
Let $H$ be a graph with $n$ vertices, whose complement is a disk graph (resp.~unit ball graph).
If $\ocp(H) > n/ \log^2 n$, then $H$ has a vertex of degree at least $n / \log ^4 n$.
\end{lemma}
\begin{proof}
  Consider a maximum odd cycle packing $\mathcal{P}$.
  By assumption, it contains more than $n/\log^2 n$ vertex-disjoint cycles. 
  Hence, a shortest cycle $C$ in $\mathcal{P}$ has size at most $\log^2 n$.
  Now, by Theorem~\ref{thm:main-structural-non-disk} (resp.~by Theorem~\ref{thm:no-two-odd-cycles-ubg}), $H$ has no two vertex-disjoint odd cycles with no edges between them.
  Therefore there must be an edge from $C$ to every other cycle of $\mathcal{P}$, which constitutes at least $n / \log^2 n$ edges.
  Let $v$ be a vertex of $C$ with the maximum number of edges to other cycles in $\mathcal{P}$.
  By the pigeonhole principle, its degree is at least $n / \log^4 n$.
\end{proof}

We are ready to suggest a QPTAS for \cli in disk and unit ball graphs.

\begin{theorem}\label{thm:qptas}
For any $\varepsilon > 0$, \cli can be $(1-\varepsilon)$-approximated in time $2^{O(\log^5 n)}$, when the input is a disk graph (resp. unit ball graph) with $n$ vertices.
\end{theorem}

\begin{proof}
Let $G$ be the input disk graph (resp.~unit ball graph) and let $\overline{G}$ be its complement, we want to find a $(1-\varepsilon)$-approximation for \mis in $\overline{G}$. We consider two cases.
If $\overline{G}$ has no vertex of degree at least $n / \log ^4 n$, then, by Lemma \ref{lem:bigdeg}, we know that $\ocp(\overline{G}) \leqslant n / \log^2 n = o(n / \log n)$. In this case we run the PTAS of Theorem~\ref{thm:bock-ptas} and we are done.

In the other case, $\overline{G}$ has a vertex $v$ of degree at least $n / \log ^4 n$ (note that it may still be the case that $\ocp(\overline{G}) = o(n / \log n)$).
We branch on $v$: either we include $v$ in our solution and remove it and all its neighbors, or we discard $v$.
The complexity of this step is described by the recursion $F(n) \leqslant F(n-1) + F(n- n / \log^4 n) \leqslant F(n-2) + 2F(n- n / \log^4 n) \leqslant \cdots \leqslant (n / \log^4 n)F(n- n / \log^4 n)$.
Thus $F(n) \leqslant n^{O(\log^4 n)}=2^{O(\log^5 n)}$.
Note that this step is exact, i.e., we do not lose any solution.
\end{proof}

We can actually generalize the QPTAS from graphs with $\iocp \leqslant 1$ to graphs with $\iocp=O(1)$.
\begin{lemma}\label{lem:bigdegGen}
Let $H$ be a sufficiently large graph with $n$ vertices, and $\iocp(H)=i$ for some integer $i$.
If $\ocp(H) > n/ \log^2 n$, then $H$ has a vertex of degree at least $\frac{n}{2 i \log ^5 n}$.
\end{lemma}
\begin{proof}
  Again, let $\mathcal P$ be an odd cycle packing ($|\mathcal P| \geqslant n/\log^2 n$).
  The number of cycles of $\mathcal P$ with length at least $\log^3 n$ is at most $n/\log^3n$.
  Thus there are at least $n/\log^2 n - n/\log^3 n > \frac{1}{2} n/\log^2 n$ cycles of $\mathcal P$ of length at most $\log^3 n$.
  We now consider the graph $J$ whose vertices are the cycles of $\mathcal P$ of length at most $\log^3 n$, and where there is an edge between $C_1$ and $C_2$ if there is at least one edge in $H$ between a vertex of $C_1$ and a vertex of $C_2$.
  Let $n_J := |V(J)| > \frac{1}{2} n/\log^2 n$.
  
  By assumption, $\alpha(J) \leqslant i$ (since $\iocp(H)=i$).
  Hence by Turan's theorem, $J$ has at most $(1-\frac{1}{i}){n_J \choose 2}$ non-edges.
  So $J$ has at least $\frac{1}{i}{n_J \choose 2}$ edges, and average degree at least $\frac{n_J-1}{i} \geqslant n/(2 i \log^2 n)$.
  Let $C_\ell$ be a vertex of maximum degree in $J$.
  By construction, the cycle $C_\ell$ is adjacent to at least $n/(2 i \log^2 n)$ other cycles of $\mathcal P$, and $|C_\ell| \leqslant \log^3 n$.
  Hence the vertex of $C_\ell$ with the largest degree in $H$ has at least $n/(2 i \log^5 n)$ neighbors.
\end{proof}

From the previous lemma, we get a QPTAS with slightly worse running time, similarly to Theorem~\ref{thm:qptas}.
\begin{theorem}\label{thm:qptasGen}
For any $\varepsilon > 0$ and for every integer $d$, \mis (resp. \cli) can be $(1-\varepsilon)$-approximated in time $2^{O(i \log^6 n)}$, when the input is a graph $G$ on $n$ vertices with $\iocp(G)=i$ (resp. $\iocp(\overline G)=i$).
\end{theorem}

\subsection{Subexponential algorithm}

The \emph{odd girth} of a graph is the size of a shortest odd cycle.
An \emph{odd cycle cover} is a subset of vertices whose deletion makes the graph bipartite.
We will use a result by Györi et al. \cite{Gyori97}, which says that graphs with large odd girth have small odd cycle cover.
Bock et al. \cite{Bock14} turned the non-constructive proof into a polynomial-time algorithm.
\begin{theorem}[Györi et al. \cite{Gyori97}, Bock et al. \cite{Bock14}]\label{thm:occ}
Let $H$ be a graph with $n$ vertices and no odd cycle shorter than $\delta n$ ($\delta$ may be a function of $n$).
Then there is an odd cycle cover $X$ of size  at most $(48/\delta) \ln (5/\delta)$
Moreover, $X$ can be found in polynomial time.
\end{theorem}

Let us show the following three options for an algorithm.
\begin{theorem} \label{thm:subexp}
  Let $G$ be a disk graph (resp. unit ball graph) with $n$ vertices. Let $\Delta$ be the maximum degree of $\overline{G}$ and $c$ the odd girth of $\overline{G}$ (they may be functions of $n$).
  \cli has a branching or can be solved, up to a polynomial factor, in time:\\
\begin{enumerate*}[label=(\roman*),itemjoin={\quad}]
\item $2^{\tilde{O}(n/\Delta)}$ (branching), \label{case:subexp-delta}
\item $2^{\tilde{O}(n/c)}$ (solved),		\label{case:subexp-oddgirth}
\item $2^{{O}(c  \Delta)}$ (solved). \label{case:subexp-both}
\end{enumerate*}
\end{theorem}
\begin{proof}
Let $G$ be the input disk graph (resp. unit ball graph) and let $\overline{G}$ be its complement, we look for a maximum independent set in $\overline{G}$. 

To prove \ref{case:subexp-delta}, consider a vertex $v$ of degree $\Delta$ in $\overline{G}$. We branch on $v$: either we include $v$ in our solution and remove $N[v]$, or discard $v$. The complexity is described by the recursion $F(n) \leqslant F(n-1) + F(n- (\Delta+1))$ and solving it gives \ref{case:subexp-delta}.
Observe that this does not give an algorithm running in time $2^{\tilde{O}(n/\Delta)}$ since the maximum degree might drop.
Therefore, we will do this branching as long as it is \emph{good enough} and then finish with the algorithms corresponding to \ref{case:subexp-oddgirth} and \ref{case:subexp-both}.

For \ref{case:subexp-oddgirth} and \ref{case:subexp-both}, let $C$ be a cycle of length $c$, it can be found in polynomial time (see for instance \cite{AlonYZ97}). By application of Theorem \ref{thm:occ} with $\delta = c/n$, we find an odd cycle cover $X$ in $\overline{G}$ of size $\tilde{O}(n/c)$ in polynomial time. Next we exhaustively guess in time $2^{\tilde{O}(n/c)}$ the intersection $I$ of an optimum solution with $X$ and finish by finding a maximum independent set in the bipartite graph $\overline{G}-(X \cup N(I))$, which can be done in polynomial time. The total complexity of this case is $2^{\tilde{O}(n/c)}$, which shows \ref{case:subexp-oddgirth}.

Finally, observe that the graph $\overline{G} - N[C]$ is bipartite, since otherwise $\overline{G}$ contains two vertex-disjoint odd cycles with no edges between them.
Moreover, since every vertex in $\overline{G}$ has degree at most $\Delta$, it holds that $|N[C]| \leqslant c  (\Delta-1) \leqslant c  \Delta$.
Indeed, a vertex of $C$ can only have $c  (\Delta-2)$ neighbors outside $C$. 
We can proceed as in the previous step: we exhaustively guess the intersection of the optimal solution with $N[C]$ and finish by finding the maximum independent set in a bipartite graph (a subgraph of $\overline{G}-N[C]$), which can be done in total time $2^{O(c  \Delta)}$, which shows \ref{case:subexp-both}.
\end{proof}

Now we show how the structure of $G$ affect the bounds in Theorem \ref{thm:subexp}.

\begin{corollary}\label{cor:subexp}
Let $G$ be a disk graph (resp. unit ball graph) with $n$ vertices. \cli can be solved in time:
\begin{compactenum}[(a)]
\item $2^{\tilde{O}(n^{2/3})}$,
\item $2^{\tilde{O}(\sqrt{n})}$ if the maximum degree of $\overline{G}$ is constant,
\item polynomial, if both the maximum degree and the odd girth of $\overline{G}$ are constant.
\end{compactenum}
\end{corollary}
\begin{proof}
We use the notation from Theorem \ref{thm:subexp}.
Both  $\Delta$ and $c$ can be computed in polynomial time (see e.g. \cite{AlonYZ97}).
  Therefore, knowing what is faster among cases \ref{case:subexp-delta}, \ref{case:subexp-oddgirth}, and \ref{case:subexp-both} is tractable.
  For case (a), while there is a vertex of degree at least $n^{1/3}$, we branch on it.
  When this process stops, we do what is more advantageous between cases \ref{case:subexp-oddgirth} and \ref{case:subexp-both}.
  Note that $\min (n/\Delta, n/c, c \Delta) \leqslant n^{2/3}$ (the equality is met for $\Delta = c = n^{1/3}$).
  For case (b), we do what is best between cases \ref{case:subexp-oddgirth} and \ref{case:subexp-both}.
  Note that $\min (n/c, c) \leqslant \sqrt{n}$ (the equality is met for $c = \sqrt{n})$.
Finally, case (c) follows directly from case \ref{case:subexp-both} in Theorem \ref{thm:subexp}.
\end{proof}
Observe that case (b) is typically the hardest one for \cli.
Moreover, the win-win strategy of Corollary \ref{cor:subexp} can be directly applied to solve \wcli, as finding a maximum weighted independent set in a bipartite graph is still polynomial-time solvable.
On the other hand, this approach cannot be easily adapted to obtain a subexponential algorithm for \textsc{Clique Partition} (even \textsc{Clique $p$-Partition} with constant $p$), since \textsc{List Coloring} (even \textsc{List $3$-Coloring}) has no subexponential algorithm for bipartite graphs, unless the ETH fails~(see \cite{precoloring}, the bound can be obtained if we start the reduction from a sparse instance of {\textsc{1-in-3-Sat} instead of {\textsc{Planar 1-in-3-Sat}).

\section{EPTAS for \mis on $\cl(d,\beta,i)$}\label{sec:eptas}

We start this section by showing that $\cl(d,\beta,1)$ has a randomized EPTAS.


\eptas*

\begin{proof}
Let $H$ be a graph in $\cl(d,\beta,1)$ with $n$ vertices and $I$ be a maximum independent set of $H$.
In particular, $|I| \geqslant \beta n$.
Since finding a maximum independent set in a bipartite graph can be done in polynomial time, we get the desired $(1-\varepsilon)$-approximation if we can find a subset of vertices $T$ such that: 
\begin{itemize}
\item $T$ is an \emph{odd-cycle transversal}, i.e., its removal yields a bipartite graph, and 
\item $|T \cap I|\leqslant \varepsilon |I|$.
\end{itemize}

At high level, our algorithm will thus select and remove some odd-cycle transversals $T$, and then apply the bipartite case algorithm.
We will do this at least once for a set $T$ that satisfies the second item, with some strong guarantee.
Of course the key ingredient in finding a suitable odd-cycle transversal is the fact that $\iocp(H) \leqslant 1$.
Indeed, this implies that for any odd cycle $C$, the set $N[C]$ is an odd-cycle transversal.

Let $c := 8(\frac{1}{(\beta \varepsilon)^2}+\frac{1}{\beta \varepsilon}+1) = O(1/\varepsilon^2)$, $\delta := \frac{\varepsilon}{c}=O(\varepsilon^3)$, and $s:=\frac{10 d}{\delta} \log \frac{1}{\delta}$.
We call \emph{short induced odd cycle} an induced odd cycle of length at most $c$, and \emph{long induced odd cycle} an induced odd cycle of length more than~$c$.
First we can assume that $\beta n$ is larger than $2s$, otherwise we can find an optimum solution by brute-force in time $2^n=2^{\tilde{O}(1/\varepsilon^3)}$.
Hence, $|I| > 2s$.

\begin{claim}\label{clm:vc-dim}
There exists a subset $S\subseteq I$ of size $s=\frac{10 d}{\delta} \log \frac{1}{\delta}$ such that $N(S)$ contains all vertices that have more than $\delta|I|$ neighbors in $I$. 
\end{claim}

\begin{proof}
  Let $A$ denote the set of vertices $v$ such that $|N(v) \cap I| \geqslant \delta |I|$.
  We define the hypergraph $K:=(I,\{N(v) \cap I, v \in A\})$.
  By assumption on $H$, the hypergraph $K$ has VC-dimension at most $d$.
  By the definition of $K$, all its edges have size at least $\delta |V(K)|$.
  A celebrated result in VC-dimension theory by Haussler and Welzl~\cite{HausslerW86}, later improved by Blumer \emph{et al.}~\cite{Blumer89}, ensures that every such a hypergraph $K$ admits a hitting set (a set of vertices that intersects every edge) of size at most $\frac{10 d}{\delta} \log \frac{1}{\delta}$.
\end{proof}

Algorithmically, we have two ways of selecting the set $S$, leading to a deterministic PTAS (see Theorem~\ref{thm:withoutBeta}) or a randomized EPTAS.
Either we run the rest of the algorithm for every subset of $V(H)$ of size $\frac{10 d}{\delta} \log \frac{1}{\delta}$ inducing an independent set (which constitutes $n^{f(\varepsilon)}$ possible sets), or we use another result proven in \cite{HausslerW86}: not only the hitting set exists but a uniform sample of $V(K)$ of size $\frac{10 d}{\delta} \log \frac{1}{\delta}$ is a hitting set with high probability.
So we do the following $t := \lceil \frac{\log(10^{-10})}{\log(1-(\beta/2)^s)} \rceil = 2^{\tilde{O}(1/\varepsilon^3)}$ many times: we select uniformly at random a set $S$ of size $s=\frac{10 d}{\delta} \log \frac{1}{\delta}$, and continue the rest of the algorithm if $S$ is an independent set.
Since $|I| > 2s$, it holds that $\Pr(S \subseteq I) > (\beta/2)^s$.
As we try out $t$ samples, at least one satisfies $S \subseteq I$ with probability at least $1-(1-(\beta/2)^s)^t \geqslant 1-10^{-10}$.\\
  
We now assume that the sample $S$ satisfies the properties of Claim~\ref{clm:vc-dim}.
We start by putting in $T$ all the vertices of $N(S)$ (note that no such vertex is in $I$ since $I$ is an independent set).
We define the graph $H':=H-N(S)$.
We got rid of the vertices with at least $\delta |I|$ neighbors in $I$: in $H'$, there are no such vertices anymore.
We want to find an odd-cycle transversal in $H'$ that has few vertices in $I$. 

We now run a polynomial-time algorithm (see for instance \cite{AlonYZ97}) that determines whether the graph is bipartite and, if not, outputs a shortest odd cycle $C_{\text{og}}=v_1v_2 \ldots v_g$ in $H'$.
\begin{itemize}
\item If $H'$ is bipartite, then $T:=N_H(S)$ is an odd-cycle transversal of $H$ with $|T \cap I|=0$.
\item If $g=|C_{\text{og}}| \leqslant c$, that is, if $C_{\text{og}}$ is a short induced odd cycle, then $|N_{H'}[C_{\text{og}}] \cap I| \leqslant c \delta |I| = \varepsilon|I|$, and therefore $T:=N_H(S)\cup N_{H'}[C_{\text{og}}]$ is an odd-cycle transversal of $H$ with $|T\cap I|\leqslant \varepsilon |I|$.
\end{itemize}

We can now safely assume that $g > c$, i.e., $C_{\text{og}}$ is a long induced odd cycle.
We decompose $H'$ into the successive neighborhoods of $C_{\text{og}}$, which we call \emph{layers}.
We define the first layer as $L_1 := N_{H'}(C_{\text{og}})$.
We define by induction the other layers as the non-empty sets $L_i := \{ v $ $|$ there exists $ u \in L_{i-1}$ with $uv \in E(H')$ and $v \notin L_j$ for $j<i \}$.
Let us denote by $\lambda$ the index of the last non-empty layer.
Before entering into the formal details of the second part of the proof let us briefly explain its structure:
\begin{itemize}
\item First, we observe that if there are many layers, there is one with index at most $\frac{2}{\beta \varepsilon}$ that contains at most $\frac{\varepsilon \beta}{2} n \leqslant \frac{\varepsilon}{2}|I|$ of the vertices. We can thus delete this layer, and note that connected components that do not contain $C_{\text{og}}$ are bipartite.
  We then focus on the component containing $C_{\text{og}}$, which has only few layers.
 \item Secondly, we show that this component admits an odd-cycle transversal of size at most $\frac{\varepsilon}{2} |I|$ (informally, the neighborhood at distance up to $O(\frac{1}{\varepsilon})$ of $O(\frac{1}{\varepsilon})$ consecutive vertices on the cycle $C_{\text{og}}$).
\end{itemize}

In other words, we can find $\varepsilon |I|$ vertices whose deletion yields a bipartite graph (see Figure~\ref{fig:large-successive-neighborhoods}), which together with $N(S)$ form the desired odd-cycle transversal.

\begin{figure}[!ht]
\centering
\begin{tikzpicture}[scale=0.5,vertex/.style={draw,circle,inner sep=-0.02cm}]
\def\s{41}
\def\l{6}
\def\ll{13}
\def\hd{0.2}
\def\vd{1}
\foreach \i in {1,...,\s}{
\node[vertex] (v\i) at (0,\i * \hd) {} ;
}
\foreach \i [count=\j from 1] in {2,...,\s}{
  \draw (v\i) -- (v\j) ;
}
\path (v1) edge[bend left=10] (v\s) ;
\foreach \j in {1,...,\ll}{
  \foreach \i in {1,...,\s}{
\draw (\vd * \j - \vd / 3, \i * \hd - \hd / 2) rectangle (\vd * \j + \vd / 3, \i * \hd + \hd / 2) ; 
}
}

\foreach \i in {8,...,16}{
  \node[vertex,fill=red] at (v\i) {} ; 
 \foreach \j in {1,...,\l}{
\draw[fill=red,opacity=0.5] (\vd * \j - \vd / 3, \i * \hd - \hd / 2) rectangle (\vd * \j + \vd / 3, \i * \hd + \hd / 2) ; 
}
}
\foreach \j in {7}{
  \foreach \i in {1,...,\s}{
\draw[fill=blue,opacity=0.5] (\vd * \j - \vd / 3, \i * \hd - \hd / 2) rectangle (\vd * \j + \vd / 3, \i * \hd + \hd / 2) ; 
}
}
\end{tikzpicture}
\caption{The layers (columns) and the strata (rows of a column).
  If the number of successive neighborhoods is large, a small cutset (in blue) is found among the first $\lceil \frac{2}{\beta \varepsilon} \rceil$ layers. To the right of this cutset, we know that the graph is bipartite. This brings us back to the case with fewer than $\frac{2}{\beta \varepsilon}$ layers, where we can find a small odd cycle transversal (in red).}
\label{fig:large-successive-neighborhoods}
\end{figure}
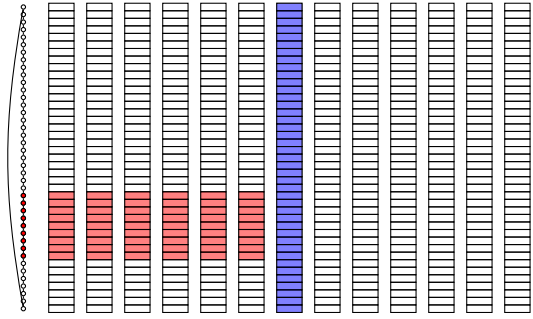

If $\lambda > \frac{2}{\beta \varepsilon}$, then there is some index $i \leqslant \lceil \frac{2}{\beta \varepsilon} \rceil$ such that $L_i$ is of size at most $\frac{\beta \varepsilon}{2} n \leqslant \frac{\varepsilon}{2} |I|$. 
We remove that layer $L_i$ from the graph.
Since $\iocp(H')=1$, the set $\bigcup_{i < j \leqslant \lambda}L_j$ induces a bipartite graph.
Indeed, it is disjoint from the closed neighborhood of the odd cycle $C_{\text{og}}$.
We can easily find a maximum independent set on this part of the graph, and focus on the other part, which is $C_{\text{og}} \cup \bigcup_{1 \leqslant j < i}L_j$.
We set $H'' := H'[C_{\text{og}} \cup \bigcup_{1 \leqslant j < i}L_j]$. If  $\lambda \leqslant \frac{2}{\beta \varepsilon}$, we set $H'':=H'$.

So the graph $H''$ has at most $\frac{2}{\beta \varepsilon}$ layers emanating from $C_{\text{og}}$.
We will find an odd-cycle transversal of size at most $\frac{\varepsilon}{2} |I|$.
We first need some new definitions.

For $1 \leqslant j \leqslant g$, let $S_j$ be the set of vertices $w \in V(H')$ such that there is a shortest path from $w$ to $C_{\text{og}}$ which ends in $v_j$, while no shortest path from $w$ to $C_{\text{og}}$ ends in $v_i$ with $i < j$ (note that $v_j \in S_j$).
We point out that the sets $(S_1,\ldots,S_g)$ induce a partition of each layer $L_k$.
This simply follows from the fact that for every vertex $w \in L_k$, there is a minimum index $j(w)$ such that there is a shortest path from $v$ to $C_{\text{og}}$ ending in $v_{j(w)}$.
For each pair $(k,\ell)$, we define a \emph{stratum} as $L_k^\ell := S_\ell \cap L_k$.
Note that if $L_k^\ell = \emptyset$, then for any $k'>k$, $L_{k'}^\ell = \emptyset$.

Let $z := \lceil \frac{4}{\beta \varepsilon} \rceil + 2$ and for any integer $\gamma$ such that $0 \leqslant \gamma \leqslant \frac{g}{z}-1$, let $S^\gamma := \underset{\gamma z+1 \leqslant j \leqslant (\gamma+1)z}{\bigcup}S_j$.
Informally, $S^\gamma$ consists of the layers emanating from $z$ consecutive vertices of $C_{\text{og}}$.
Note that if $\gamma \neq \gamma'$, then $S^{\gamma}$ and $S^{\gamma'}$ are disjoint.

\begin{claim}\label{clm:bip}
  For any non-negative integer $\gamma \leqslant \frac{2}{\beta \varepsilon}$, the graph $B := H''- S^\gamma$ is bipartite.
\end{claim}
\begin{proof}
  Observe that $\frac{g}{z}-1 \geqslant \frac{2}{\beta \varepsilon}$, so each $S^\gamma$ of the claim is well-defined.
  It holds that $C_{\text{og}} \cap S^\gamma = \bigcup_{\gamma z+1 \leqslant j \leqslant (\gamma+1)z}\{v_j\}$.
  We exhibit a proper $2$-coloring of $B$ by coloring its vertices as follows.
  We start by coloring each vertex of the path $C_{\text{og}} \setminus S^\gamma$ in an alternated fashion, i.e., one endpoint of the path gets color 0, its neighbor gets color 1, the next vertex gets color 0, and so on. 

  For each pair $(k,\ell)$ such that $1 \leqslant k < i$, $1 \leqslant \ell \leqslant g$, and $\ell \notin [\gamma z + 1, (\gamma+1)z]$, we color all the vertices in the stratum $L_k^\ell = S_\ell \cap L_k$ with the opposite color of the one used for the stratum $L_{k-1}^\ell = S_\ell \cap L_{k-1}$ (with the convention that $L_0 := C_{\text{og}}$).
  This process colors unambiguously all the vertices of $B$.

  Let us prove that the resulting coloring is proper.
  First note that the vertices of a same stratum form an independent set.
  Indeed, assume by contradiction that $L_k^\ell$ contains an edge $uw$.
  There is a shortest path $P_1$ from $v_\ell$ to $u$ and a path $P_2$ from $w$ to $v_\ell$.
  Since $u$ and $w$ are in the same layer $L_k$, $P_1$ and $P_2$ have the same length; more precisely, $|P_1|=|P_2|=k < i \leqslant \frac{2}{\beta \varepsilon}$.
  Thus, $P_1, uw, P_2$ defines a closed walk with $2k+1$ edges.
  An odd closed walk of length $2k+1$ implies the existence of an odd induced cycle of length at most $2k+1$.
  As $2k+1 < \frac{4}{\beta \varepsilon}+1 < g$, we reach a contradiction on the minimality of $C_{\text{og}}$.
  
  There is no edge between a stratum $L_k^\ell$ and a stratum $L_{k'}^{\ell'}$ with $|k-k'| \geqslant 2$, by definition of the layers.
  Moreover, for $1 \leqslant \ell < \ell' \leqslant g$, there is no edge $uw$ with $u \in L_k^\ell$ and $w \in L_{k'}^{\ell'}$ with $\min(\ell'-\ell,\ell+g-\ell') \geqslant \frac{4}{\beta \varepsilon}+1$ since otherwise it would be possible to construct an odd cycle strictly shorter than $C_{\text{og}}$.
  Indeed, if $P_1$ is a shortest path between $v_\ell$ and $u$ and $P_2$ is a shortest path between $w$ and $v_{\ell'}$,
  then $P_1, uw, P_2$ is a walk of length $k+k'+1$.
  However, a shortest path between $v_\ell$ and $v_{\ell'}$ within $C_{\text{og}}$ has length $\min(\ell'-\ell,\ell+g-\ell') \geqslant \frac{4}{\beta \varepsilon}+1 > k+k'+1$.
  Hence, the walk $P_1, uw, P_2$ can be extended into an odd closed walk of length strictly smaller than $g$, by taking the path from $v_{\ell'}$ to $v_\ell$ in $C_{\text{og}}$ with the same parity as $k+k'$; a contradiction.

  Therefore, if there is a monochromatic edge $uw$ in $B$, it must be between $L_k^\ell$ and $L_{k'}^{\ell'}$ with $|k-k'| \in \{0,1\}$, $\min(\ell'-\ell,\ell+g-\ell') < \frac{4}{\beta \varepsilon}+1$, and $\{\ell,\ell'\} \cap [\gamma z + 1, (\gamma+1) z] = \emptyset$.
  We fix $k, k', \ell, \ell'$ satisfying those conditions.
  We call \emph{small interval of $\ell$ and $\ell'$}, denoted by $\text{si}(\ell,\ell')$, the integer interval $[\ell,\ell']$ if $\min(\ell'-\ell,\ell+g-\ell')=\ell'-\ell$ and $[\ell',g] \cup [1,\ell]$ if $\min(\ell'-\ell,\ell+g-\ell')=\ell+g-\ell'$.
  What we showed in the previous paragraph implies that if $\{\ell,\ell'\} \cap [\gamma z + 1, (\gamma+1) z] = \emptyset$, then $\text{si}(\ell,\ell') \cap [\gamma z + 1, (\gamma+1) z] = \emptyset$.
  Indeed, the small interval of $\ell$ and $\ell'$ is a circular interval over $[1,g]$ of length less than $\frac{4}{\beta \varepsilon}+1 < z$.
  In particular, the vertices of $C_{\text{og}}$ indexed by the small interval of $\ell$ and $\ell'$ are all in $B$.
  
  Assume first that $k=k'$.
  There is a path $P_1$ from $v_\ell$ to $u$, and a path $P_2$ from $w$ to $v_{\ell'}$, both of length $k$.
  Since by assumption the color for $L_k^\ell$ is the same as the color for $L_{k'}^{\ell'}$, the vertices $v_\ell$ and $v_\ell'$ have the same color (by construction of the 2-coloring).
  Thus we have, in $C_{\text{og}}-S^\gamma$, a path $P$ indexed by $\text{si}(\ell,\ell')$ from $v_{\ell'}$ to $v_\ell$ of even length less than $\frac{4}{\beta \varepsilon}+1$.
  We emphasize that it is crucial that $\text{si}(\ell,\ell') \cap [\gamma z + 1, (\gamma+1) z] = \emptyset$ (meaning that all the vertices of $C_{\text{og}}$ indexed by $\text{si}(\ell,\ell')$ are still in $B$), to deduce that there is a path of even length between $v_\ell$ and $v_{\ell'}$. It follows from the mere fact that $v_\ell$ and $v_\ell'$ have the same color.   
  Finally, the concatenation $P_1, uw, P_2, P$ yields an odd cycle of  length less than $2k+1+\frac{4}{\beta \varepsilon}+1 \leqslant \frac{8}{\beta \varepsilon}+2 < g$.

  Now let us assume that $|k-k'|=1$; say, without loss of generality, $k'=k+1$.
  In that case, by construction of the 2-coloring, the edge can only be monochromatic if $v_\ell$ and $v_{\ell'}$ receive distinct colors.
  Furthermore, there is a path $P_1$ from $v_\ell$ to $u$, and a path $P_2$ from $w$ to $v_{\ell'}$ with length of distinct parities ($k$ and $k+1$, respectively).
  Moreover, since $v_\ell$ and $v_{\ell'}$ get distinct colors, there is in $C_{\text{og}}-S^\gamma$, a path $P$ indexed by $\text{si}(\ell,\ell')$ from $v_{\ell'}$ to $v_\ell$ of odd length at most $\frac{4}{\beta \varepsilon}+1$.
  Again, we crucially use that all the vertices of $C_{\text{og}}$ indexed by $\text{si}(\ell,\ell')$ are in $B$, to deduce that $P$ is of odd length from the fact that $v_\ell$ and $v_{\ell'}$ get distinct colors. 
  Finally, the concatenation $P_1, uw, P_2, P$ is an odd cycle of length less than $k+1+k+1+\frac{4}{\beta \varepsilon}+1 \leqslant \frac{8}{\beta \varepsilon}+3 < g$; a contradiction.

We conclude that the $2$-coloring is indeed proper. 
\end{proof}

\begin{algorithm}
  \caption{EPTAS for \mis on $\cl(d,\beta,1)$}
    \label{alg:eptas}
  \begin{algorithmic}[1]
    \Require{$H$ satisfies $d := \vcdim(G)=O(1)$, $\alpha(G) \geqslant \beta |V(G)|$, and $\iocp(G) \leqslant 1$}
    \Function{Stable}{$H,\varepsilon$}:
    \Let{$c$}{$8(1/(\beta \varepsilon)^2+1/(\beta \varepsilon)+1)$}
    \Let{$\delta$}{$\frac{\varepsilon}{c}$}
    \Let{$s$}{$\frac{10d}{\delta} \log{\frac{1}{\delta}}$}
    \If{$\beta |V(H)| < 2s$}
       solve $H$ optimally by brute-force                          \Comment{$|V(H)| = \tilde{O}(1/\varepsilon^3)$}
    \EndIf
      \For{$\_ \gets 1 \textrm{ to } t=2^{\tilde O(1/\varepsilon^3)}$}{
        \Let{$S$}{\text{uniform sample of $V(G)$ of size $s$}}  \Comment{$S \subseteq I$ with probability $> (\frac{\beta}{2})^s$}
         \If{$G[S]$ contains an edge} break and go to the next iteration \EndIf  
         \Let{$H'$}{$H-N[S]$}                                    \Comment{remove $S$ and its neighborhood}  
         \Let{$C_{\text{og}}$}{shortest odd cycle in $H'$}       \Comment{in polynomial time \cite{AlonYZ97}}
         \Let{$g$}{$|C_{\text{og}}|$}                            \Comment{$C_{\text{og}}=v_1v_2 \ldots v_g$}
         \If{$g \leqslant c$}                  \Comment{short induced odd cycle}
         \Let{$S$}{$S~\cup~$ max stable on the bipartite $H'-N[C_{\text{og}}]$} \Comment{$\iocp(G)=1$}
         \Else{}                                              \Comment{long induced odd cycle}
           \Let{$L_\ell$}{vertices of $H'$ at distance exactly $\ell$ from $C_{\text{og}}$}   
           \Let{$L_i$}{smallest layer among $\{L_\ell\}_{1 \leqslant \ell \leqslant \lceil 2/{\beta \varepsilon} \rceil}$} \Comment{$|L_i| \leqslant \frac{\varepsilon}{2}\alpha(G)$}
           \Let{$H''$}{$H'[C_{\text{og}} \cup \bigcup_{1 \leqslant j < i}L_j]$}   
           \Let{$S_k$}{vertices of $H''$ whose closest vertex on $C_{\text{og}}$ of minimum index is $v_k$}   
           \Let{$z$}{$\lceil \frac{4}{\beta \varepsilon} \rceil + 2$}
           \Let{$S^\gamma$}{smallest set among $\{\bigcup_{k \in [\gamma z+1,(\gamma + 1)z]}S_k\}_{\gamma \in [0,\lfloor \frac{2}{\beta \varepsilon} \rfloor]}$}                \Comment{$|S^\gamma| \leqslant \frac{\varepsilon}{2}\alpha(G)$}
           \Let{$S$}{$S~\cup~$ max stable on the bipartite $H'[\bigcup_{j > i}L_j]$}    \Comment{$\iocp(G)=1$}
           \Let{$S$}{$S~\cup~$ max stable on the bipartite $H''-S^\gamma$}              \Comment{Claim~\ref{clm:bip}}
         \EndIf
      }
      \EndFor
      \State \Return{$S$ at the iteration maximizing its cardinality}
      \EndFunction
   \Ensure{output $S$ is a stable set of size at least $(1-\varepsilon)\alpha(G)$ with high probability}
  \end{algorithmic}
\end{algorithm}

Since the sets of $\{S^\gamma\}_{\gamma \in [0, \lfloor \frac{2}{\beta \varepsilon} \rfloor]}$ are pairwise disjoint, a smallest set of the collection satisfies $|S^\gamma| \leqslant \frac{\beta \varepsilon}{2}n \leqslant \frac{\varepsilon}{2} |I|$.
By Claim~\ref{clm:bip}, removing this $S^\gamma$ from $H''$ makes the graph bipartite.
We finally compute a maximum independent set in polynomial time in $H''- S^\gamma$.
We return the best solution found.
The pseudo-code is detailed in Algorithm~\ref{alg:eptas}.
\end{proof}

\disk*

\begin{proof}
  With a geometric representation, we can invoke the following argument to get a linear maximum stable set. 
  Recall that the piercing number of a family of geometric objects is the minimum number of points such that each object contains at least one of those points.
  The piercing number of a collection of pairwise intersecting disks in the plane is $4$ \cite{stacho,danzer,Carmi18}.
  The number of faces in an arrangement of $n$ circles (disk boundaries) is $O(n^2)$, and all the points within one face hit the same disks.
  In time $O(n^8)$, one can therefore exhaustively guess four points piercing a maximum clique $\mathcal C$.
  We can remove all the disks which are not hit by any of those four points, since they are not part of $\mathcal C$.
  This new instance $G$ can have its vertices partitioned into four cliques, hence $\alpha(\overline{G}) \geqslant |V(G)|/4$.

  Without a geometric representation, we suggest the following.
  Disk graphs always contain a vertex whose neighborhood has independence number at most $6$ (think of a vertex which has the smallest radius in one representation).
  For each vertex $v$, we run the robust PTAS of Chan and Har-Peled \cite{ChanH12} for \mis in $G[N(v)]$ with ratio strictly larger than $6/7$ (say $7/8$).
  By \emph{robust} we mean that their local-search based algorithm does not require a geometric representation.
  By the previous remark, at least one run has to report a value of at most $6$ (indeed, $7(7/8) > 6$).
  Let $u$ be a vertex corresponding to such a run.
  For every disk graph $G$, $\chi(G) \leqslant 6 \omega(G)$ (actually a better bound of $6\omega(G)-6$ is known \cite{MalesiskaPW98}).
  Then $\alpha(G) \omega(G) \geqslant \frac{\chi(G) \alpha(G)}{6} \geqslant \frac{|V(G)|}{6}$, hence $\omega(G) \geqslant  \frac{|V(G)|}{6 \alpha(G)}$.
  
  We branch on two outcomes.
  Either $u$ is in a maximum clique: we run the approximation of Theorem~\ref{thm:eptas} on $G_u := \overline{G[N[u]]}$ which satisfies $\omega(\overline{G_u})=\alpha(G_u) \geqslant |V(G_u)|/36$ (recall that we chose $u$ so that $\alpha(\overline{G_u}) \leqslant 6$).
  Or this vertex is not in any maximum clique: we delete it from the graph.
  Our branching tree has size $2n+1$, so it only costs an extra linear multiplicative factor.

  The VC-dimension of the neighborhoods of disk graphs, and even pseudo-disk graphs \cite{Aronov18}, is at most~$4$.
  Since the VC-dimension of a graph is equal to the one of its complement, the VC-dimension of $\overline G$ is also at most $4$.
  Finally, by Theorem~\ref{thm:main-structural-non-disk}, $\iocp(\overline{G}) \leqslant 1$.
  We only call the approximation algorithm (a polynomial number of times) with disk graphs $G'$ such that $\overline{G'} \in \cl(4,\frac{1}{36},1)$ (argument without the geometric representation) or $\overline{G'} \in \cl(4,\frac{1}{4},1)$ (argument with the geometric representation).
  Hence, we conclude by Theorem~\ref{thm:eptas}.
\end{proof}


It is folklore that unit ball graphs have geometric VC-dimension $4$.
One can observe that, in the case of \emph{unit} ball graphs, the geometric VC-dimension coincides with the VC-dimension of the neighborhoods.
At the price of a multiplicative factor $n$ in the running time, one can guess a vertex $v$ in a maximum clique of a unit ball graph $G$, and look for a clique in its neighborhood $H := G[N(v)]$.
As the kissing number for unit spheres is bounded (it is $12$), one can also show that the neighborhood of this vertex (in fact, of any vertex) can be partitioned into a constant number of cliques.
The exact number is irrelevant here: An easy volume-based argument can show that this number is no greater than 30. 
Thus $\alpha(\overline H) \geqslant |V(H)|/30$.
Therefore, from Theorem~\ref{thm:eptas} and Theorem~\ref{thm:no-two-odd-cycles-ubg}, we immediately obtain the following.

\ball*

We can extend the EPTAS to work for constant (not necessarily~$1$) induced odd cycle packing number.

\begin{theorem}\label{thm:eptas-bis}
  For any constants $d, i \in \mathbb N$, $0 < \beta \leqslant 1$, for every $\varepsilon > 0$, there is a randomized $(1-\varepsilon)$-approximation algorithm running in time $2^{\tilde{O}({1/\varepsilon}^3)}n^{O(1)}$ for \mis on graphs of $\cl(d,\beta,i)$ with $n$ vertices.
\end{theorem}
\begin{proof}
  Let $I$ be a maximum independent set.
  We show by induction on $i$ that for any $\varepsilon'$ we can find in time $2^{\tilde{O}({1/\varepsilon'}^3)}n^{O(1)}$ a stable set of size $(1-i\varepsilon')|I|$.
  The base case is Theorem~\ref{thm:eptas}.
  We assume that there is such an algorithm when the induced odd cycle packing number is $i-1$.
  We follow Algorithm~\ref{alg:eptas} with a graph $H$ such that $\iocp(H)=i$.
  On line 15 and 23, the graph is not necessarily bipartite anymore (on line 24, the resulting graph is still bipartite in this case).
  Although, the induced odd cycle packing number is decreased to $i-1$.
  So, by the induction hypothesis we get a stable set within a factor $(1-(i-1)\varepsilon')$ of the optimum.
  To get there, we removed a subset of vertices of size at most $\varepsilon'|I|$.
  Therefore, the solution $S$ that we obtain satisfies $|S| \geqslant (1-i\varepsilon')|I|$.

  We obtain the theorem by setting $\varepsilon := i \varepsilon'$ since $i$ is absorbed in the $\tilde{O}$ in the running time.
\end{proof}

To the detriment of the efficiency of the approximation scheme, when $\iocp(G) \leqslant 1$ (or even $\iocp(G)=O(1)$), we can discard one of the two other assumptions of Theorem~\ref{thm:eptas}.
Namely, we do not need bounded VC-dimension or that the optimum solution is a positive fraction of the number of vertices.

\begin{theorem}\label{thm:withoutVCdim}
  There is a randomized PTAS for \mis on graphs of $\cl(\infty,\beta,1)$.
\end{theorem}
\begin{proof}
  We observe that $\vcdim(G)$ is always at most $\log |V(G)|=\log n$.
  So following Algorithm~\ref{alg:eptas}, we now sample a set $S$ of size $\frac{10 \log n}{\delta}\log{\frac{1}{\delta}}$.
  The probability for $S$ to be successfully contained in an optimum solution $I$ is at least $\frac{1}{n^{\Tilde{O}(1/\varepsilon^3)}}$.
  Thus by repeating this experience $n^{\Tilde{O}(1/\varepsilon^3)}$ times, $S \subseteq I$ holds in at least one branch, with high probability.
  The rest of the algorithm is unchanged.
\end{proof}

\begin{theorem}\label{thm:withoutBeta}
  There is a deterministic PTAS for \mis on graphs of $\cl(d,0,1)$.
\end{theorem}
\begin{proof}
  Now, instead of sampling $s = \frac{10 d}{\delta}\log{\frac{1}{\delta}}$ vertices, we try all the $n^s=n^{f(\varepsilon)}$ subsets of size $s$.
  One of them falls entirely in an optimum solution $I$ and is a desired $\varepsilon$-net (note that if the largest independent set contains less than $s$ vertices, we can find it in time $O(n^s)$).
  We also change line 18 and 22 of Algorithm~\ref{alg:eptas}: Instead of deleting the lightest layer among the first $O(\frac{1}{\varepsilon})$, and the lightest block of strata among $O(\frac{1}{\varepsilon})$ disjoint blocks, we try all possible such pairs.
  This only adds a multiplicative factor in $O(1/\varepsilon^2)$ to the running time.
  One of the deleted pair contains less than $\varepsilon |I|$ vertices of $I$, and we conclude similarly.
\end{proof}

Both results generalize to $\cl(\infty,\beta,i)$ and $\cl(d,0,i)$, for every integer $i$, with the same arguments as in the proof of Theorem~\ref{thm:eptas-bis}.
  
\section{Other intersection graphs}\label{sec:gen&lim}

In this section, we discuss the impossibility of generalizing our algorithms to higher dimensions and related classes of intersection graphs.
We first show some hardness of approximation for \textsc{Maximum Independent Set} on the class of all the 2-subdivisions, hence the same lower bound for \textsc{Maximum Clique} on all the co-2-subdivisions.

It is folklore that from the PCP of Moshkovitz and Raz \cite{Moshkovitz10}, which roughly implies that \textsc{Max 3-SAT} cannot be $(7/8+\varepsilon)$-approximated in subexponential time under the ETH, one can derive such inapproximability in subexponential time for many hard graph and hypergraph problems; see for instance \cite{Bonnet15}.
The following inapproximability result for \mis on bounded-degree graphs was shown by Chleb\'ik and Chleb\'ikov\'a \cite{Chlebik06}.
As their reduction is almost linear, the PCP of Moshkovitz and Raz boosts this hardness result from ruling out polynomial-time up to ruling out subexponential time $2^{n^\gamma}$ for any $\gamma < 1$. 
\begin{theorem}[\cite{Chlebik06,Moshkovitz10}]\label{thm:inapprox-mis}
For any $\Delta \geqslant 3$, and $\gamma < 1$, there is a constant $\eta < 1$ such that \mis on graphs with $n$ vertices and maximum degree $\Delta$ cannot be $\eta$-approximated in time $2^{n^\gamma}$, unless the ETH fails. 
\end{theorem}
We could actually state a slightly stronger statement for the running time but will settle for this for the sake of clarity.

\begin{theorem}\label{thm:hardness-2subd}
  For any $\gamma < 1$, there is a constant $\zeta < 1$ such that \textsc{Maximum Independent Set} on the class of all the 2-subdivisions has no $\zeta$-approximation algorithm running in subexponential time $2^{n^\gamma}$, unless the ETH fails.
\end{theorem}
\begin{proof}
Let $G$ be a graph with maximum degree a constant $\Delta \geqslant 3$, with $n$ vertices $v_1, \ldots, v_n$ and $m$ edges $e_1, \ldots, e_m$, and let $H$ be its 2-subdivision.
Recall that to form $H$, we subdivided every edge of $G$ exactly twice.
These $2m$ vertices in $V(H) \setminus V(G)$, representing edges, are called \emph{edge vertices} and are denoted by $v^+(e_1), v^-(e_1), \ldots, v^+(e_m), v^-(e_m)$, as opposed to the other vertices of $H$, which we call \emph{original vertices}.
If $e_k=v_iv_j$ is an edge of $G$, then $v^+(e_k)$ (resp. $v^-(e_k)$) has two neighbors: $v^-(e_k)$ and $v_i$ (resp. $v^+(e_k)$ and $v_j$). 

Observe that there is a maximum independent set $S$ which contains exactly one of $v^+(e_k), v^-(e_k)$ for every $k \in [m]$. Indeed, $S$ cannot contain both $v^+(e_k)$ and $v^-(e_k)$ since they are adjacent. On the other hand, if $S$ contains neither $v^+(e_k)$ nor $v^-(e_k)$, then adding $v^+(e_k)$ to $S$ and potentially removing the other neighbor of $v^+(e_k)$ which is $v_i$ (with $e_k=v_iv_j$) can only increase the size of the independent set.
Hence $S$ contains $m$ edge vertices and $s \leqslant n$ original vertices, and there is no larger independent set in $H$.

We observe that the $s$ original vertices in $S$ form an independent set in $G$.
Indeed, if  $v_iv_j=e_k \in E(G)$ and $v_i,v_j \in S$, then neither $v^+(e_k)$ nor $v^-(e_k)$ could be in $S$.
  
Now, assume there is an approximation with ratio $\zeta := (1+\frac{2(1-\eta)}{\eta(\Delta+1)^2})^{-1}$ for \textsc{Maximum Independent Set} on 2-subdivisions running in subexponential time, where $\eta < 1$ is a ratio which is not attainable for \textsc{Maximum Independent Set} on graphs of maximum degree $\Delta$ according to Theorem~\ref{thm:inapprox-mis}.
On instance $H$, this algorithm would output a solution with $m'$ edge vertices and $s'$ original vertices.
As we already observed this solution can be easily (in polynomial time) transformed into an at-least-as-good solution with $m$ edge vertices and $s''$ original vertices forming an independent set in $G$. Further, we may assume that  $s'' \geqslant n / (\Delta+1)$ since for any independent set of $G$, we can obtain an independent set of $H$ consisting of the same set of original vertices and $m$ edge vertices. 
Since $m \leqslant n \Delta / 2$ and $s'' \geqslant n / (\Delta+1)$, we obtain $m \leqslant s'' \Delta(\Delta+1)/2$ and $2m/(\Delta+1)^2 \leqslant s''\Delta /(\Delta+1)$.
From $\frac{m+s''}{m+s} \geqslant \zeta$ and $\Delta \geqslant 3$, we have 
\[ s	\leqslant m\cdot \frac{2(1-\eta)}{\eta(\Delta+1)^2} + s''\cdot (1+ \frac{2(1-\eta)}{\eta(\Delta+1)^2})
	\leqslant s'' (\frac{(1-\eta)\Delta}{\eta(\Delta+1)} + 1 + \frac{2(1-\eta)}{\eta(\Delta+1)^2} )
	\leqslant s''(1 + \frac{1-\eta}{\eta})  = \frac{s''}{\eta}
\]
  This contradicts the inapproximability of Theorem~\ref{thm:inapprox-mis}.
Indeed, note that the number of vertices of $H$ is only a constant times the number of vertices of $G$ (recall that $G$ has bounded maximum degree, hence $m=O(n)$).
\end{proof}

We get the following as a direct corollary.
\begin{corollary}\label{cor:hardness-co2subd}
  For every $\gamma < 1$, there is a constant $\zeta < 1$, \cli on the class of all the co-2-subdivisions has no $\zeta$-approximation algorithm running in subexponential time $2^{n^\gamma}$, unless the ETH fails.
\end{corollary}

For exact algorithms, the subexponential time that we rule out under the ETH is not only $2^{n^\gamma}$ (with $\gamma < 1$) but actually any $2^{o(n)}$.

\subsection{Balls and higher dimensions}\label{subsec:higher}

In sharp contrast to the algorithms for disk and unit ball graphs, we will prove that with an extra dimension or with different radii (even arbitrarily close to each other), a PTAS is highly unlikely.

By Corollary~\ref{cor:hardness-co2subd}, we just need to show that all co-2-subdivisions can be realized by our geometric objects. 
This appears like a simple and powerful method to rule out a PTAS (QPTAS, and even SUBEXPAS) for a geometric clique problem.

The co-2-subdivision of a graph with $n$ vertices and $m$ edges can be thought of as follows.
It is made of a clique on $n$ vertices, representing the initial vertices, and a clique on $2m$ vertices minus a perfect matching, representing endpoints of the initial edges.
Each anti-matched pair of vertices corresponds to an edge in the initial graph.
Each vertex representing one endpoint of an initial edge is adjacent to all the vertices representing the initial vertices but this endpoint. 

\begin{theorem}\label{thm:4udg}
  The class of 4-dimensional unit ball graphs contains all the co-2-subdivisions.
\end{theorem}
\begin{proof}
  Given any simple graph $G=(V,E)$ with $n$ vertices and $m$ edges, we want to build a set $S$ of $n+2m$ points in $\mathbb R^4$ where each vertex $v$ is represented by a point $p(v)$ and each edge $e$, by two points $p^+(e)$ and $p^-(e)$.
  A pair of points in $S$ should be at distance at most 2, except if it is the two points of the same edge $e$, or if it is a point $p(v)$ with either a $p^+(vw)$ or a $p^-(wv)$; those pairs should be at distance strictly more than 2.
  We denote by $x,y,z,t$ the coordinates of $\mathbb R^4$.
  Let $\mathcal P$ be the plane defined by the intersection of the hyperplanes of equation $z=0$ and $t=0$.
  The projection $\pi$ onto $\mathcal P$ of the points $p^+(e)$ (resp. $p^-(e)$) fall regularly on the ``top'' part (resp. ``bottom'' part) of a circle $\mathcal C$ of $\mathcal P$ centered at the origin $O=(0,0,0,0)$ and of diameter $2$, such that for each edge $e$, $\pi(p^+(e))$ and $\pi(p^-(e))$ are antipodal on $\mathcal C$.
  We just defined the points $\pi(p^+(e))$ and $\pi(p^-(e))$.
  The actual points $p^+(e)$ and $p^-(e)$ will be fixed later by moving them very slightly away from their projection in a two-dimensional plane orthogonal to $\mathcal P$. 
  Let $\eta$ be the maximum distance between a pair $(\pi(p^+(e)),\pi(p^-(e')))$ with $e \neq e'$.
  By construction $\eta<2$.
  
  Let $\mathcal P^\bot$ be the (2-dimensional) plane containing the center $O$ of $\mathcal C$ and orthogonal to $\mathcal P$; in other words, the intersection of the hyperplanes of equation $x=0$ and $y=0$.
  We observe that all the points of $\mathcal P^\bot$ are equidistant to all the points $\pi(p^+(e))$ and $\pi(p^-(e))$.
  We place all the points $p(v)$ in $\mathcal P^\bot$ regularly spaced on a arc of a circle lying on $\mathcal P^\bot$ centered at $O$ and of radius $\sqrt{3}-\varepsilon$.
  One can notice that for any $(v, e, s) \in V \times E \times \{+,-\}$, $d(p(v),\pi(p^s(e))) = \sqrt{4-2 \sqrt 3 \varepsilon + \varepsilon^2}$.
  We will choose $\varepsilon \ll 2 - \eta \ll 1$ so that this shared distance is just below 2, and the points $\{p^+(e),p^-(e)\}_{e \in E}$ realize the same adjacencies than their projection by $\pi$.

  For every $e=uv \in E$, we place $p^+(e)$ such that $\overrightarrow{\pi(p^+(e))p^+(e)}=-\frac{(\varepsilon+\varepsilon')}{\lVert Op(u) \rVert}\overrightarrow{Op(u)}$ and $p^-(e)$ such that $\overrightarrow{\pi(p^-(e))p^-(e)}=-\frac{(\varepsilon+\varepsilon')}{\lVert Op(v) \rVert}\overrightarrow{Op(v)}$.
  In words, we push very slightly $p^+(e)$ (resp. $p^-(e)$) away from $\pi(p^+(e))$ (resp. $\pi(p^-(e))$) in the opposite direction of $\overrightarrow{Op(u)}$ (resp. $\overrightarrow{Op(v)}$).
  This way, $p^+(e)$ is at distance more than 2 from $p(u)$.
  Indeed, $$d(p^+(e),p(u))=\lVert \overrightarrow{p^+(e)\pi(p^+(e))} + \overrightarrow{\pi(p^+(e))O} + \overrightarrow{Op(u)} \rVert = \lVert \overrightarrow{\pi(p^+(e))O} + \overrightarrow{Op(u)} + \overrightarrow{p^+(e)\pi(p^+(e))} \rVert$$ $$= \lVert \overrightarrow{\pi(p^+(e))O} + (\sqrt 3 -\varepsilon+\varepsilon+\varepsilon')\frac{\overrightarrow{Op(u)}}{\lVert \overrightarrow{Op(u)}\rVert} \rVert = \sqrt{1+(\sqrt 3+\varepsilon')^2} > 2.$$
  Similarly, $d(p^-(e),p(v))>2$.
  We choose $\varepsilon'$ infinitesimal (in particular, $\varepsilon' \ll \varepsilon$) so that, still, $d(p^+(e),p(w))<2$ for any $w \neq u$ and $d(p^-(e),p(w))<2$ for any $w \neq v$.
  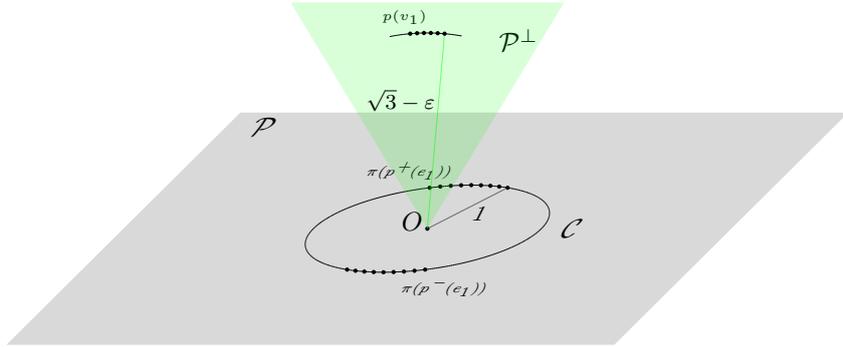
\begin{figure}
    \centering
    \begin{tikzpicture}[
        dot/.style={fill,circle,inner sep=-0.02cm},
      ]
      \def\xs{0.5}
      \coordinate (a) at (-4,0,-4) {} ;
      \coordinate (b) at (4,0,-4) {} ;
      \coordinate (c) at (4,0,4) {} ;
      \coordinate (d) at (-4,0,4) {} ;
      \fill[opacity=0.15] (a) -- (b) -- (c) -- (d) -- cycle ;
      \coordinate (O) at (0,0,0) {} ;
      \node[dot] at (O) {} ;
      \node at (-0.2,0.1,0) {$O$} ;
      \node[xslant=\xs]  at (-3.5,0,-3.5) {$\mathcal P$} ;
      \node[xslant=\xs]  at (1.9,0,0) {$\mathcal C$} ;
      
      \def \r{1.5}
      \foreach \i in {70,75,80,85,90,95,100,105,110}{
        \pgfmathsetmacro{\CosValue}{\r * cos(\i)}
        \pgfmathsetmacro{\SinValue}{\r * sin(\i)}
        \node[dot] (pm\i) at (\CosValue,0,\SinValue) {};
        \node[dot] (pp\i) at (\CosValue,0,-\SinValue) {};
      }
      \foreach \i [count=\j from 2] in {1,...,360}{
        \pgfmathsetmacro{\CosValue}{\r * cos(\i)}
        \pgfmathsetmacro{\SinValue}{\r * sin(\i)}
        \pgfmathsetmacro{\CosValueb}{\r * cos(\j)}
        \pgfmathsetmacro{\SinValueb}{\r * sin(\j)}
        \draw[very thin] (\CosValue,0,\SinValue) -- (\CosValueb,0,\SinValueb) {};
      }
      \node[xslant=\xs] (pp1) at (-1,0,-2) {\tiny{$\pi(p^+(e_1))$}} ;
      \node[xslant=\xs] (pm1) at (1,0,2) {\tiny{$\pi(p^-(e_1))$}} ;
      \draw[opacity=0.5] (pp70) -- (O) ;
      \node[xslant=\xs] at (0.5,0,-0.5) {\footnotesize{1}} ;

      \coordinate (e) at (-1.8,3,0) {} ;
      \coordinate (f) at (1.8,3,0) {} ;
      \node at (1.2,2.5,0) {$\mathcal P^{\bot}$} ;
      \node at (-0.3,2.8,0) {\tiny{$p(v_1)$}} ;

      \fill[green,opacity=0.15] (O) -- (e) -- (f) -- cycle ;

      \foreach \i in {85,87,89,91,93,95}{
        \pgfmathsetmacro{\CosValue}{1.73 * \r * cos(\i)}
        \pgfmathsetmacro{\SinValue}{1.73 * \r * sin(\i)}
        \node[dot] (q\i) at (\CosValue,\SinValue,0) {};
      }
      \draw[green,opacity=0.5] (q85) -- (O) ;
      \node at (-0.35,1.7,0) {\footnotesize{$\sqrt{3}-\varepsilon$}} ;
      
      \foreach \i [count=\j from 81] in {80,...,100}{
        \pgfmathsetmacro{\CosValue}{1.73 *\r * cos(\i)}
        \pgfmathsetmacro{\SinValue}{1.73 *\r * sin(\i)}
        \pgfmathsetmacro{\CosValueb}{1.73 *\r * cos(\j)}
        \pgfmathsetmacro{\SinValueb}{1.73 *\r * sin(\j)}
        \draw[very thin] (\CosValue,\SinValue,0) -- (\CosValueb,\SinValueb,0) {};
      }
    \end{tikzpicture}
    \caption{The overall construction for 4-dimensional unit balls. We only represent the centers.}
    \label{fig:4dim-unit-disks}
  \end{figure}
\end{proof}

\begin{corollary}\label{cor:4ubg}
  For every $\gamma < 1$, there is a constant $\zeta < 1$ such that \cli on 4-dimensional unit ball graphs is not $\zeta$-approximable even in time $2^{n^\gamma}$, unless the ETH fails.
  Moreover, \cli is NP-hard on 4-dimensional unit ball graphs.
\end{corollary}

The proof of Theorem~\ref{thm:4udg} can be tweaked for 3-dimensional balls of different radii.
For a real $\varepsilon > 0$, we say that the radii of a representation (or the representation itself) are \emph{$\varepsilon$-close} if the radii are all contained in the interval $[1,1+\varepsilon]$.
We denote by $\mathcal B(1,1+\varepsilon)$ the ball graphs having an $\varepsilon$-close representation.

\begin{theorem}\label{thm:bg}
 For any $\varepsilon > 0$, the subclass of ball graphs $\mathcal B(1,1+\varepsilon)$ contains all the co-2-subdivisions.
\end{theorem}
\begin{proof}
  Let $x, y, z$ be the coordinates of $\mathbb R^3$.
  We start similarly and define the same $\pi(p^+(e))$ and $\pi(p^-(e))$ as in the previous construction for the points $p^+(e)$ and $p^-(e)$ on a circle $\mathcal C$ of diameter $2$ centered at $O=(0,0,0)$ on a plane $\mathcal P$ of equation $z=0$.
  One difference is that $\pi(p^+(e))$ and $\pi(p^-(e))$ are no longer projections.
  We then place the points $p(v)$ regularly spaced along the $z$-axis.
  More precisely, if the vertices are numbered $v_1,v_2,\ldots,v_n$, the position of $p(v_i)$ is $(0,0,\sqrt 3+i\varepsilon')$.
  The radius of the disk representing $v_i$ centered at $p(v_i)$ is set to $r_i := \sqrt{1+(\sqrt 3+i\varepsilon')^2}-1+\varepsilon''$.
  The radii associated to the centers $p^+(e)$ and $p^-(e)$ are all set to 1.
  We move $p^+(v_iv_j)$ (resp. $p^-(v_iv_j)$) away from $\pi(p^+(v_iv_j))$ (resp. $\pi(p^-(v_iv_j))$) in the opposite direction of $\overrightarrow{p^+(v_iv_j)p(v_i)}$ (resp. $\overrightarrow{p^-(v_iv_j)p(v_j)}$) by an infinitesimal quantity, in order to only suppress the overlap of the disks centered at $p^+(v_iv_j)$ and $p(v_i)$ (resp. $p^-(v_iv_j)$ and $p(v_j)$).
  See Figure~\ref{fig:3dim-disks} for an illustration of the construction.
  We make $\varepsilon'$ and $\varepsilon''$ small enough that all the values $r_i$ are between 1 and $1+\varepsilon$.
 
   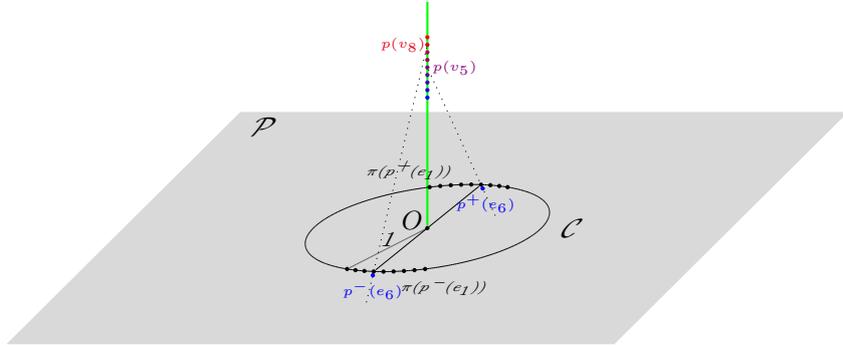
\begin{figure}
    \centering
    \begin{tikzpicture}[
        dot/.style={fill,circle,inner sep=-0.02cm},
        extended line/.style={shorten >=-#1},
        extended line/.default=1cm]
      ]
      \def\xs{0.5}
      \coordinate (a) at (-4,0,-4) {} ;
      \coordinate (b) at (4,0,-4) {} ;
      \coordinate (c) at (4,0,4) {} ;
      \coordinate (d) at (-4,0,4) {} ;
      \fill[opacity=0.15] (a) -- (b) -- (c) -- (d) -- cycle ;
      \coordinate (O) at (0,0,0) {} ;
      \node[dot] (od) at (O) {} ;
      \node at (-0.2,0.1,0) {$O$} ;
      \node[xslant=\xs]  at (-3.5,0,-3.5) {$\mathcal P$} ;
      \node[xslant=\xs]  at (1.9,0,0) {$\mathcal C$} ;

      \draw[thick,green] (od) -- (0,3,0) ;
      \foreach \i in {0,...,8}{
        \pgfmathsetmacro{\c}{\i * 12}
        \node[dot,color=red!\c!blue] (w\i) at (0,1.73+\i * 0.1,0) {} ;
      }
      \node[color=red!52!blue] at (0.36,1.73+0.4,0) {\tiny{$p(v_5)$}} ;
      \node[color=red!91!blue] at (-0.32,1.73+0.7,0) {\tiny{$p(v_8)$}} ;

      \def \r{1.5}
      \pgfmathsetmacro{\CosValue}{\r * cos(84.15)}
      \pgfmathsetmacro{\SinValue}{\r * sin(85)}
      \node[dot,blue] at (\CosValue,-0.05,-\SinValue) {};
      \node at (\CosValue + 0.04,-0.25,-\SinValue) {\tiny{\textcolor{blue}{$p^+(e_6)$}}} ;

      \pgfmathsetmacro{\CosValue}{\r * cos(95.4)}
      \pgfmathsetmacro{\SinValue}{\r * sin(94)}
      \node[dot,blue] at (\CosValue,-0.05,\SinValue) {};
      \node at (\CosValue,-0.25,\SinValue) {\tiny{\textcolor{blue}{$p^-(e_6)$}}} ;
      
      \foreach \i in {70,75,80,85,90,95,100,105,110}{
        \pgfmathsetmacro{\CosValue}{\r * cos(\i)}
        \pgfmathsetmacro{\SinValue}{\r * sin(\i)}
        \node[dot] (pm\i) at (\CosValue,0,\SinValue) {};
        \node[dot] (pp\i) at (\CosValue,0,-\SinValue) {};
      }
      \foreach \i [count=\j from 2] in {1,...,360}{
        \pgfmathsetmacro{\CosValue}{\r * cos(\i)}
        \pgfmathsetmacro{\SinValue}{\r * sin(\i)}
        \pgfmathsetmacro{\CosValueb}{\r * cos(\j)}
        \pgfmathsetmacro{\SinValueb}{\r * sin(\j)}
        \draw[very thin] (\CosValue,0,\SinValue) -- (\CosValueb,0,\SinValueb) {};
      }
      \node[xslant=\xs] (pp1) at (-1,0,-2) {\tiny{$\pi(p^+(e_1))$}} ;
      \node[xslant=\xs] (pm1) at (1,0,2) {\tiny{$\pi(p^-(e_1))$}} ;
      \draw[opacity=0.5] (pm110) -- (O) ;
      \node[xslant=\xs] at (-0.5,-0.15,0) {\footnotesize{1}} ;

      \draw[dotted,extended line=0.5cm] (w7) -- (pm95) ;
      \draw[dotted,extended line=0.5cm] (w4) -- (pp85) ;
      \draw[very thin] (pp85) -- (pm95) ;
      
    \end{tikzpicture}
    \caption{The overall construction for 3-dimensional balls. Blue centers indicate unit balls, while the redder the center, the larger the radius of the ball. We only represented the two centers for the edge $e_6 = v_5v_8$, showing how $p^+(e_6)$ is pushed away from $p(v_5)$, and $p^-(e_6)$, from $p(v_8)$. One may observe that $p(v_1), \ldots, p(v_n), O, \pi(p^+(e_6)), \pi(p^-(e_6)), p^+(e_6), p^-(e_6)$ are coplanar.}
    \label{fig:3dim-disks}
  \end{figure}
\end{proof}

We call \emph{quasi unit ball graphs} those graphs in the intersection $\bigcap_{\varepsilon > 0}\mathcal B(1,1+\varepsilon)$.
As a corollary, we get some strong inapproximability even for quasi unit ball graphs.

\begin{corollary}\label{cor:qub}
  For every $\gamma < 1$, there is a constant $\zeta < 1$ such that \cli on quasi unit ball graphs is not $\zeta$-approximable even in time $2^{n^\gamma}$, unless the ETH fails.
  Moreover, \cli is NP-hard on quasi unit ball graphs.
\end{corollary}

We observe that the lower bounds of Corollaries~\ref{cor:4ubg}~and~\ref{cor:qub} also hold when the geometric representation is given in input.
For that we need to argue that we can compute the coordinates (and radii) of our constructions in polynomial time.
Let us estimate $\eta$, the maximum distance between a pair $(\pi(p^+(e)),\pi(p^-(e')))$ with $e \neq e'$, for an $m$-edge graph with $m$ sufficiently large.
We place the $2m$ points $\pi(p^+(e_1)), \ldots, \pi(p^+(e_m)), \pi(p^-(e_1)), \ldots, \pi(p^-(e_m))$ regularly spaced every $\pi/m$ radians.
Thus $$\eta = 2\sqrt{1-\sin^2(\frac{\pi}{2m})} = 2\cos(\frac{\pi}{2m}) \leqslant 2(1-\frac{\pi^2}{8m^2}).$$
Hence $2-\eta \leqslant \frac{\pi^2}{4m^2}$.
This implies that in the proof of Theorem~\ref{thm:4udg}, we can choose $\varepsilon := \frac{1}{100m^3}$ and $\varepsilon' := \frac{1}{100m^4}$.
Now all the centers can be snapped to grid points on a three-dimensional grid with precision $\frac{1}{100m^5}$ (i.e., two adjacent grid points are at distance $\frac{1}{100m^5}$).
In the the proof of Theorem~\ref{thm:bg}, we can choose $\varepsilon' := \frac{\varepsilon}{100m^3}$ and $\varepsilon'' := \frac{\varepsilon}{100m^4}$.
And all the centers and radii can be defined up to precision $\frac{\varepsilon}{100m^5}$.

Note that we realize both constructions with polynomial precision (in $m$ and, for Theorem~\ref{thm:bg}, in $1/\varepsilon$), although even a single-exponential precision would be allowed in a polytime construction.
Indeed, in binary, polynomially many bits yield exponentially precised coordinates. 

\subsection{Filled ellipses and filled triangles}

A natural generalization of a disk is an \emph{elliptical disk}, also called a \emph{filled ellipse}, i.e., an ellipse plus its interior.
Arguably the simplest convex set with non empty interior is a filled triangle (a triangle plus its interior).  

APX-hardness was shown for \cli in the intersection graphs of (\emph{non-filled}) ellipses and triangles by Amb\"uhl and Wagner~\cite{Ambuhl05}.
Their reduction also implies that there is no subexponential algorithm for this problem, unless the ETH fails.
Moreover, they claim that their hardness result extends to filled ellipses since \emph{``intersection graphs of ellipses without interior are also intersection graphs of filled ellipses''}.
Unfortunately this claim is incorrect, as we prove below: We construct a graph, which can be represented by ellipses without their interior, but cannot be represented by \emph{any} convex sets (and thus by filled ellipses).

\begin{theorem}\label{thm:counterexample}
There is a graph $G$ which has an intersection representation with ellipses without their interior, but has no intersection representation with convex sets.
\end{theorem}

\begin{proof}
The argument is similar to the one used by Brimkov et al. \cite{Brimkov18}, which was in turn inspired by the construction by Kratochv\'il and Matou\v{s}ek \cite{DBLP:journals/jct/KratochvilM94}. 
Consider the graph $G$ in Figure~\ref{fig:counterexample2}, containing what we will henceforth refer to as \emph{black}, \emph{gray}, \emph{red}, \emph{blue}, and \emph{white} vertices.
Gray, blue, and red vertices are called \emph{connector vertices}.
We aim to show that $G$ cannot be represented by a family of convex sets.

For the sake of contradiction, suppose that $G$ admits a representation by convex sets.
We denote by $R_v$ the convex set representing the vertex $v$.
The high-level idea is as follows.
First, we observe that the union of representatives of white vertices separates the rest of the plane into two disjoint regions.
Furthermore, since the subgraph induced by black vertices is connected, their representatives must all lie in one of these two regions.
Next, we look at the set $R_a \cup R_b \cup R_c$, for the black vertices $a$, $b$, and $c$, and observe that due to the connector vertices, certain parts of $R_a, R_b, R_c$ need to appear in a prescribed order as we move along the boundary of $R_a \cup R_b \cup R_c$.
This forces the representation of $G - \{d\}$ to have a somewhat rigid structure.
We carefully analyze this structure and conclude that the only way to realize all the adjacencies of $d$ is to place $R_d$ where it partially overlaps $R_c$, which contradicts the non-edge between $c$ and $d$.

\begin{figure}[ht]
\centering
\begin{center}
\tiny
\begin{tikzpicture}[scale=0.72]
\tikzstyle{every node}=[shape = circle]


\node[fill = black, inner sep=-0.18cm] (c) at (0,0) {\color{white}{\normalsize $c$}}; 
\node[fill = black, inner sep=-0.2cm] (d) at (-3.7, 3.7) {\color{white}{\normalsize $d$}}; 
\node[fill = black, inner sep=-0.2cm] (b) at (180+360/24:1) {\color{white}{\normalsize $b$}}; 
\node[fill = black, inner sep=-0.18cm] (a) at (5*360/24:1) {\color{white}{\normalsize $a$}}; 

\draw (a) -- (b) -- (c) -- (a);

\foreach \i in {0,...,23}
{   
   \draw (\i*360/24:3) -- (\i*360/24+360/24:3);
   \draw (\i*360/12:3) -- (\i*360/12:2);
}
\foreach \i in {0,...,23}
{
   \node[fill=white,draw] (o\i) at  (\i*360/24:3) {};        
}

\foreach \i/\nam in {0/$a_2$,6/$a_1$}
{
       \node[fill=lightgray, inner sep=-0.18cm] (i\i) at (\i*360/12:2) {\color{black} \nam};
       \draw (a) -- (i\i);
}

\foreach \i/\nam in {3/$b_1$,9/$b_2$}
{
       \node[fill=lightgray, inner sep=-0.2cm] (i\i) at (\i*360/12:2) {\color{black} \nam};
       \draw (b) -- (i\i);
}

\foreach \i/\nam in {2/$c_2$,4/$c_1$,8/$c_4$,10/$c_3$}
{
       \node[fill=blue!50, inner sep=-0.18cm] (i\i) at (\i*360/12:2) {\color{black} \nam};       
       \draw (c) -- (i\i);
}

\draw (d) to [bend left] (a);
\draw (d) to [bend right] (b);
\draw (d) to [bend left = 25] (i2);
\draw (d) to [bend right = 44] (i8);
\draw (d) to  (i4);
\draw (d) to [bend right = 43] (i10);

\foreach \i/\nam in {1/$\overline{c}_2$,5/$\overline{c}_1$,7/$\overline{c}_4$,11/$\overline{c}_3$}
{
       \node[fill=red!50, inner sep=-0.2cm] (i\i) at (\i*360/12:2) {\color{black} \nam};   
       \draw (c) -- (i\i);
}
\end{tikzpicture}
\hskip 10 pt
\begin{tikzpicture}[scale=0.45]
\draw (-2.75,4.5) ellipse (1.5 and 0.75);
\draw (2.75,4.5) ellipse (1.5 and 0.75);
\draw (2.75,-4.5) ellipse (1.5 and 0.75);
\draw (-2.75,-4.5) ellipse (1.5 and 0.75);

\draw (-4.5,4.5) ellipse (0.75 and 0.5);
\draw (4.5,4.5) ellipse (0.75 and 0.5);
\draw (-4.5,-4.5) ellipse (0.75 and 0.5);
\draw (4.5,-4.5) ellipse (0.75 and 0.5);

\draw (-5.8,1.5) ellipse (0.75 and 0.5);
\draw (5.8,1.5) ellipse (0.75 and 0.5);
\draw (-5.8,-1.5) ellipse (0.75 and 0.5);
\draw (5.8,-1.5) ellipse (0.75 and 0.5);

\draw (0,6.5) ellipse (1.75 and 0.75);
\draw (0,-6.5) ellipse (1.75 and 0.75);
\draw (-6.5,0) ellipse (0.75 and 2);
\draw (6.5,0) ellipse (0.75 and 2);

\draw (1.5, 5.5) ellipse (0.5 and 1);
\draw (1.5, -5.5) ellipse (0.5 and 1);
\draw (-1.5, 5.5) ellipse (0.5 and 1);
\draw (-1.5, -5.5) ellipse (0.5 and 1);

\draw (5, 2.75) ellipse (0.5 and 1.8);
\draw (5, -2.75) ellipse (0.5 and 1.8);
\draw (-5, 2.75) ellipse (0.5 and 1.8);
\draw (-5, -2.75) ellipse (0.5 and 1.8);

\draw[line width =1, color=red!50] (-3.5, 1.3) ellipse (1.5 and 0.2);
\draw[line width =1, color=red!50] (3.5, 1.3) ellipse (1.5 and 0.2);
\draw[line width =1, color=red!50] (-3.5, -1.3) ellipse (1.5 and 0.2);
\draw[line width =1, color=red!50] (3.5, -1.3) ellipse (1.5 and 0.2);

\draw[line width=1, color=gray] (0, 4.5) ellipse (0.25 and 1.75);
\draw[line width=1, color=gray] (0, -4.5) ellipse (0.25 and 1.75);
\draw[line width=1, color=gray] (4.5, 0) ellipse (1.75 and 0.25);
\draw[line width=1, color=gray] (-4.5, 0) ellipse (1.75 and 0.25);

\draw[line width=1, color=blue!50] (1.6, 2.7) ellipse (0.15 and 1.7);
\draw[line width=1, color=blue!50] (-1.6, 2.7) ellipse (0.15 and 1.7);
\draw[line width=1, color=blue!50] (1.6, -2.7) ellipse (0.15 and 1.7);
\draw[line width=1, color=blue!50] (-1.6, -2.7) ellipse (0.15 and 1.7);

\draw[line width = 1] (0,0) ellipse (2.5 and 2.5);
\draw[line width = 1] (0,0) ellipse (2.2 and 2.2);

\draw[line width = 1] (0,0) ellipse (3.1 and 0.9);
\draw[line width = 1] (0,0) ellipse (0.9 and 3.1);

\node at (0.7,3) {\large $b$};
\node at (-1.87,0.0) {\large $d$};
\node at (2.45,1.65) {\large $c$};
\node at (3,0.8) {\large $a$};
\end{tikzpicture}

\end{center}
\caption{A graph (left), which has a representation with empty ellipses (right) but no representation with convex sets.
}
\label{fig:counterexample2}
\end{figure}
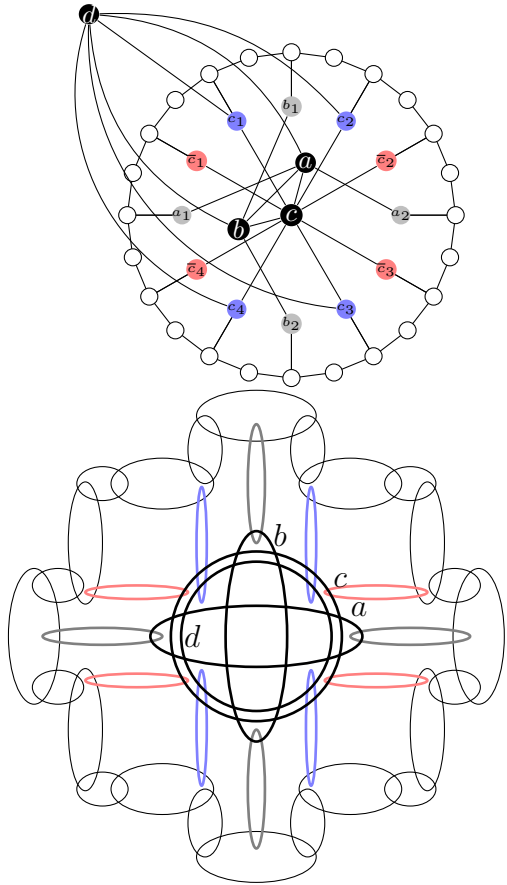

\begin{figure}[ht]
\begin{center}
\includegraphics[scale=1]{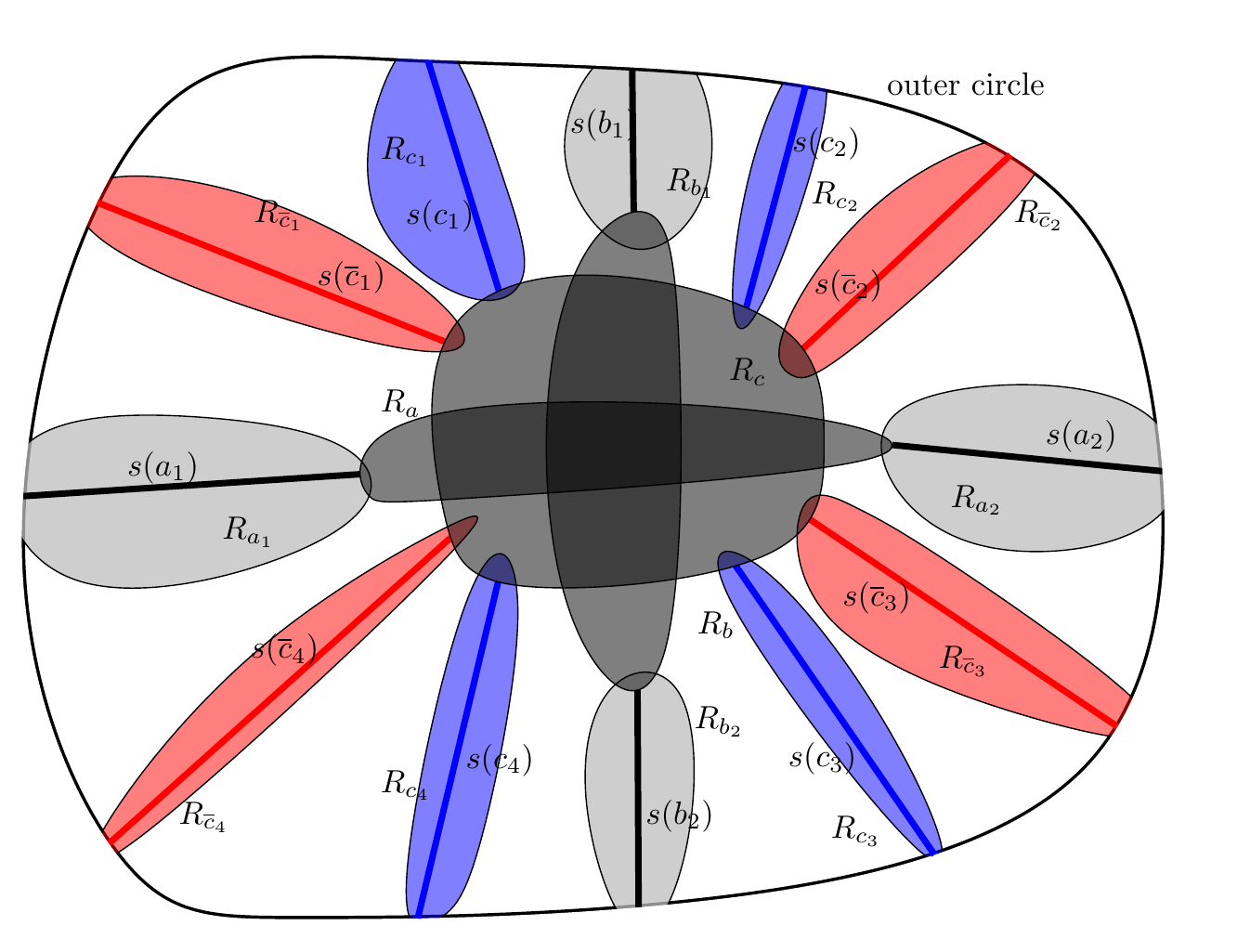}
\end{center}
\caption{A hypothetical representation of $G$ with intersecting convex sets.}
\label{fig:counterexample-construction}
\end{figure}

We refer to Figure~\ref{fig:counterexample-construction} for a picture of the construction.
The union of the representatives of white vertices contains a (closed) Jordan curve that we will call the outer circle.
Let us choose the outer circle in such a way that it intersects the representatives of all connector vertices and such that its intersection with each white or connector vertex representative is a connected set.
The outer circle divides the plane into two faces -- an interior and an exterior.

Since no black vertex is adjacent to a white vertex, the outer circle cannot be crossed by the representative of any black vertex.
Moreover, as black vertices form a connected subgraph, they have to be represented in the same face $F$ (with respect to the outer circle). Since the circle intersects vertex representatives only in connected sets, as we traverse it, it enters and leaves the representatives exactly once. We order the representatives of white vertices in such a way that when traversing the circle in a clockwise fashion the representatives of connector vertices are visited in the following order: $a_1,\overline{c}_1,c_1,b_1,c_2,\overline{c}_2,a_2,\overline{c}_3,c_3,b_2,c_4,\overline{c}_4$, where each $z_i$ for $z \in \{a,b\}$ is a distinct gray neighbor of $z$, each $c_i$ is a distinct blue neighbor of $c$, and each $\overline{c}_i$ is a distinct red neighbor of $c$.

Clearly, each gray neighbor of $a$ must intersect $R_a$ outside $R_a \cap (R_b \cup R_c)$, each gray neighbor of~$b$ must intersect $R_b$ outside $R_b \cap (R_a \cup R_c)$, and each blue or red neighbor of $c$ must intersect $R_c$ outside $R_c \cap (R_a \cup R_b)$.
Thus, some parts of $R_a$, $R_b$, and $R_c$ are exposed (i.e., outside the intersection with the union of the representatives of the remaining two vertices) in the order: $a,c,b,c,a,c,b,c$, as we move along the boundary of $R_a \cup R_b \cup R_c$.

Let $z'$ be a connector vertex and let $z \in \{a,b,c\}$ be its neighbor.
Notice that for each $z'$, the vertex $z$ is uniquely defined.
The set $R_{z'}$ contains a segment $s(z')$, whose one end is on the boundary of $R_z$ and the other end is on the outer circle (recall that all representatives are convex).

Observe that the segments $s(a_1),s(a_2),s(b_1),s(b_2)$ are pairwise disjoint and they separate $F \setminus (R_a \cup R_b)$ into four regions $Q_1,Q_2,Q_3,Q_4$, such that for $i\in [4]$ it holds that $R_{c_i} \cap F \subseteq Q_i$.
Note that one of these regions may be unbounded, if $F$ is the unbounded face of the outer circle.
For each $i \in [4]$, since $c_i$ is non-adjacent to $\overline{c}_i$, we observe that $(R_{c_i} \cap F) \setminus R_c$ is contained in the subregion of $Q_i$ bounded by $s(\overline{c_i}) \cup R_c$.

For $i = \{1,2,3,4\}$, let $p_i$ be a point in $R_d \cap R_{c_i}$.
Such a point exists since $d$ is adjacent to $c_i$.
By convexity of $R_d$, the segment $p_1p_2$ is contained in $R_d$.
Since $d$ is non-adjacent to $c,\overline{c}_1$, and $\overline{c}_2$, we observe that $p_1p_2$ must cross the $s(b_1) \cup R_b$.
As $d$ is non-adjacent to $b_1$, the segment $p_1p_2$ intersects $R_b$.
Let $q_1$ be an intersection point of $p_1p_2$ and $R_b$.
In an analogous way, we define $q_2$ to be an intersection point of $p_3p_4$ and $R_b$.

Let us now consider the segment $q_1q_2$.
By convexity of $R_d$ and $R_b$, we have $q_1q_2 \subseteq R_b \cap R_d$.
The segment $q_1q_2$ must intersect $s(c_1) \cup R_c \cup s(c_2)$, and let $q'$ be this intersection point.
If $q' \in s(c_1) \cup s(c_2)$, we get the contradiction from the fact that $b$ is non-adjacent to $c_1$ and $c_2$.
On the other hand, if $q' \in R_c$, we get the contradiction from the fact that $d$ and $c$ are non-adjacent.

Finally, it is easy to represent $G$ with empty ellipses (see Fig. \ref{fig:counterexample2} right).
\end{proof}
This error and the confusion between filled ellipses and ellipses without their interior has propagated to other more recent papers \cite{Keller17}.
Fortunately, we show that the hardness result does hold for filled ellipses (and filled triangles) with a different reduction.
Our construction can be seen as streamlining the ideas of Ambühl and Wagner \cite{Ambuhl05}.
It is simpler and, in the case of (filled) ellipses, yields a somewhat stronger statement.

\begin{theorem}\label{thm:hardness-filled-ellipses}
  For every $\gamma < 1$, there is a constant $\zeta < 1$ such that \cli on the intersection graphs of filled ellipses has no $\zeta$-approximation algorithm running in subexponential time $2^{n^\gamma}$, unless the ETH fails, even when the ellipses have arbitrarily small eccentricity and the different lengths of the major axis are arbitrarily close.
\end{theorem}

This is in sharp contrast with our subexponential algorithm and PTAS when the eccentricity is 0 (case of disks).
For any $\varepsilon > 0$, if the eccentricity is only allowed to be at most $\varepsilon$, a SUBEXPAS is very unlikely.
This result subsumes \cite{ceroi} (where NP-hardness is shown for connected shapes contained in a disk of radius 1 and containing a concentric disk of radius $1-\varepsilon$ for arbitrarily small $\varepsilon > 0$) and corrects \cite{Ambuhl05}.
We show the same hardness for the intersection graphs of filled triangles.

\begin{theorem}\label{thm:hardness-filled-triangles}
 For every $\gamma < 1$, there is a constant $\zeta < 1$ such that \cli on the intersection graphs of filled triangles has no $\zeta$-approximation algorithm running in subexponential time $2^{n^\gamma}$, unless the ETH fails.
\end{theorem}

Once again, to show Theorem~\ref{thm:hardness-filled-ellipses} and Theorem~\ref{thm:hardness-filled-triangles}, it is sufficient to show that intersection graphs of (filled) ellipses or of (filled) triangles contain all co-2-subdivisions.
We start with (filled) triangles since the construction is simpler.

\begin{lemma}\label{lem:triangles-co2subd}
The class of intersection graphs of filled triangles contains all co-2-subdivisions.
\end{lemma}
  
\begin{proof}
Let $G$ be any graph with $n$ vertices $v_1, \ldots, v_n$ and $m$ edges $e_1,\ldots,e_m$, and $H$ be its co-2-subdivision.  
  We start with $n+2$ points $p_0, p_1, p_2, \ldots, p_n, p_{n+1}$ forming a convex monotone chain.
  Those points can be chosen as $p_i := (i,p(i))$ where $p$ is the equation of a positive parabola taking its minimum at $(0,0)$.
  For each $i \in [0,n+1]$, let $q_i$ be the reflection of $p_i$ by the line of equation $y = 0$.
  Let $x := (n+1,0)$.
  For each vertex $v_i \in V(G)$ the filled triangle $\delta_i := p_iq_ix$ encodes $v_i$.
  Observe that the points $p_0=q_0$, $p_{n+1}$, and $q_{n+1}$ will only be used to define the filled triangles encoding edges. 

  To encode (the two new vertices of) a subdivided edge $e_k=v_iv_j$, we use two filled triangles $\Delta^+_k$ and $\Delta^-_k$. The triangle 
  $\Delta^+_k$ (resp. $\Delta^-_k$) has an edge which is supported by $\ell(p_{i-1},p_{i+1})$ (resp. $\ell(q_{j-1},q_{j+1})$) and is prolonged so that it crosses the boundary of each $\delta_{i'}$ but $\delta_i$ (resp. but $\delta_j$).
  A second edge of $\Delta^+_k$ and $\Delta^-_k$ are parallel and make with the horizontal a small angle $\varepsilon k$, where $\varepsilon> 0$ is chosen so that $\varepsilon m$ is smaller than the angle formed by $\ell(p_0,p_1)$ with the horizontal line.
  Those almost horizontal edges intersect for each pair $\Delta^+_{k'}$ and $\Delta^-_{k''}$ with $k' \neq k''$ intersects close to the same point.
  Filled triangles $\Delta^+_k$ and $\Delta^-_k$ do not intersect.
  See Figure~\ref{fig:triangles-co2subd} for the complete picture.

    It is easy to check that the intersection graph of $\{\delta_i\}_{i \in [n]} \cup \{\Delta^+_k,\Delta^-_k\}_{k \in [m]}$ is $H$.
    The family $\{\delta_i\}_{i \in [n]}$ forms a clique since they all contain for instance the point $x$.
    The filled triangle $\Delta^+_k$ (resp. $\Delta^-_k$) intersects every other filled triangles except $\Delta^-_k$ (resp. $\Delta^+_k$) and $\delta_i$ (resp. $\delta_j$) with $e_k=v_iv_j$.
    
    One may observe that no triangle is fully included in another triangle.
    So the construction works both as the intersection graph of filled triangles \emph{and} triangles without their interior.
    The edges of $\Delta^+_k$ and $\Delta^-_k$ crossing the boundary of all but one $\delta_i$ can be arbitrary prolonged to the right.
    The almost horizontal edges of these triangles can be arbitrary prolonged to the left.
    Thus, the triangles can all be made isosceles. 
   \end{proof}
  
    \begin{figure}[h!]
      \centering
      \begin{tikzpicture}[
          inv/.style={opacity=0},
          dot/.style={fill,circle,inner sep=-0.01cm},
          vert/.style={draw, fill=red, opacity=0.2},
          verta/.style={draw, fill=blue, opacity=0.2},
          vertb/.style={draw, fill=green, opacity=0.2},
          extended line/.style={shorten >=-#1,shorten <=-#1},
          extended line/.default=1cm,
          one end extended/.style={shorten >=-#1},
          one end extended/.default=1cm,
        ]
        
        \def\n{5}
        \def\a{0.08}
        \def\b{0.2}
        \def\c{0.02}
        \def\d{-5}
        \def\e{5}
        \def\f{0.05}
        \coordinate (su) at (\d,\c) ;
        \coordinate (sd) at (\d,-\c) ;
        \coordinate (eu) at (\e,\c) ;
        \coordinate (ed) at (\e,-\c) ;

        \coordinate (su2) at (\d,-0.1) ;
        \coordinate (eu2) at (\e,0.4) ;
        \coordinate (sd2) at (\d,-0.1 - 2 * \c) ;
        \coordinate (ed2) at (\e,0.4 - 2 * \c) ;
        \coordinate (x) at (\n+1,0) ;
        \node[dot] at (x) {} ;
        
        \foreach \i in {1,...,5}{
          \coordinate (p\i) at (\i,\i * \i * \a + \i * \b) ;
          \coordinate (q\i) at (\i,- \i * \i * \a - \i * \b) ;
          \coordinate (pd\i) at (\i,\i * \i * \a + \i * \b - \f) ;
          \coordinate (qu\i) at (\i,- \i * \i * \a - \i * \b + \f) ;
        }
        \foreach \i in {1,...,5}{
          \node[dot] at (p\i) {} ;
          \node[dot] at (q\i) {} ;
          \draw[vert] (p\i) -- (q\i) -- (x) -- cycle ;
        }

        \path[name path=J1,overlay] (su) -- (eu)--([turn]0:5cm);
        \path[name path=K1,overlay] (pd2) -- (0,0)--([turn]0:5cm);
        \path[name path=J2,overlay] (sd) -- (ed)--([turn]0:5cm);
        \path[name path=K2,overlay] (qu5) -- (qu3)--([turn]0:5cm);

        \path [name intersections={of=J1 and K1,by={I1}}];
        \path [name intersections={of=J2 and K2,by={I2}}];

        \coordinate (e1) at (I1) ;
        \coordinate (e2) at (I2) ;

        \coordinate (c1) at ( $ (pd2)!-2!(e1) $ ) ;
        \coordinate (c2) at ( $ (qu5)!-0.1!(e2) $ ) ;

        \path[name path=J3,overlay] (su2) -- (eu2)--([turn]0:5cm);
        \path[name path=K3,overlay] (pd3) -- (pd1)--([turn]0:5cm);
        \path[name path=J4,overlay] (sd2) -- (ed2)--([turn]0:5cm);
        \path[name path=K4,overlay] (qu4) -- (qu2)--([turn]0:5cm);

        \path [name intersections={of=J3 and K3,by={I3}}];
        \path [name intersections={of=J4 and K4,by={I4}}];

        \coordinate (e3) at (I3) ;
        \coordinate (e4) at (I4) ;
        \coordinate (c3) at ( $ (pd3)!-1.4!(e3) $ ) ;
        \coordinate (c4) at ( $ (qu4)!-0.35!(e4) $ ) ;

         \foreach \i/\j/\k in {1/su/vertb,2/sd/vertb,3/su2/verta,4/sd2/verta}{
            \draw[\k] (\j) -- (c\i) -- (e\i) -- cycle ;
        }
        
      \end{tikzpicture}
      \caption{A co-2-subdivision of a graph with $5$ vertices (in red) represented with triangles. Only two edges are shown: one between vertices $1$ and $4$ (green) and one between vertices $2$ and $3$ (blue).}
      \label{fig:triangles-co2subd}
    \end{figure}
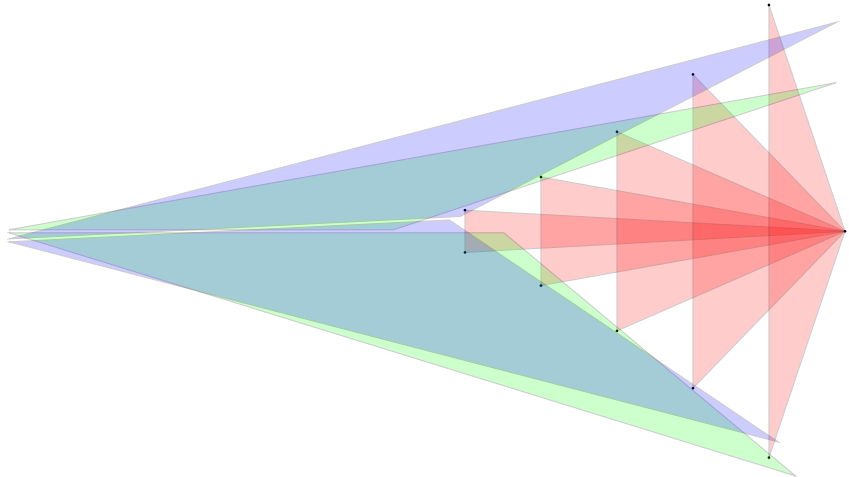

  We use the same ideas for the construction with filled ellipses.
  Two tangents of the ellipse will play the role of the two important sides of the triangles encoding edges of the initial graph $G$.   

  \begin{lemma}\label{lem:ellipses-co2subd}
    The class of intersection graphs of filled ellipses contains all co-2-subdivisions.
  \end{lemma}

  \begin{proof}
    Let $G$ be any graph with $n$ vertices $v_1, \ldots, v_n$ and $m$ edges $e_1,\ldots,e_m$, and $H$ be its co-2-subdivision. 
    We start with the convex monotone chain $p_0, p_1, p_2, \ldots, p_{n-1}, p_n, p_{n+1}$, only the gap between $p_i$ and $p_{i+1}$ is chosen very small compared to the positive $y$-coordinate of $p_0$.
    The disks $\mathcal D_i$ encoding the vertices $v_i \in G$ must form a clique.
    We also take $p_0$ with a large $x$-coordinate.
    For $i \in [0,n+1]$, $q_i$ is the symmetric of $p_i$ with respect to the $x$-axis.
    For each $i \in [n]$, we define $\mathcal D_i$ as the disk whose boundary is the unique circle which goes through $p_i$ and $q_i$, and whose tangent at $p_i$ has the direction of $\ell(p_{i-1},p_{i+1})$.
    It can be observed that, by symmetry, the tangent of $\mathcal D_i$ at $q_i$ has the direction of $\ell(q_{i-1},q_{i+1})$.

    Let us call $\tau^+_i$ (resp. $\tau^-_i$) the tangent of $\mathcal D_i$ at $p_i$ (resp. at $q_i$) very slightly translated upward (resp. downward).
    The tangent $\tau^+_i$ (resp. $\tau^-_i$) intersects every disks $\mathcal D_{i'}$ but $\mathcal D_i$ (see Figure~\ref{fig:disks-all-but-one}).
    Let denote by $p'_i$ (resp. $q'_i$) be the projection of $p_i$ (resp. $q_i$) onto $\tau^+_i$ (resp. onto $\tau^-_i$).  
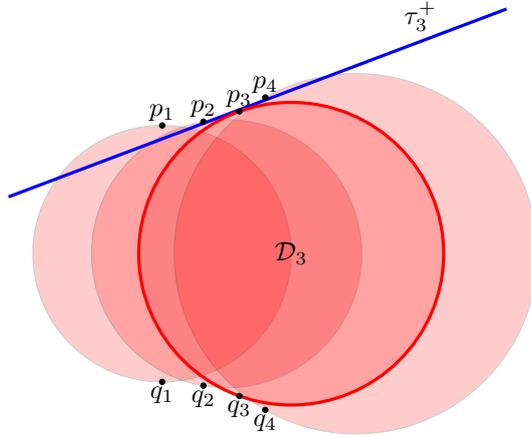
\begin{figure}
  \centering
  \begin{tikzpicture}[scale=1.7, xscale=-1,
      vert/.style={draw, fill=red, opacity=0.2},
      dot/.style={fill,circle,inner sep=-0.03cm},
    ]
    \def\s{0.5}
    \def\t{0.045}
    \foreach \i in {0,...,3}{
      \coordinate (c\i) at (- \i * \s,0) ;
      \pgfmathsetmacro\r{1+\i * \i * \t} ;
      \draw[vert] (c\i) circle (\r) ;
      \pgfmathsetmacro\a{90.1 - \i * 10}
      \pgfmathsetmacro\x{- \i * \s + \r * cos(\a)}
      \pgfmathsetmacro\y{ \r * sin(\a)}
      \pgfmathsetmacro\ym{- \r * sin(\a)}
      \pgfmathsetmacro\j{1.1 + \i * 0.4}
      \pgfmathsetmacro\jj{2.6 - \i * 0.2}
      \pgfmathtruncatemacro\ipo{\i+1}
      \coordinate (p\i) at (\x, \y) ;
      \node at (\x, \y+0.1) {$p_\ipo$} ;
      \coordinate (e\i) at (\x, \ym) ;
      \node at (\x, \ym-0.1) {$q_\ipo$};
      \path[overlay] (c\i) -- (p\i) -- ([turn]-90:\j cm) node (q\i) {} ;
      \path[overlay] (c\i) -- (p\i) -- ([turn]90:\jj cm) node (qq\i) {} ;
      \path[overlay] (c\i) -- (e\i) -- ([turn]90:\j cm) node (f\i) {} ;
      \path[overlay] (c\i) -- (e\i) -- ([turn]-90:\jj cm) node (ff\i) {} ;
    }
    \node at (c2) {$\mathcal D_3$} ; 
    \draw[blue,very thick] ($ (q2) + (0,-0.02) $) -- ($ (qq2) + (0,0.05) $) ;
    \draw[red,very thick] (c2) circle (1+4*\t) ;

    \node at (-2,1.85) {$\tau^+_3$}; 

    \foreach \i in {0,...,3}{
      \node[dot] at (p\i) {} ;
      \node[dot] at (e\i) {} ;
    }
  \end{tikzpicture}
   \caption{The blue line intersects every red disk but the third one.}
  \label{fig:disks-all-but-one}
\end{figure}
For each $k \in [m]$, let $\ell_k$ be the line crossing the origin $O=(0,0)$ and forming with the horizontal an angle $\varepsilon k$, where $\varepsilon k$ is smaller than the angle formed by $\ell(p_0,p_1)$ with the horizontal.
Let $\ell^+_k$ (resp. $\ell^-_k$) be $\ell_k$ very slightly translated upward (resp. downward). 
  To encode an edge $e_k=v_iv_j$, we have two filled ellipses $\mathcal E^+_k$ and $\mathcal E^-_k$. The ellipse
  $\mathcal E^+_k$ (resp. $\mathcal E^-_k$) is defined as being tangent with $\tau^+_i$ at $p'_i$ (resp. with $\tau^-_j$ at $q'_j$)
  and tangent at $\ell^+_k$ (resp. $\ell^-_k$) at the point of $x$-coordinate $0$ (thus very close to $O$), where $e_k=v_iv_j$.
  The proof that the intersection graph of $\{\mathcal D_i\}_{i \in [n]} \cup \{\mathcal E^+_k,\mathcal E^-_k\}_{k \in [m]}$ is $H$ is similar to the case of filled triangles.

  As no ellipse is fully contained in another ellipse, this construction works for both filled ellipses \emph{and} ellipses without their interior.
  
  We place $p_0$ at $P:=(\sqrt 3/2,1/2)$ and make the distance between $p_i$ and $p_{i+1}$ very small compared to 1.
  All points $p_i$ are very close to $P$ and all points $q_i$ are very close to $Q:=(\sqrt 3/2,-1/2)$.
  This makes the radius of all disks $\mathcal D_i$ arbitrarily close to 1.
  We choose the convex monotone chain $p_0, \ldots, p_{n+1}$ so that $\ell(p_0,p_1)$ forms a 30-degree angle with the horizontal.
As, the chain is strictly convex but very close to a straight-line, $\ell(p_0,p_1) \approx \ell(p_n,p_{n+1}) \approx \ell(p_i,p_{i+1}) \approx \ell(p_i,p_{i+2})$. Thus, all those lines almost cross $P$ and form an angle of roughly 30-degree with the horizontal.
  The same holds for points $q_i$.
  For the choice of an elliptical disk tangent to the $x$-axis at $O$ and to a line with a 60-degree slope at $P$ (resp. at $Q$), we take a disk of radius 1 centered at $(0,1)$ (resp. at $(0,-1)$); see Figure~\ref{fig:almost-disks}.
  
  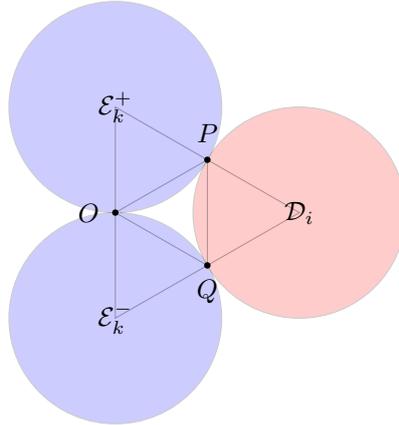
\begin{figure}[h!]
  \centering
  \begin{tikzpicture}[
      scale=1.4,
      vert/.style={draw, fill=red, opacity=0.2},
      verta/.style={draw, fill=blue, opacity=0.2},
      vertb/.style={draw, fill=blue, opacity=0.2},
      dot/.style={fill,circle,inner sep=-0.03cm}]
      \draw[vertb] (0,0) circle (1) ;
      \node at (0,0) {$\mathcal E^-_k$} ;
      \coordinate (cm) at (0,0) ;
      \draw[verta] (0,2) circle (1) ;
      \node at (0,2) {$\mathcal E^+_k$} ;
      \coordinate (cp) at (0,2) ;
      \draw[vert] (1.73,1) circle (1) ;
      \node at (1.73,1) {$\mathcal D_i$} ;
      \coordinate (cv) at (1.73,1) ;
      \node[dot] at (1.73/2,1.5) {} ;
      \node[dot] at (1.73/2,0.5) {} ;
      \node[dot] at (0,1) {} ;
      \node at (1.73/2,1.75) {$P$} ;
      \node at (1.73/2,0.25) {$Q$} ;
      \node at (-0.25,1) {$O$} ;
      \coordinate (P) at (1.73/2,1.5) ;
      \coordinate (Q) at (1.73/2,0.5) ;
      \coordinate (O) at (0,1) ;
      \draw[opacity=0.3] (O) -- (P) -- (Q) -- cycle ;
      \draw[opacity=0.3] (cm) -- (O) -- (cp) -- (P) -- (cv) -- (Q) -- cycle ;
    \end{tikzpicture}
    \caption{The layout of the disks $\mathcal D_i$, and the elliptical disks $\mathcal E^+_k$ and $\mathcal E^-_k$.}
    \label{fig:almost-disks}
  \end{figure}

  The acute angle formed by $\ell_1$ and $\ell_m$ (incident in $O$) is made arbitrarily small so that, by continuity of the elliptical disk defined by two tangents at two points, the filled ellipses $\mathcal E^+_k$ and $\mathcal E^-_k$ have eccentricity arbitrarily close to 0 and major axis arbitrarily close to 1. 
  \end{proof}

  In the construction, we made \emph{both} the eccentricity of the (filled) ellipses arbitrarily close to 0 and the ratio between the largest major axis and the smallest major axis arbitrarily close to 1.
  We know that this construction is very unlikely to work for the extreme case of unit disks, since a polynomial algorithm is known for \textsc{Max Clique}.
  Note that even with disks of arbitrary radii, Theorem~\ref{thm:main-structural-non-disk} unconditionally proves that the construction does fail.
  Indeed the co-2-subdivision of $C_3+C_3$ is the complement of $C_9+C_9$, hence not a disk graph.

  As in the previous section, the constructions of Lemmas~\ref{lem:triangles-co2subd}~and~\ref{lem:ellipses-co2subd} require only polynomial precision (when single-exponential precision would be enough for a polytime algorithm).
  Hence the lower bounds of Theorems~\ref{thm:hardness-filled-ellipses}~and~\ref{thm:hardness-filled-triangles} also hold when the geometric representation is part of the input.

\subsection{Homothets of a convex polygon}

Another natural direction of generalizing a result on disk intersection graphs is to consider {\em pseudodisk intersection graphs}, i.e., intersection graphs of collections of closed subsets of the plane (regions bounded by simple Jordan curves) that are pairwise in a {\em pseudodisk} relationship (see Kratochv\'il \cite{DBLP:conf/gd/Kratochvil96}). Two regions $A$ and $B$ are in pseudodisk relation if both differences $A\setminus B$ and $B\setminus A$ are arc-connected.
It is known that $P_{hom}$ graphs, i.e., intersection graphs of homothetic copies of a fixed polygon $P$, are pseudodisk intersection graphs~\cite{agarwal453state}. As shown by Brimkov {\em et al.}, for every convex $k$-gon $P$, a $P_{hom}$ graph with $n$ vertices has at most $n^k$ maximal cliques~\cite{Brimkov18}. This clearly implies that \cli, but also \textsc{Clique $p$-Partition} for fixed $p$ is polynomially solvable in $P_{hom}$ graphs.
Actually, the bound on the maximum number of maximal cliques from \cite{Brimkov18} holds for a more general class of graphs, called $k_{DIR}$-CONV, which  admit a intersection representation by convex polygons, whose every side is parallel to one of $k$ directions.

Moreover, we observe that Theorem \ref{thm:coEvenCycles} cannot be generalized to $P_{hom}$ graphs or $k_{DIR}$-CONV graphs. Indeed, consider the complement $\overline{P_n}$ of an $n$-vertex path $P_n$. The number of maximal cliques in $\overline{P_n}$, or, equivalently, maximal independent sets  in $P_n$ is $\Theta(c^n)$ for $c \approx 1.32$, i.e., exponential in $n$ \cite{DBLP:journals/jgt/Furedi87}. Therefore, for every fixed polygon $P$ (or for every fixed $k$) there is $n$, such that $\overline{P_n}$ is not a $P_{hom}$ ($k_{DIR}$-CONV) graph.

\section{Remarks and further directions}\label{sec:perspectives}

The algorithm of Theorem~\ref{thm:eptas} also works for weighted graphs.
\begin{theorem}
  \textsc{Maximum Weighted Independent Set} admits a randomized EPTAS and a deterministic PTAS on disk and unit ball graphs.
\end{theorem}
  The slight modifications to approximate \textsc{Maximum Weighted Independent Set} consist in sampling $S$ \emph{proportionally to the weights}, and to remove the \emph{lightest} layer $L_i$ (among the first $\lceil {2}/{\beta \varepsilon} \rceil$) and the \emph{lightest} set $S^\gamma$ (rather than the ones of minimum size).
  We then use repeatedly that \textsc{Maximum Weighted Independent Set} can be solved in polynomial time on bipartite graphs.

  This implies a randomized EPTAS for \textsc{Maximum Weighted Clique} on disk graphs and unit ball graphs, with the same arguments used to get $|I| \geqslant \beta |V(G)|$.
  Now, what we obtain is $w(I) \geqslant \beta w(V(G))$ where $w$ is the weight function and $w(X):=\Sigma_{u \in X}w(u)$.

  \medskip
  
  One might wonder what is the constant hidden in $O(1)$ in the time complexity of the randomized EPTAS $f(\varepsilon)n^{O(1)}$.
  Here is how to achieve near quadratic time $f(\varepsilon) n^2 \log n$ where $n$ is the number of vertices of our \emph{unit ball graph} $G$ (the geometric representation is not required).
  The first observation is that instead of finding an optimum stable set in a bipartite graph (which we do several times as a subroutine), one only needs a $(1-\varepsilon)$-approximation of it.
  This can be done in time $f(\varepsilon)m$, where $m$ is the number of edges \cite{Duan14}.
  We also need to overcome our first branching guess of a vertex which belongs to an optimal solution (and hence multiply by $n$ our complexity).
  To achieve this, we start by packing greedily disjoint neighborhoods $N(v_1),N(v_2), \ldots, N(v_k)$ of $G$, while it is possible.
  Then we consider the set of subgraphs $G_1,G_2,...,G_k$ induced by the vertices at distance at most 3 of $v_1,v_2,...,v_k$, respectively.
Observe that by maximality of the packing, every vertex is at distance at most 2 of some $v_i$, and thus, every edge, and even every clique of $G$ belongs to at least one $G_i$.
By a volume argument, any vertex, and thus any edge, belongs to at most a constant number of graphs $G_i$.
Thus we only need to compute the maximum clique over the graphs $G_i$, each with number of edges $m_i$, with the property that $\sum_{i=1}^km_i=O(m)$.
Now, the difference is that the clique number of each $G_i$ is at least a constant fraction of its number of vertices, so the sampling step succeeds with probability at least a positive constant.
We consider the complement of each graph $G_i$ and approximate the independence number.
The main and only obstacle is that finding a shortest odd cycle in quadratic time seems hopeless.
Indeed, the other elements of the algorithm: removing $N[S]$, computing the sets $L_i$ and $S^\gamma$, removing the lightest of them, and $(1-\varepsilon)$-approximating the maximum stable set in a bipartite graph, can all be done in quadratic time.

We only compute one (not necessarily shorter) odd cycle of length $h$, via breadth-first search for instance.
We can assume that $h =\omega(1/\varepsilon^2)$ otherwise we are done.
We take potentially \emph{thicker} slices for $S^\gamma$: instead of $z=\Theta(1/\varepsilon)$ consecutive strata in each layer up to $L_i$, we take $10 \beta \varepsilon h$ consecutive strata.
Either the suggested $2$-coloring of $H''-S^\gamma$ is proper and we are done.
Or there is a monochromatic edge, which, going through the different cases of Claim~\ref{clm:bip}, yields a new odd cycle of length $O(1/\varepsilon)$ (which is an easy case to handle) or short-cutting the previous odd cycle by at least $\Omega(\varepsilon h)$ vertices.
In the latter case, our new odd cycle is shorter by a constant multiplicative factor $1-\varepsilon$.
Hence, after at most $O(\log n)$ improvements, we find an odd cycle which is small enough to conclude.
Therefore, the overall complexity of the algorithm is $f(\varepsilon)n^2 \log n$.
  
\medskip

 The obvious remaining question is the complexity of \cli in disk graphs and in unit ball graphs.
We showed why the versatile approach of representing complements of even subdivisions of graphs forming a class on which \mis is NP-hard fails if the class is \emph{general graphs}, \emph{planar graphs}, or even any class containing the disjoint union of two odd cycles.
This approach was used by Middendorf and Pfeiffer for some string graphs \cite{Middendorf92} (with the class of all graphs), Cabello et al. \cite{CabelloCL13} to settle the then long-standing open question of the complexity of \cli for segments (with the class of planar graphs), in Section~\ref{sec:gen&lim} of this paper for filled ellipses, filled triangles, quasi unit ball graphs, and 4-dimensional unit ball graphs (with the class of all graphs).
Determining the complexity of \mis on graphs without two vertex-disjoint odd cycles as an induced subgraph is a valuable first step towards settling the complexity of \cli on disks.
 
  An interesting direction would be to find a toy problem on which we could prove NP-hardness.
  A nice class, which appears to be a subclass of unit ball graphs, is that of the so-called \emph{Borsuk graphs}: We are given some (small) real $\varepsilon>0$ and a finite collection $V$ of unit vectors in $\mathbb R^3$.
  The Borsuk graph $B(V,\varepsilon)$ has vertex set $V$ and its edges are all pairs $\{v,v'\}$ whose dot product is at most $-1+\varepsilon$ (i.e., near antipodal).

  The difficulty of computing the (weighted) independence number on Borsuk graphs is also an open question.
  A notable subclass of Borsuk graphs where this problem is polynomial-time solvable is the class of the quadrangulations of the projective plane.
  These well-studied objects have the striking property to be either bipartite or 4-chromatic.
  Furthermore, the odd cycle packing number of these graphs is at most 1.
  Artmann et al. recently showed that so-called \emph{bimodular integer programming}, that is integer programming where the constraint matrix has full rank and all its subdeterminants are in $\{-2,-1,0,1,2\}$, can be solved in strongly polynomial time \cite{Artmann17}.
  They also observe that \textsc{Maximum Weighted Independent Set} on graphs with $\text{ocp} \geqslant 1$ is a bimodular integer programming problem.
  This implies the tractability of computing the weighted independence number on quadrangulations of the projective plane.
  


\medskip

A second question is to derandomize the EPTAS.
  The difficulty here is concentrated in the sampling.
  The VC-dimension argument seems easy to deal with, however we need that our sampling falls in the maximum independent set (or at least in some independent set which is close to maximum). 

  \medskip

  Another natural question is to find a superclass of geometric intersection graphs which both contain unit ball graphs and disk graphs.
  More generally, is it possible to explain why we have the same forbidden induced subgraph (the complement of a disjoint union of two odd cycles) for disk graphs and unit ball graphs?
  The suggested proofs for Theorem~\ref{thm:main-structural-non-disk} and Theorem~\ref{thm:no-two-odd-cycles-ubg} are quite different.
  
  Let us call \emph{quasi unit disk graphs} those disk graphs that can be realized for any $\varepsilon > 0$ with disks having all the radii in the interval $[1,1+\varepsilon]$.   Recall that we showed that, for the clique problem, quasi unit ball graphs are unlikely to have a QPTAS, while unit ball graphs admit an EPTAS.
  In dimension 2, it can be easily shown that unit disk graphs form a proper subset of quasi unit disk graphs, which form themselves a proper subset of disk graphs.
  Can we find for this intermediate class an efficient exact algorithm solving \cli?
  \begin{problem}
   Is there a polynomial-time algorithm for \cli on quasi unit disk graphs?
  \end{problem}

  \medskip

  Our randomized EPTAS works for \mis under three assumptions.
  While it is clear that we crucially need that $\iocp \leqslant 1$ (or at least that $\iocp$ is constant), as far as we can tell, the boundedness of the VC-dimension and the fact that the solution is of linear size might not be required.
  Recall that we did obtain a randomized PTAS in the class $\cl(\infty,\beta,O(1))$, and a deterministic PTAS in $\cl(d,0,O(1))$.
  \begin{problem}
    Is there an EPTAS for \mis on graphs without the union of two odd cycles as an induced subgraph, or even with $\iocp=O(1)$, and either one of the following conditions:
    \begin{itemize}
    \item there is a solution of size at least $\beta n$ for some constant $\beta$, 
    \item the VC-dimension of the graph is bounded by a constant $d$? 
    \end{itemize}
  \end{problem}
  As this paper was under review, Dvo\v{r}\'ak and Pek\'arek~\cite{Dvorak20} answered the first item positively.
  More precisely they obtained a randomized EPTAS for the class $\cl(\infty,\beta,i)$ with running time $f(\varepsilon)\Tilde{O}(n^{i+4})$.
  Matching the $f(\varepsilon)\Tilde{O}(n^2)$-time for unit balls (hence with $i=1$), as well as the second item remain open. 
  
  It might also be that no additional condition is needed.
  \begin{problem}
    Is there a(n E)PTAS for \mis on graphs without the union of two odd cycles as an induced subgraph?
  \end{problem}

  Atminas and Zamaraev \cite{Atminas16} showed that the complement of $K_2+C_s$ is not a unit disk graph when $s$ is odd (where $K_2$ is an edge and $C_s$ is a cycle on $s$ vertices).
  Is this obstruction enough to obtain an alternative polynomial-time algorithm for \cli on unit disk graphs? 
  \begin{problem}
  Is \mis solvable in polynomial-time on graphs excluding the union of an edge and an odd cycle as an induced subgraph?
  \end{problem}

\bibliographystyle{abbrv}
\bibliography{main}

\begin{thebibliography}{10}

\bibitem{Afshani05}
P.~Afshani and T.~M. Chan.
\newblock Approximation algorithms for maximum cliques in 3d unit-disk graphs.
\newblock In {\em Proceedings of the 17th Canadian Conference on Computational
  Geometry, CCCG'05, University of Windsor, Ontario, Canada, August 10-12,
  2005}, pages 19--22, 2005.

\bibitem{Afshani08}
P.~Afshani and H.~Hatami.
\newblock Approximation and inapproximability results for maximum clique of
  disc graphs in high dimensions.
\newblock {\em Inf. Process. Lett.}, 105(3):83--87, 2008.

\bibitem{AgarwalM06}
P.~K. Agarwal and N.~H. Mustafa.
\newblock Independent set of intersection graphs of convex objects in 2d.
\newblock {\em Comput. Geom.}, 34(2):83--95, 2006.

\bibitem{agarwal453state}
P.~K. Agarwal, J.~Pach, and M.~Sharir.
\newblock State of the union (of geometric objects).
\newblock {\em Surveys in Discrete and Computational Geometry: Twenty Years
  Later. Contemporary Mathematics}, 453:9--48, 2008.

\bibitem{Alber04}
J.~Alber and J.~Fiala.
\newblock Geometric separation and exact solutions for the parameterized
  independent set problem on disk graphs.
\newblock {\em J. Algorithms}, 52(2):134--151, 2004.

\bibitem{AlonYZ97}
N.~Alon, R.~Yuster, and U.~Zwick.
\newblock Finding and counting given length cycles.
\newblock {\em Algorithmica}, 17(3):209--223, 1997.

\bibitem{Ambuhl05}
C.~Amb{\"{u}}hl and U.~Wagner.
\newblock The clique problem in intersection graphs of ellipses and triangles.
\newblock {\em Theory Comput. Syst.}, 38(3):279--292, 2005.

\bibitem{Aronov18}
B.~Aronov, A.~Donakonda, E.~Ezra, and R.~Pinchasi.
\newblock On pseudo-disk hypergraphs.
\newblock {\em CoRR}, abs/1802.08799, 2018.

\bibitem{Artmann17}
S.~Artmann, R.~Weismantel, and R.~Zenklusen.
\newblock A strongly polynomial algorithm for bimodular integer linear
  programming.
\newblock In {\em Proceedings of the 49th Annual {ACM} {SIGACT} Symposium on
  Theory of Computing, {STOC} 2017, Montreal, QC, Canada, June 19-23, 2017},
  pages 1206--1219, 2017.

\bibitem{Atminas16}
A.~Atminas and V.~Zamaraev.
\newblock On forbidden induced subgraphs for unit disk graphs.
\newblock {\em Discrete {\&} Computational Geometry}, 60(1):57--97, 2018.

\bibitem{bang2006}
J.~Bang-Jensen, B.~Reed, M.~Schacht, R.~{\v{S}}{\'a}mal, B.~Toft, and
  U.~Wagner.
\newblock On six problems posed by {Jarik Ne{\v{s}}et{\v{r}}il}.
\newblock {\em Topics in Discrete Mathematics}, pages 613--627, 2006.

\bibitem{Biro17}
C.~Bir{\'{o}}, {\'{E}}.~Bonnet, D.~Marx, T.~Miltzow, and
  P.~Rz{\k{a}}{\.{z}}ewski.
\newblock Fine-grained complexity of coloring unit disks and balls.
\newblock {\em JoCG}, 9(2):47--80, 2018.

\bibitem{Blumer89}
A.~Blumer, A.~Ehrenfeucht, D.~Haussler, and M.~K. Warmuth.
\newblock Learnability and the {V}apnik-{C}hervonenkis dimension.
\newblock {\em J. {ACM}}, 36(4):929--965, 1989.

\bibitem{Bock14}
A.~Bock, Y.~Faenza, C.~Moldenhauer, and A.~J. Ruiz{-}Vargas.
\newblock Solving the stable set problem in terms of the odd cycle packing
  number.
\newblock In {\em 34th International Conference on Foundation of Software
  Technology and Theoretical Computer Science, {FSTTCS} 2014, December 15-17,
  2014, New Delhi, India}, pages 187--198, 2014.

\bibitem{Bonamy18}
M.~Bonamy, E.~Bonnet, N.~Bousquet, P.~Charbit, and S.~Thomass{\'{e}}.
\newblock {EPTAS} for max clique on disks and unit balls.
\newblock In {\em 59th {IEEE} Annual Symposium on Foundations of Computer
  Science, {FOCS} 2018, Paris, France, October 7-9, 2018}, pages 568--579,
  2018.

\bibitem{Bonnet15}
{\'{E}}.~Bonnet, B.~Escoffier, E.~J. Kim, and V.~Th.~Paschos.
\newblock On subexponential and {FPT}-time inapproximability.
\newblock {\em Algorithmica}, 71(3):541--565, 2015.

\bibitem{bonnetetal18}
{\'{E}}.~Bonnet, P.~Giannopoulos, E.~J. Kim, P.~Rz{\k{a}}{\.{z}}ewski, and
  F.~Sikora.
\newblock {QPTAS} and subexponential algorithm for maximum clique on disk
  graphs.
\newblock In {\em 34th International Symposium on Computational Geometry, SoCG
  2018, June 11-14, 2018, Budapest, Hungary}, pages 12:1--12:15, 2018.

\bibitem{Brandstadt1999}
A.~Brandst{\"a}dt, V.~B. Le, and J.~P. Spinrad.
\newblock {\em Graph classes: a survey}.
\newblock SIAM, 1999.

\bibitem{Breu98}
H.~Breu and D.~G. Kirkpatrick.
\newblock Unit disk graph recognition is {NP}-hard.
\newblock {\em Comput. Geom.}, 9(1-2):3--24, 1998.

\bibitem{Brimkov18}
V.~E. Brimkov, K.~Junosza{-}Szaniawski, S.~Kafer, J.~Kratochv{\'{\i}}l,
  M.~Pergel, P.~Rz{\k{a}}{\.{z}}ewski, M.~Szczepankiewicz, and J.~Terhaar.
\newblock Homothetic polygons and beyond: Maximal cliques in intersection
  graphs.
\newblock {\em Discrete Applied Mathematics}, 247:263--277, 2018.

\bibitem{CabelloOpen}
S.~Cabello.
\newblock Maximum clique for disks of two sizes.
\newblock Open problems from {Geometric Intersection Graphs: Problems and
  Directions CG Week Workshop, Eindhoven, June 25, 2015}
  (\url{http://cgweek15.tcs.uj.edu.pl/problems.pdf}), 2015.
\newblock [Online; accessed 07-December-2017].

\bibitem{Cabello2015}
S.~Cabello.
\newblock {Open problems presented at the Algorithmic Graph Theory on the
  Adriatic Coast workshop, Koper, Slovenia}
  (\url{https://conferences.matheo.si/event/6/picture/35.pdf}).
\newblock June 16-19 2015.

\bibitem{CabelloCL13}
S.~Cabello, J.~Cardinal, and S.~Langerman.
\newblock The clique problem in ray intersection graphs.
\newblock {\em Discrete {\&} Computational Geometry}, 50(3):771--783, 2013.

\bibitem{Carmi18}
P.~Carmi, M.~J. Katz, and P.~Morin.
\newblock Stabbing pairwise intersecting disks by four points.
\newblock {\em CoRR}, abs/1812.06907, 2018.

\bibitem{ceroi}
S.~Ceroi.
\newblock {The clique number of unit quasi-disk graphs}.
\newblock Technical Report RR-4419, {INRIA}, Mar. 2002.

\bibitem{Chan03}
T.~M. Chan.
\newblock Polynomial-time approximation schemes for packing and piercing fat
  objects.
\newblock {\em J. Algorithms}, 46(2):178--189, 2003.

\bibitem{ChanH12}
T.~M. Chan and S.~Har{-}Peled.
\newblock Approximation algorithms for maximum independent set of pseudo-disks.
\newblock {\em Discrete {\&} Computational Geometry}, 48(2):373--392, 2012.

\bibitem{Chlebik06}
M.~Chleb{\'{\i}}k and J.~Chleb{\'{\i}}kov{\'{a}}.
\newblock Complexity of approximating bounded variants of optimization
  problems.
\newblock {\em Theor. Comput. Sci.}, 354(3):320--338, 2006.

\bibitem{Clark90}
B.~N. Clark, C.~J. Colbourn, and D.~S. Johnson.
\newblock Unit disk graphs.
\newblock {\em Discrete Mathematics}, 86(1-3):165--177, 1990.

\bibitem{danzer}
L.~Danzer.
\newblock {Zur L{\"o}sung des Gallaischen Problems {\"u}ber Kreisscheiben in
  der Euklidischen Ebene}.
\newblock {\em Studia Sci. Math. Hungar}, 21(1-2):111--134, 1986.

\bibitem{Duan14}
R.~Duan and S.~Pettie.
\newblock Linear-time approximation for maximum weight matching.
\newblock {\em J. {ACM}}, 61(1):1:1--1:23, 2014.

\bibitem{Dvorak20}
Z.~Dvo\v{r}{\'{a}}k and J.~Pek{\'{a}}rek.
\newblock Induced odd cycle packing number, independent sets, and chromatic
  number.
\newblock {\em CoRR}, abs/2001.02411, 2020.

\bibitem{Erlebach05}
T.~Erlebach, K.~Jansen, and E.~Seidel.
\newblock Polynomial-time approximation schemes for geometric intersection
  graphs.
\newblock {\em {SIAM} J. Comput.}, 34(6):1302--1323, 2005.

\bibitem{DBLP:conf/waoa/Fishkin03}
A.~V. Fishkin.
\newblock Disk graphs: {A} short survey.
\newblock In K.~Jansen and R.~Solis{-}Oba, editors, {\em Approximation and
  Online Algorithms, First International Workshop, {WAOA} 2003, Budapest,
  Hungary, September 16-18, 2003, Revised Papers}, volume 2909 of {\em Lecture
  Notes in Computer Science}, pages 260--264. Springer, 2003.

\bibitem{DBLP:journals/jgt/Furedi87}
Z.~F{\"{u}}redi.
\newblock The number of maximal independent sets in connected graphs.
\newblock {\em Journal of Graph Theory}, 11(4):463--470, 1987.

\bibitem{Gibson10}
M.~Gibson and I.~A. Pirwani.
\newblock Algorithms for dominating set in disk graphs: Breaking the
  log\emph{n} barrier - (extended abstract).
\newblock In {\em Algorithms - {ESA} 2010, 18th Annual European Symposium,
  Liverpool, UK, September 6-8, 2010. Proceedings, Part {I}}, pages 243--254,
  2010.

\bibitem{Gyori97}
E.~Györi, A.~V. Kostochka, and T.~Łuczak.
\newblock Graphs without short odd cycles are nearly bipartite.
\newblock {\em Discrete Mathematics}, 163(1):279 -- 284, 1997.

\bibitem{Harpeled18}
S.~Har{-}Peled, H.~Kaplan, W.~Mulzer, L.~Roditty, P.~Seiferth, M.~Sharir, and
  M.~Willert.
\newblock Stabbing pairwise intersecting disks by five points.
\newblock In {\em 29th International Symposium on Algorithms and Computation,
  {ISAAC} 2018, December 16-19, 2018, Jiaoxi, Yilan, Taiwan}, pages
  50:1--50:12, 2018.

\bibitem{HausslerW86}
D.~Haussler and E.~Welzl.
\newblock Epsilon-nets and simplex range queries.
\newblock {\em Discrete \& Computational Geometry}, 2:127--151, 1987.

\bibitem{ImpagliazzoETH}
R.~Impagliazzo, R.~Paturi, and F.~Zane.
\newblock Which problems have strongly exponential complexity?
\newblock {\em Journal of Computer and System Sciences}, 63(4):512--530, Dec.
  2001.

\bibitem{Kang12}
R.~J. Kang and T.~M{\"{u}}ller.
\newblock Sphere and dot product representations of graphs.
\newblock {\em Discrete {\&} Computational Geometry}, 47(3):548--568, 2012.

\bibitem{DBLP:conf/stoc/KawarabayashiR10}
K.~Kawarabayashi and B.~A. Reed.
\newblock Odd cycle packing.
\newblock In {\em Proceedings of the 42nd {ACM} Symposium on Theory of
  Computing, {STOC} 2010, Cambridge, Massachusetts, USA, 5-8 June 2010}, pages
  695--704, 2010.

\bibitem{Keller17}
C.~Keller, S.~Smorodinsky, and G.~Tardos.
\newblock {On Max-Clique for intersection graphs of sets and the
  Hadwiger-Debrunner numbers}.
\newblock In {\em Proceedings of the Twenty-Eighth Annual {ACM-SIAM} Symposium
  on Discrete Algorithms, {SODA} 2017, Barcelona, Spain, Hotel Porta Fira,
  January 16-19}, pages 2254--2263, 2017.

\bibitem{koebe}
P.~Koebe.
\newblock {Kontaktprobleme der konformen Abbildung}.
\newblock {\em Berichte \"{u}ber die Verhandlungen der S\"{a}chsischen Akademie
  der Wissenschaften zu Leipzig, Mathematisch-Physikalische Klasse},
  88:141--164, 1936.

\bibitem{precoloring}
J.~Kratochvil.
\newblock Precoloring extension with fixed color bound.
\newblock {\em Acta Math. Univ. Comen}, 62:139--153, 1993.

\bibitem{DBLP:conf/gd/Kratochvil96}
J.~Kratochv{\'{\i}}l.
\newblock Intersection graphs of noncrossing arc-connected sets in the plane.
\newblock In {\em Graph Drawing, Symposium on Graph Drawing, {GD} '96,
  Berkeley, California, USA, September 18-20, Proceedings}, pages 257--270,
  1996.

\bibitem{DBLP:journals/jct/KratochvilM94}
J.~Kratochv{\'{\i}}l and J.~Matou{\v{s}}ek.
\newblock Intersection graphs of segments.
\newblock {\em J. Comb. Theory, Ser. {B}}, 62(2):289--315, 1994.

\bibitem{MalesiskaPW98}
E.~Malesinska, S.~Piskorz, and G.~Wei{\ss}enfels.
\newblock On the chromatic number of disk graphs.
\newblock {\em Networks}, 32(1):13--22, 1998.

\bibitem{Marx15}
D.~Marx and M.~Pilipczuk.
\newblock Optimal parameterized algorithms for planar facility location
  problems using voronoi diagrams.
\newblock In {\em Algorithms - {ESA} 2015 - 23rd Annual European Symposium,
  Patras, Greece, September 14-16, 2015, Proceedings}, pages 865--877, 2015.

\bibitem{McKee1999}
T.~A. McKee and F.~R. McMorris.
\newblock {\em Topics in intersection graph theory}.
\newblock SIAM, 1999.

\bibitem{Middendorf92}
M.~Middendorf and F.~Pfeiffer.
\newblock The max clique problem in classes of string-graphs.
\newblock {\em Discrete Mathematics}, 108(1-3):365--372, 1992.

\bibitem{Moshkovitz10}
D.~Moshkovitz and R.~Raz.
\newblock Two-query {PCP} with subconstant error.
\newblock {\em J. {ACM}}, 57(5):29:1--29:29, 2010.

\bibitem{Nieberg05}
T.~Nieberg and J.~Hurink.
\newblock A {PTAS} for the minimum dominating set problem in unit disk graphs.
\newblock In {\em Approximation and Online Algorithms, Third International
  Workshop, {WAOA} 2005, Palma de Mallorca, Spain, October 6-7, 2005, Revised
  Papers}, pages 296--306, 2005.

\bibitem{Nieberg04}
T.~Nieberg, J.~Hurink, and W.~Kern.
\newblock A robust {PTAS} for maximum weight independent sets in unit disk
  graphs.
\newblock In {\em Graph-Theoretic Concepts in Computer Science, 30th
  International Workshop, WG 2004, Bad Honnef, Germany, June 21-23, 2004,
  Revised Papers}, pages 214--221, 2004.

\bibitem{Raghavan03}
V.~Raghavan and J.~P. Spinrad.
\newblock {Robust algorithms for restricted domains}.
\newblock {\em J. Algorithms}, 48(1):160--172, 2003.

\bibitem{SmithWormald}
W.~D. Smith and N.~C. Wormald.
\newblock Geometric separator theorems {\&} applications.
\newblock In {\em 39th Annual Symposium on Foundations of Computer Science,
  {FOCS} '98, November 8-11, 1998, Palo Alto, California, {USA}}, pages
  232--243, 1998.

\bibitem{stacho}
L.~Stach{\'o}.
\newblock A solution of {G}allai’s problem on pinning down circles.
\newblock {\em Mat. Lapok}, 32(1-3):19--47, 1981.

\bibitem{Tait1877}
P.~G. Tait.
\newblock Some elementary properties of closed plane curves.
\newblock {\em Messenger of Mathematics, New Series}, 69:270--272, 1877.

\bibitem{Leeuwen06}
E.~J. van Leeuwen.
\newblock Better approximation schemes for disk graphs.
\newblock In {\em Algorithm Theory - {SWAT} 2006, 10th ScandinavianWorkshop on
  Algorithm Theory, Riga, Latvia, July 6-8, 2006, Proceedings}, pages 316--327,
  2006.

\bibitem{EJvL2009}
E.~J. van Leeuwen.
\newblock {\em Optimization and Approximation on Systems of Geometric Objects}.
\newblock PhD thesis, Utrecht University, 2009.

\bibitem{vapnik}
V.~N. Vapnik and A.~Y. Chervonenkis.
\newblock On the uniform convergence of relative frequencies of events to their
  probabilities.
\newblock In {\em Measures of complexity}, pages 11--30. Springer, 2015.

\end{thebibliography}

\end{document}